\documentclass[conference,compsoc]{IEEEtran}
\IEEEoverridecommandlockouts

\usepackage{cite}
\usepackage[implicit=false]{hyperref}
\usepackage{amsmath,amssymb,amsfonts}
\usepackage[dvipsnames]{xcolor}
\usepackage{url} 
\usepackage{algorithmic}
\usepackage{graphicx}
\usepackage{textcomp}
\usepackage{xcolor}
\usepackage{tikz}
\usepackage{enumitem}
\usepackage{amsmath}
\usepackage{multirow}
\usepackage{algorithm}
\usepackage{algorithmic}
\usepackage{enumitem}
\usepackage{amsmath,amssymb,amsfonts}
\usepackage{graphicx}
\usepackage[english]{babel}
\usepackage{amsthm}
\usepackage{amssymb}
\usepackage{multirow}
\usepackage{graphicx}
\usepackage{xcolor,color,colortbl}
\usepackage{enumitem}
\usepackage{amsmath,amssymb,amsfonts}
\usepackage{theoremref}

\newtheorem{thm}{Theorem}
\newtheorem{prop}{Proposition}

\newtheorem{cor}{Corollary}[thm]
\usepackage{soul}

\def\BibTeX{{\rm B\kern-.05em{\sc i\kern-.025em b}\kern-.08em
T\kern-.1667em\lower.7ex\hbox{E}\kern-.125emX}}

\begin{document}
\title{%
\vspace{-1.3cm}
{\fontsize{12}{18}\selectfont   \normalfont To Appear in the 45th IEEE Symposium on Security and Privacy, May 20-23, 2024.}\\[3ex]
Multi-Instance Adversarial Attack on GNN-Based Malicious Domain Detection}
\author{\IEEEauthorblockN{Mahmoud Nazzal\IEEEauthorrefmark{1},
Issa Khalil\IEEEauthorrefmark{4},
Abdallah Khreishah\IEEEauthorrefmark{1}, 
NhatHai Phan\IEEEauthorrefmark{1}, and
Yao Ma\IEEEauthorrefmark{1}}
\IEEEauthorblockA{\IEEEauthorrefmark{1}New Jersey Institute of Technology,
Newark, NJ 07102, USA\\}
\IEEEauthorblockA{\IEEEauthorrefmark{4}Qatar Computing Research Institute (QCRI), Hamad Bin Khalifa University (HBKU), Doha, Qatar\\
Email: mn69@njit.edu, ikhalil@hbku.edu.qa,\{abdallah,phan,yao.ma\}@njit.edu 
}
}

\maketitle

\begin{abstract}
Malicious domain detection (MDD) is an open security challenge that aims to detect if an Internet domain name is associated with cyber attacks. Many techniques have been applied to tackle this problem, among which graph neural networks (GNNs) are deemed one of the most effective approaches. GNN-based MDD employs domain name system (DNS) logs to represent Internet domains as nodes in a graph, dubbed domain maliciousness graph (DMG) and trains a GNN model to infer the maliciousness of Internet domains by leveraging the maliciousness of already identified ones. As this method heavily relies on the ``publicly'' accessible DNS logs to build DMGs, it creates a vulnerability for adversaries to manipulate the features and edges of their domain nodes within these graphs.
The current body of literature primarily focuses on threat models that involve manipulating individual adversary (attacker) nodes. Nonetheless, adversaries usually create numerous domains to accomplish their attack objectives, aiming to reduce costs and evade detection. Hence, they aim to remain undetected across as many domains as possible. In this work, we call the attack that manipulates several nodes in the DMG concurrently \textit{a multi-instance evasion attack}. To the best of our knowledge, this type of attack has not been explored in the prior art. 
We present both theoretical and empirical evidence to show that the existing single-instance evasion techniques for GNN-based MDDs are inadequate to launch multi-instance evasion attacks. Therefore, we propose an inference-time, multi-instance adversarial attack, dubbed MintA, against GNN-based MDD. MintA optimizes node perturbations to enhance the evasiveness of a node and its neighborhood. MintA only requires black-box access to the target model to launch the attack successfully. In other words, MintA does not require any knowledge of the MDD model's parameters, architecture, or information on non-adversary nodes. We formulate an optimization problem that satisfies the attack objectives of MintA and devise an approximate solution for it. We evaluate MintA on a state-of-the-art GNN-based MDD technique using real-world data, and our experiments demonstrate an attack success rate of over 80\%.
The findings of this study serve as a cautionary note for security experts, highlighting the vulnerability of GNN-based MDD to practical attacks that can impede the effectiveness and advantages of this approach.
\end{abstract}


\section{Introduction}
\label{Section1}
\par Internet domains form a foundation for adversaries to launch various cyber attacks. For example, adversaries can use Internet domains to disseminate malware \cite{abraham2010overview}, facilitate command and control (C\&C) communications \cite{zeidanloo2009botnet}, and host scam, phishing, and brand squatting web pages \cite{nabeel2022brand}. Malicious domain detection (MDD) refers to the problem of deciding whether a given Internet domain is used to launch malicious activities. 
Considerable efforts have been made in developing various MDD methodologies \cite{zhauniarovich2018survey}. Among the different MDD methods, those that utilize domain name system (DNS) data possess unique advantages, such as scalability \cite{khalil2017killing}, rich domain features, and public accessibility \cite{zhauniarovich2018survey}. Existing DNS-based MDD methods can be broadly categorized into classification-based and inference-based approaches \cite{nabeel2020following,li2021dydom,wang2022handom}. Classification-based MDD relies on local domain features \cite{antonakakis2010building,bilge2014exposure}. In contrast, inference-based MDD integrates relations among domains and their local properties to achieve timely and accurate inference of domain maliciousness \cite{manadhata2014detecting,khalil2016discovering,khalil2017killing,nabeel2020following,khalil2018domain,sun2019hindom,li2021dydom,zhang2021attributed,li2022heterogeneous}.

\par In this paper, we focus on the inference-based MDD approach given its superior performance compared with other methods \cite{sun2019hindom,li2021dydom,zhang2021attributed,li2022heterogeneous}. Inference-based MDD is based on the \textit{guilt-by-association} principle \cite{khalil2018domain,khalil2020method,xia2021identifying}, i.e., if a domain is connected to a group of known malicious domains, then it is likely to be malicious as well. The inference-based MDD process comprises two primary phases. The initial phase involves creating a domain association graph known as the ``domain maliciousness graph (DMG).'' DNS data presents a valuable resource for entities engaged in MDD\footnote{Let us refer to a party performing MDD as an \textit{MDD entity}.} to construct the DMG. The second phase involves utilizing an inference technique to identify the maliciousness of unknown domains in the DMG by using a small set of labeled domains as a reference. A DMG (e.g., Fig. \ref{system_model}) can be either a homogeneous graph with domain nodes and their relationships or a heterogeneous graph with multiple node and link types. For example, a DMG may consist of domain and client nodes \cite{rahbarinia2015segugio, manadhata2014detecting,lee2014gmad}, domain and IP address nodes \cite{khalil2016discovering,khalil2018domain}, or domain, IP address, and client nodes \cite{sun2019hindom,li2021dydom,zhang2021attributed,li2022heterogeneous}. 

\par Regarding the inference phase of MDD, techniques like belief propagation, graph neural networks (GNNs), or heterogeneous graph neural networks (hetGNNs) are typically employed. Belief propagation is a commonly utilized inference method. However, its accuracy drops significantly when the labeled training data is insufficient, and it cannot make inferences regarding isolated domain nodes \cite{satorras2021neural}. Recent research has explored more advanced inference methods using GNNs and hetGNNs \cite{sun2019hindom,sun2020hgdom,nabeel2021method,li2021dydom,zhang2021attributed,li2022heterogeneous}. GNNs have been shown to offer significant advantages over belief propagation. Firstly, GNNs employ data-driven training to learn mechanisms for message passing and aggregation across graph nodes in an end-to-end fashion \cite{zhang2020factor}. Secondly, GNNs are more adaptable to larger graphs with complex structures not initially seen \cite{yoon2019inference}. Moreover, GNNs demonstrate superior performance when the labeled training data is limited \cite{sun2019hindom,li2021dydom,zhang2021attributed,li2022heterogeneous}.

\par Despite the effectiveness of GNN-based models in accurately and quickly detecting malicious domains and extensive research on their vulnerability to adversarial attacks \cite{dai2018adversarial,zugner2018adversarial,zang2020graph}, multi-instance attacks have not been fully explored in the context of GNN-based MDD. In this scenario, an adversary (attacker) can only manipulate the features and relationships of the domains under its control (e.g., the orange subgraph in Fig. \ref{system_model}) to craft an adversarial attack to bypass the MDD model. This is possible due to the heavy reliance of MDD models on publicly available DNS logs to construct their DMGs, leaving room for malicious intervention.

\par \textbf{Prior approaches and their limitations.} The current body of literature primarily focuses on threat models that involve manipulating individual adversary nodes. However, adversaries usually create many domains to accomplish their attack objectives, aiming to reduce costs and evade detection. Hence, they aim to remain undetected across as many domains as possible. Existing approaches either target a specific node for evasion (e.g., \cite{dai2018adversarial,chen2018fast,wu2019adversarial,zugner2018adversarial,chang2020restricted,wang2020evasion})  or attack the model's overall performance through untargeted (availability) attacks (e.g., \cite{xu2019topology,ma2019attacking,ma2020towards,ma2022adversarial}). We demonstrate that repeatedly applying single-instance attacks to achieve multi-instance evasion is suboptimal and degenerates the attack's overall performance. On the other hand, availability attacks are easy to detect and do not achieve the evasion goal of the adversary.

\par \textbf{Overview of our approach.} Based on the previous discussion, we have formulated the adversarial attack in GNN-based MDD as a multi-instance evasion attack. To execute this attack, we design MintA as a multi-instance attack algorithm. MintA first constructs a surrogate model using only black-box access to the target model. It is worth noting that MintA is practical as it does not require knowledge of the target model's parameters, architecture, or other nodes in the DMG. After constructing the surrogate model, MintA employs it to identify the suitable feature and edge perturbations that will collectively evade the adversary nodes while minimizing the effect on other nodes in the DMG. This is achieved through a two-objective optimization problem at each adversary node. The first objective aims to maximize the model's loss at each adversary node to avoid detection, while the second objective aims to maximize the model's loss at the neighboring nodes to evade detection. To obtain an approximate solution, we maximize the weighted average of these two objectives. Once the appropriate edge/feature perturbations are identified, MintA implements them by making name and domain resolution edits to its nodes. This is achieved through domain registration services and IP resolution edits. 

\par \textbf{Summary of contributions.} Our contribution in this study includes: (i) Demonstrating the vulnerability of GNN-based MDD and proposing MintA, a novel algorithm for the multi-instance adversarial attack that only requires black-box access to the target model and with no knowledge of the DMG graph beyond the adversary nodes. To the best of our knowledge, this is the first practical attack against GNN-based MDD. (ii) We present a method for implementing the optimized adversarial perturbations by simple domain name and IP-resolution manipulations. (iii) We conduct extensive experiments on real-world data to evaluate the effectiveness of our proposed attack. Our results show an attack success rate of over 80\%. Also, MintA is shown to bypass outlier detection and graph purification-based defenses. 

\par \textbf{Paper outline.} Section \ref{Section2} presents relevant preliminaries. The threat model of adversarial attack in MDD is presented in Section \ref{Section3}. Section \ref{Section4} presents the proposed MintA algorithm. Experiments and results are presented in Section \ref{Section5}. Section \ref{Section6} summarizes related work. We present a discussion in Section \ref{Section7} with the conclusions in Section \ref{Section8}.

\begin{figure}[t]
\centering
\resizebox{0.99\columnwidth}{!}{
\includegraphics[width=12cm]{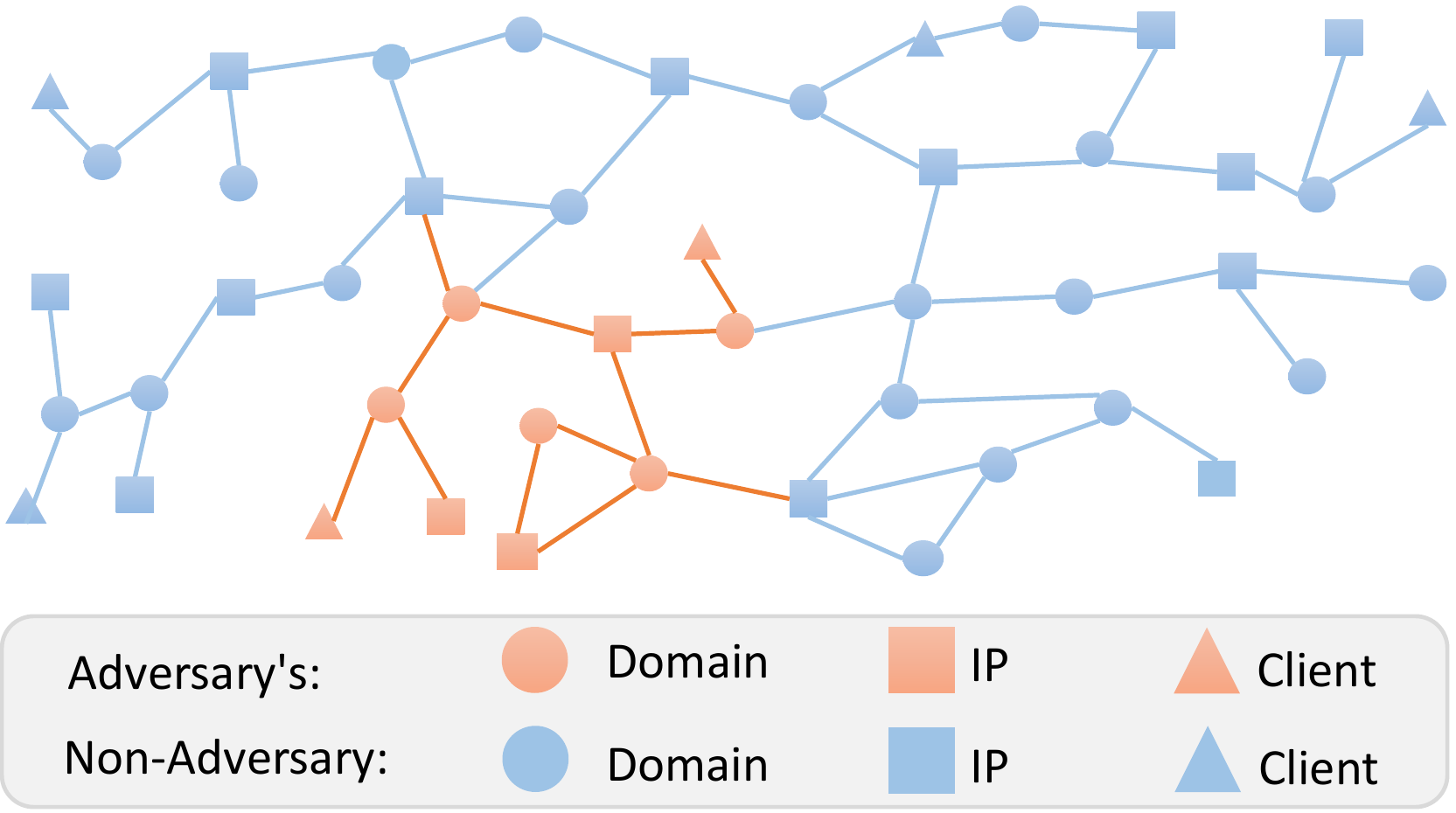}}
\caption{Based on DNS logs, the MDD entity constructs a DMG containing the adversary's subgraph (in orange) along with non-adversary subgraphs (in blue).}
\label{system_model}
\end{figure}
\section{Background}
\label{Section2}
 
\par \textbf{Graph Neural Networks.} Graph neural networks (GNNs) \cite{welling2016semi} extend deep learning to graph data by using neural network layers for handling messages passing across graph edges and their aggregation. A GNN transforms graph information into a set of low-dimensional node embeddings calculated based on local node features and graph topology. GNNs have achieved state-of-the-art performances in a wide set of applications ranging from fraud detection to drug discovery \cite{wu2020comprehensive}. A heterogeneous graph or a heterogeneous information network (HIN) is a graph with multiple node and/or edge types. Heterogeneous GNNs (hetGNNs) extend GNNs to HINs. Some hetGNNs project the HIN on a graph space to eliminate heterogeneity \cite{hu2020heterogeneous}, while others decompose a given HIN into meta-paths\footnote{A meta-path is a composition of relations linking two nodes.} which are then encoded to get node representations \cite{wang2019heterogeneous,fu2020magnn}. In a variety of applications, hetGNNs achieve state-of-the-art performances \cite{liu2018heterogeneous,li2022sybilflyover,jiang2021mafi,hu2020gfd}.

\par \textbf{GNN-based MDD inference.} The state-of-the-art MDD approaches employ heterogeneous DMGs \cite{sun2019hindom,li2021dydom,li2022heterogeneous,sun2020deepdom,zhang2021attributed}. Still, they differ in the types of nodes and edges assumed in their DMGs and the corresponding ways of achieving message passing and aggregation. First, HinDom \cite{sun2019hindom} uses meta-paths with HINs of clients, domains, and IP addresses excluding node attributes. Subsequently, this setting is extended to attributed HINs. Node attributes are obtained as character-level domain properties in HGDom \cite{sun2020hgdom} and the 21 FANCI domain properties \cite{schuppen2018fanci} in DeepDom \cite{sun2020deepdom}. Similarly, Zhang et al. \cite{zhang2021attributed} define a HIN of domains, IP addresses, and clients, where node type-aware feature transformation and edge type-aware message aggregation are employed. More recent works such as \cite{li2022heterogeneous} and \cite{wang2022handom} adopt attention mechanisms to achieve message passing and aggregation. These approaches show the potential of GNNs, particularly hetGNNs, in achieving timely and accurate MDD.

\par \textbf{From DNS logs to DMG.} DNS is a key resource for constructing DMGs. As an example, Fig. \ref{network_schema}(a) shows the network schema of a heterogeneous DMG. Domain, IP address, and client node types are included where domain features such as the 21 FANCI features \cite{schuppen2018fanci} serve as the domain node attributes. Also, the following edge types exist. 
\begin{itemize}[leftmargin=*]
 \item domain-\textit{query}-client: linking a client node to a domain node it queries
 \item domain-\textit{apex}-domain: linking two domain nodes if they share the same apex domain
 \item domain-\textit{resolve}-IP: linking a domain node to an IP node it resolves to
 \item domain-\textit{similar}-domain: linking two domain nodes if the n-gram character-level similarity between their names exceeds a preset threshold of 0.8 
 \cite{kumarasinghe2022pdns} 
\end{itemize}

\par Let us consider an example of using DNS to construct a heterogeneous DMG according to the above network schema. Consider domain names \textit{``www.b.rwth-aachen.de''}, \textit{``writes.bnxd.rwth-aachen.de''}, and \textit{``dekh1her76avy0qnelivijwd1.ddns.net''} represented by domain-type nodes $v_1$, $v_2$, and $v_3$, respectively, as shown in Fig. \ref{network_schema}(b). Let us assume further that a given DNS log shows that the domains $v_1$, $v_2$, and $v_3$ resolve to IP addresses $i_1$, $i_2$, and $i_1$, respectively, and they are queried by clients $c_1$, $c_2$, and $c_3$, respectively. Recalling the four edge types mentioned above, the following edges exist. First, $v_1$ and $v_3$ are linked to $i_1$ whereas $v_2$ is linked to $i_2$, all with \textit{resolve} edges. Second, $v_1$ through $v_3$ are linked to $c_1$ through $c_3$ with \textit{query} edges, respectively. Third, $v_1$ is connected to $v_2$ with a \textit{similar} edge due to their high character-level similarity and with an \textit{apex} edge since they share the same apex address (\textit{``rwth-aachen.de''}). Therefore, there is a direct mapping between the contents of DNS logs and the shape of the constructed DMG.



\begin{figure}[!htb]
\centering
\resizebox{0.99\columnwidth}{!}{
\begin{tabular}{cc}
\includegraphics[width=8cm]{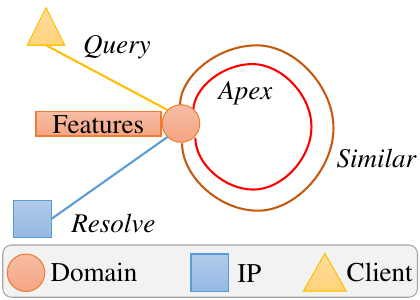}&
\includegraphics[width=6cm]{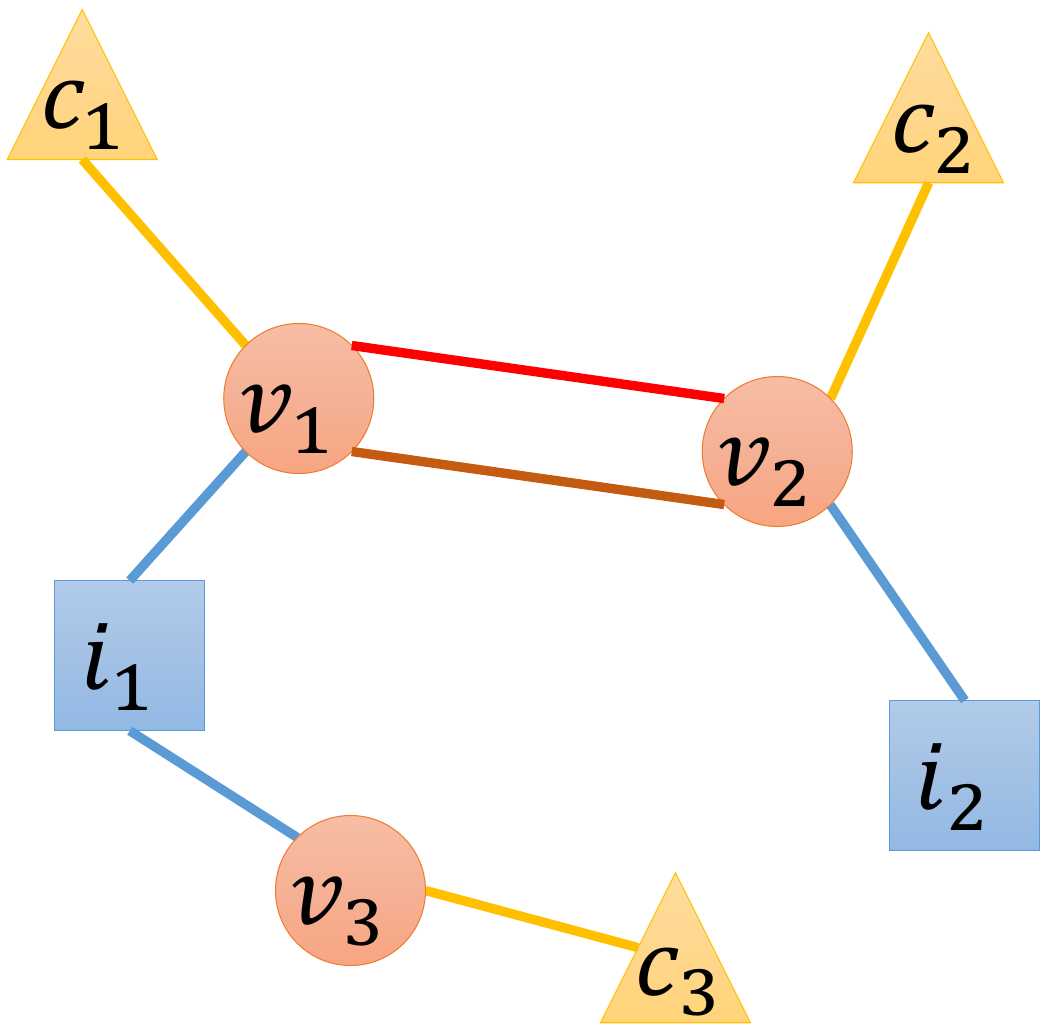}
\\
\Large{(a)}&
\Large{(b)}
\end{tabular}}
\caption{(a): Network schema of a heterogeneous DMG showing node and edge types. (b): An example DMG formed by associations extracted from DNS logs.}
\label{network_schema} 
\end{figure}

\section{Threat Model}
\label{Section3}

\par We characterize the threat model in terms of the adversary’s goals, knowledge, and capabilities. The goal of the adversary is to evade the detection of its nodes by the MDD model with budgeted perturbations. As for knowledge, the adversary has partial knowledge of its actual subgraph residing in the DMG (the orange-colored subgraph in Fig.~\ref{system_model}). Namely, the adversary knows its nodes (domains and IP addresses) but does not know the topology or the node features in its actual subgraph in the DMG. This information requires knowledge of the edge types, and node attributes the target MDD entity assumes. However, the adversary can still assume the usage of commonly used edge types and node attributes to obtain an estimate of its actual subgraph. 

While different GNN-based MDD models may have different node and edge structures in their DMGs, they often utilize common edge types, such as domain-\textit{similar}-domain or domain-\textit{resolve}-IP, that are derived from DNS logs (e.g., \cite{sun2019hindom,sun2020deepdom,sun2020hgdom,zhang2021attributed,li2022heterogeneous,wang2022handom}). However, the surrogate model employed to optimize the perturbations is a simplified homogeneous model that does not differentiate between edge types. Additionally, common node features, like whether a domain name contains digits or has a ``www'' prefix, are utilized by multiple works in the field \cite{sun2020deepdom,sun2020hgdom, zhang2021attributed,li2022heterogeneous, wang2022handom}. Thus, an adversary can assume the existence of these commonly used edges and features, even without precise knowledge of the DMG structure. In our problem formulation, we assume the adversary knows only an estimate of its subgraph. We hypothesize that perfectly knowing the remaining portions of the DMG and incorporating them into the solution may improve the attack's success, but the improvement would not be significant. This is because adversary nodes exhibit denser mutual connections among themselves compared to their connections with non-adversary nodes, as established in the following section.
Moreover, the adversary does not know the target model's architecture, parameters, or training mechanism. It can only query the target model for a certain number of domain nodes to obtain labeled data to train a surrogate model. Thus, this is a black box attack. As for capabilities, the adversary can only manipulate specific editable node features and edges in its subgraph by modifying the names of its domains and their IP resolution relationships, as shown in the next section. 

\par A general overview of the proposed attack setting is presented in Fig. \ref{attack_crafting}. The adversary queries the target MDD model to obtain labeled data and trains a surrogate model, as shown in Fig. \ref{attack_crafting}(a). Then, the adversary constructs an estimate of its subgraph in the DMG assuming the existence of commonly used edge types and node features. With this subgraph estimate, the adversary uses the surrogate model along with the proposed MintA (detailed in the next section) to obtain optimized feature (and/or edge) perturbations to the nodes in its subgraph, as shown in Fig. \ref{attack_crafting}(b). Next, the adversary implements these perturbations by manipulating its domains' names and IP resolutions, as shown in Section \ref{Section4}.1. The modifications made by the adversary will be observable in the DNS logs, and eventually, the desired subgraph will emerge in the DMG constructed by the MDD entity.

\begin{table}[!h]
\centering
\caption{Flipping a feature vector by a simple name edit.}
\begin{tabular}{|l|c|c|}
\hline
Feature & \textit{before} & \textit{after} \\\hline
\textit{Has a WWW prefix} & 1 & 0 \\\hline
\textit{Contains single-character subdomains} & 1 & 0 \\\hline
\textit{Is exclusive prefix repetition} & 0 & 1 \\\hline
\textit{Contains digits} & 0 & 1 \\\hline 
\end{tabular}
\label{table1}
\end{table}


\begin{figure}[!t]
\centering
\resizebox{0.999\columnwidth}{!}{
\begin{tabular}{c}
\includegraphics[width=25cm]{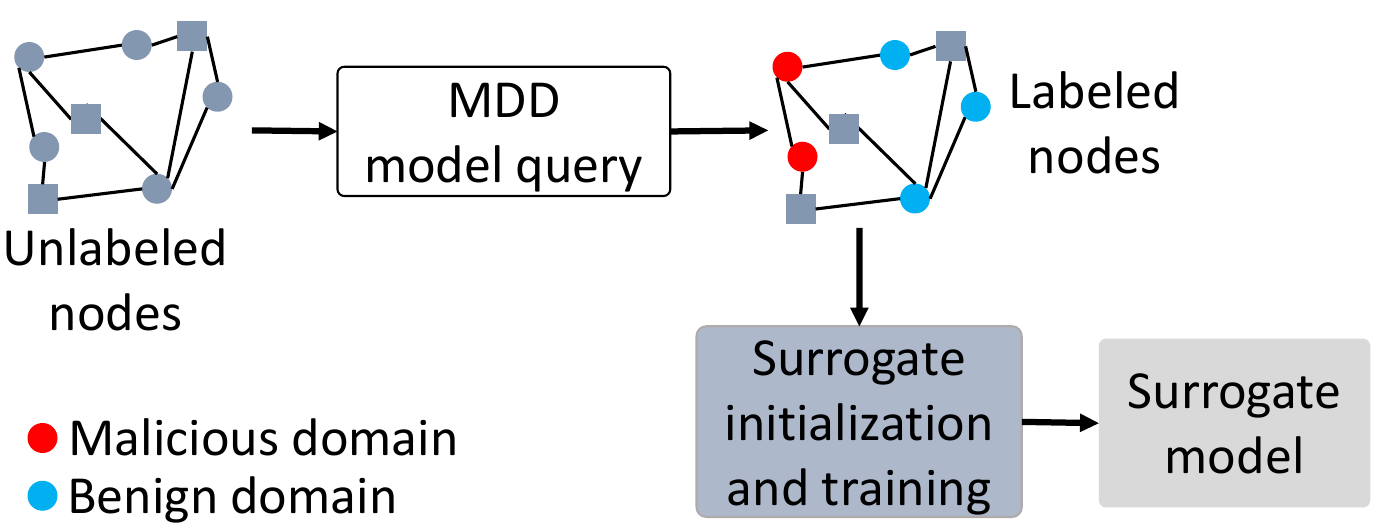}
\\
\Huge{(a)}
\\
\includegraphics[width=25cm]{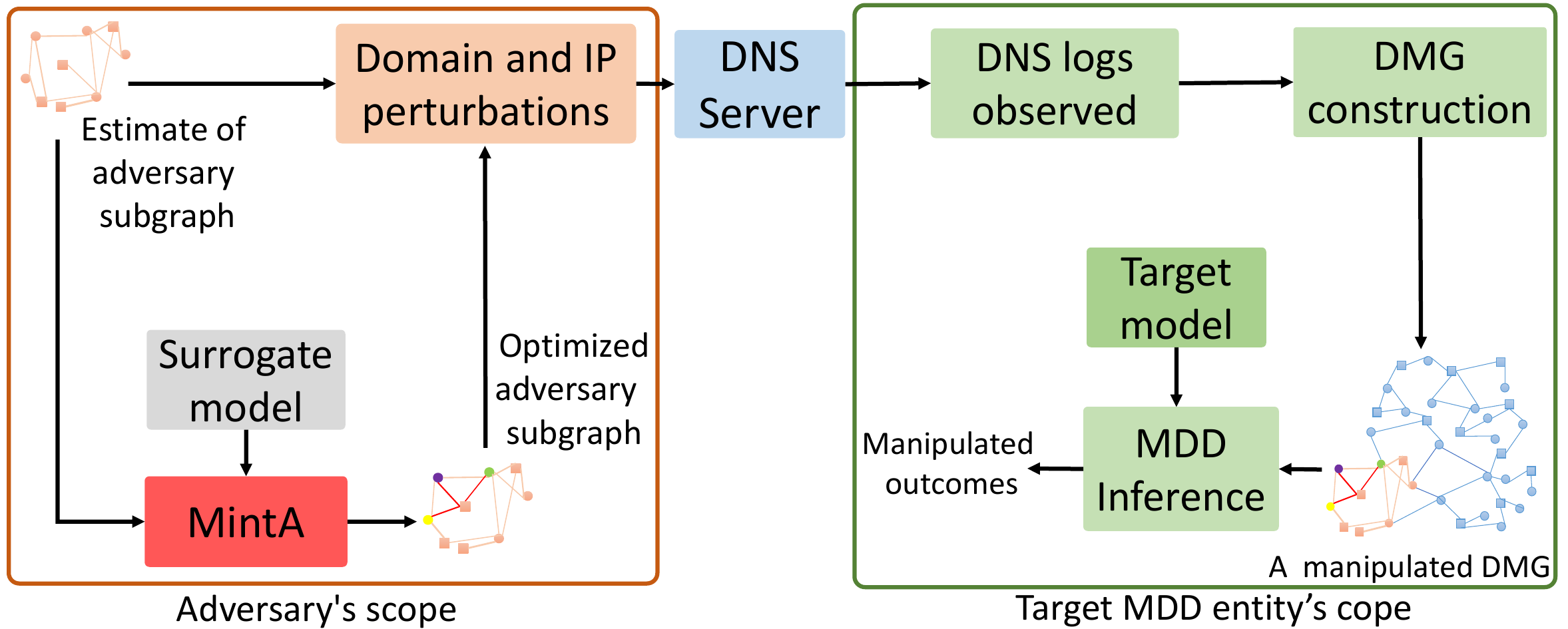}
\\
\Huge{(b)}
\end{tabular}}
\caption{An overview of the attack; surrogate model training in (a), and crafting the attack in (b).}
\label{attack_crafting} 
\end{figure}

\begin{figure}[!t]
\centering
\resizebox{0.95\columnwidth}{!}{
\begin{tabular}{cc}
\includegraphics[width=12cm]{figs/graph_perturb_exmaple1.pdf}& \includegraphics[width=12cm]{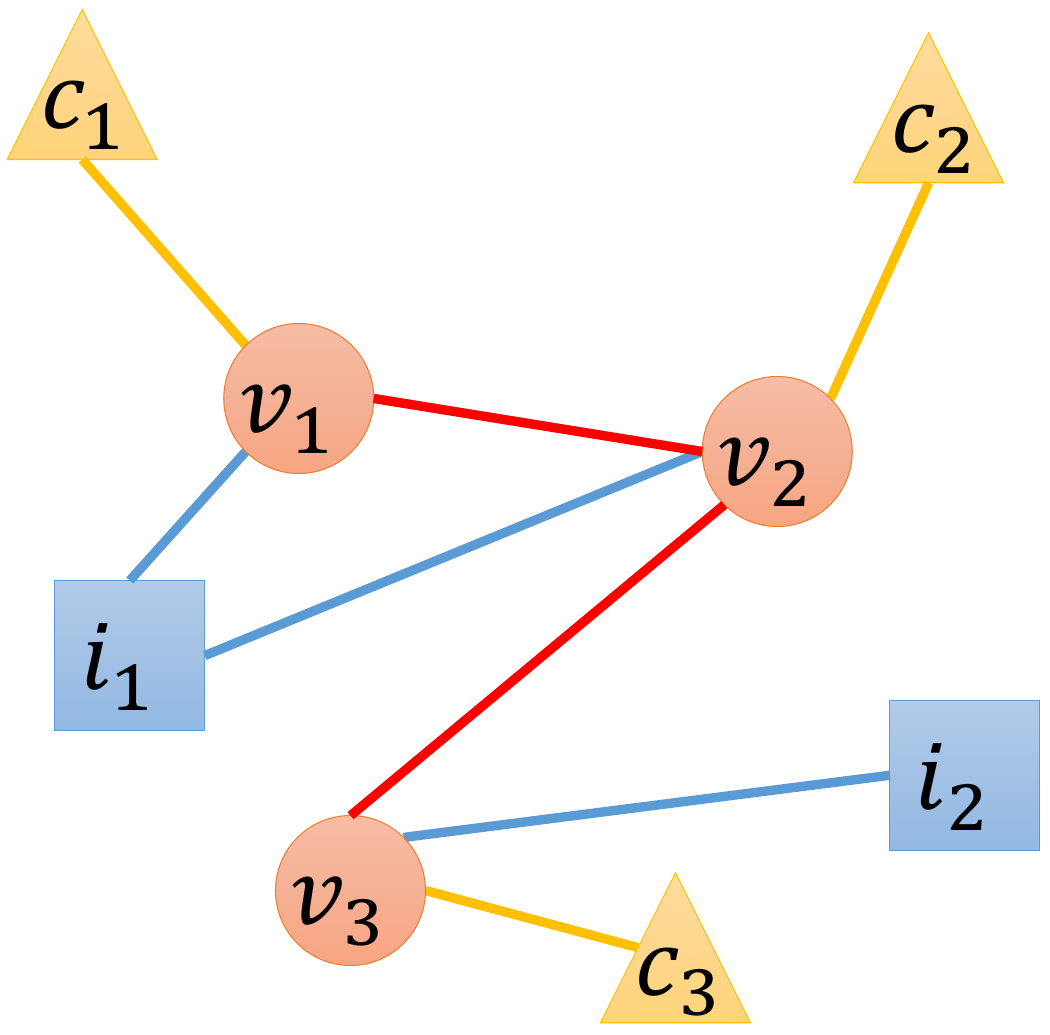}\\
\Huge{(a)} & \Huge{(b)}
\end{tabular}}
\caption {An example DMG constructed based on DNS in (a) and the required perturbed version in (b).}
\label{graph_perturb_exmaple} 
\end{figure}

\section{The Proposed Attack}
\label{Section4}

\par In this section, we first present how an adversary can implement intentional perturbations in a target DMG through DNS. Then, we analyze the conditions for an adversary to craft a successful adversarial attack. Next, we present MintA as an algorithm for optimizing the perturbations applied to the adversary's nodes.

\subsection{Practical implementation of feature and edge perturbations}

\par An Internet domain encompasses numerous features. However, altering some of these attributes can prove challenging, if not impossible. For example, changing the top-level domain (TLD) of a given domain from ``.com'' to ``.edu'' is not possible without institutional accreditation which is required for registering a domain name with this TLD. Among the 21 FANCI features, we select the following features as editable: \textit{Domain name length}, \textit{Has a www prefix}, \textit{Contains a single-character subdomain}, \textit{Is exclusive prefix repetition}, and \textit{Contains digits}. As for edges, we select the domain-\textit{resolve}-IP, domain-\textit{apex}-domain, and domain-\textit{similar}-domain edges \cite{sun2019hindom,sun2020deepdom,sun2020hgdom,zhang2021attributed,li2022heterogeneous} as editable.

\par To illustrate how an adversary can perturb edges and features in its subgraph on a DMG by domain name and IP resolution manipulations, let us recall the example domain names \textit{``www.b.rwth-aachen.de''}, \textit{``writes.bnxd.rwth-aachen.de''}, and \textit{``dekh1her76avy0qnelivijwd1.ddns.net''} and their respective node representations, $v_1$, $v_2$, and $v_3$ appearing in Fig. \ref{graph_perturb_exmaple}(a). Next, let us consider feature perturbation on $v_1$ as an example with binary features for simplicity. For $v_1$, a 4-bit binary representation of the features \textit{Has a www prefix}, \textit{Contains a single-character subdomain}, \textit{Is exclusive prefix repetition}, and \textit{Contains digits} is [1100] since $v_1$ has a www prefix (1), contains a single-character subdomain (1), does not have an exclusive prefix repetition (0), and does not contain digits (0). The domain name owner can flip the feature vector to be [0011] by changing its name to \textit{``bnxd3.bnxd3.rwth-aachen.cis''}\footnote{This modified name has no www prefix (0), does not contain a single-character subdomain (0), has an exclusive prefix repetition (1), and contains digits (1), thus the feature vector is [0011].}. Table \ref{table1} lists the feature vectors of this domain before and after the modification.

\par Assume that the adversary's goal is to modify the topology of the constructed graph shown in Fig. \ref{graph_perturb_exmaple}(a) into the one in Fig. \ref{graph_perturb_exmaple}(b). This can be implemented as follows.
\begin{itemize}[leftmargin=*]
 \item Swap the \textit{resolve} edges between $v_2$ and $v_3$: by changing the resolution of domain $v_2$ to IP $i_1$, and that of $v_3$ to IP $i_2$.
 \item Drop the \textit{similar} edge between $v_1$ and $v_2$: by changing node $v_1$'s domain name from \textit{``www.b.rwth-aachen.de''} to \textit{``www.b.sokj-bbchin.de''} thereby dropping the character-level similarity between $v_1$ and $v_2$ to below 0.8.
 \item Add an \textit{apex} edge between $v_2$ and $v_3$: by changing node $v_3$'s domain name from \textit{``dekh1her76avy0qnelivijwd1.ddns.net''} to \textit{``dekh1her76avy0qnelivijwd1.ddns.rwth-aachen.de''} so that both share the same ``rwth-aachen.de'' apex address.
\end{itemize}

\begin{figure}[!bht]
\centering
\resizebox{0.9555\columnwidth}{!}{
\includegraphics{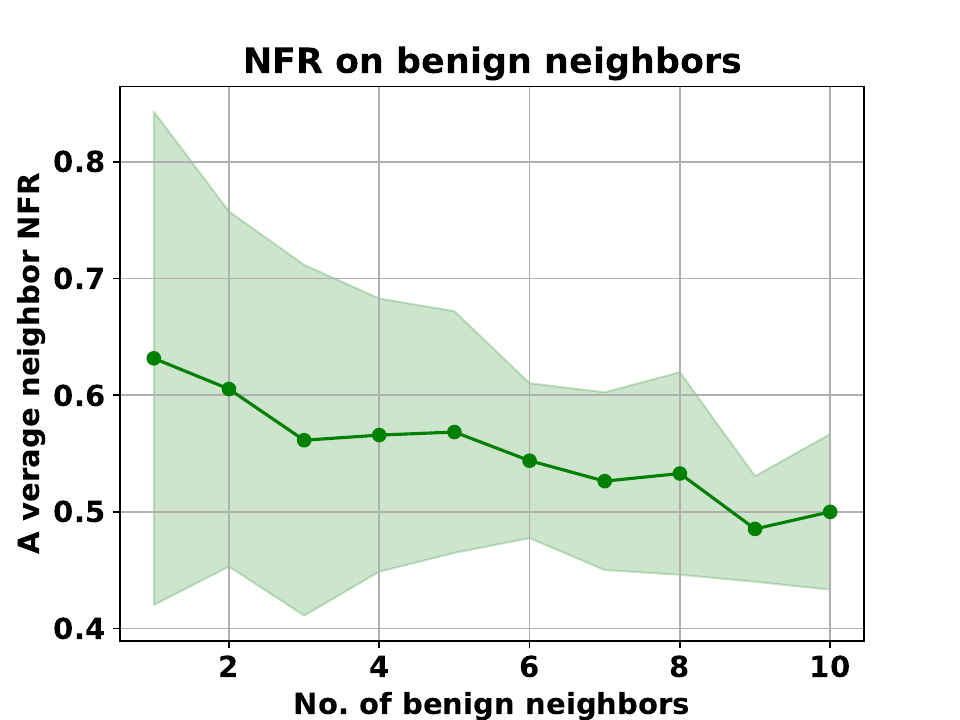}
}
\caption{The impact of evading mutually-isolated adversary nodes on their benign neighbors (NFR: negative flip rate).}
\label{essential_motivation}
\end{figure}

\subsection{Multi-instance adversarial attack}
\par In the following, we identify the requirements for an adversary to perform a successful adversarial attack against GNN-based MDD models.

\par \textbf{1. The adversary owns multiple domains.} To maintain agility in the face of individual domain blacklisting or takedown, adversaries usually operate across multiple domains \cite{kreibich2009spamcraft}. This belief is supported by empirical evidence, such as the study conducted by Bahnsen et al. \cite{bahnsen2018deepphish}, which analyzed a practical phishing attack dataset and identified three adversaries owning 19, 106, and 309 domains, respectively, to host a total of 102, 1,007, and 7,927 phishing URLs. In addition, Hao et al. \cite{hao2016predator} advocate for registering multiple domains in bulk as a means for adversaries to reduce costs.

\par \textbf{2. Interconnected adversary domains for efficient evasion in bulk.} Interconnected domains allow the adversary to conduct attacks across those domains while maintaining a low profile. Conversely, isolated domains necessitate an individualized approach to attacking them, which can be more easily detectable. Such attacks may impact non-adversary domains connected to the adversary's domains, which exist as nodes in a DMG but are neither known nor exploited by the adversary. To investigate this scenario, we conduct the following experiment.

\par We randomly sample 100 domain nodes from a given test graph such that they are not mutually connected while they have edges with other nodes. These domain nodes are regarded as the adversary's nodes. Next, we use the integrated gradients adversarial attack (IG-ADV) proposed by Wu et al. \cite{wu2019adversarial} to evade the detection of each adversary domain individually. We then assess the effects of attacking each domain node on its benign neighbors on the graph. We quantify the impact of such an approach using the average neighbor negative flip rate (NFR). NFR represents the percentage of benign neighbors that turn malicious due to the evasion attempt of the malicious node. Fig. \ref{essential_motivation} shows the average neighbor NFR for domain nodes of the same number of benign neighbors. As depicted in the figure, an adversarial attack aimed at evading detection on a specific domain node has a detrimental impact on its benign neighbors. This negative effect makes the attack more easily detectable, compelling adversaries to avoid using mutually isolated domains. Instead, it is advisable for an adversary to have a subgraph of connected domains, as highlighted in orange in Fig. \ref{system_model}.

\par \textbf{3. No Interference among adversary domains.} It is essential to ensure that an attack aimed at evading detection on a specific domain node should not undermine the evasion strategy employed by other connected domain nodes. We make the observation that perturbations at a given adversary node optimized to maximize the loss function at this node to evade the detection may contradict the evasion of other adversarial nodes in its $k$-hop neighborhood. We carry out a detailed analysis to demonstrate this observation in Appendix \ref{appendixA}.1. This observation asserts that attempting to evade the detection of individual nodes belonging to the adversary, without considering the coordinated evasion of its entire subgraph, as typically done in existing targeted adversarial attacks, undermines the adversary's objective of evading detection altogether. 


\par The following experiment is conducted to practically examine the above observation. We consider an adversary with 100 domain nodes. Then, we apply the IG-ADV attack \cite{wu2019adversarial} as a representative of targeted attacks on each domain node individually and gradually till attacking all the nodes. Meanwhile, we calculate the attack success rate (ASR) as the percentage of adversary domain nodes converted from malicious into benign\footnote{This differs from our definition of the ASR in the other experiments, where it is the ratio of malicious domains evaded from the overall number of malicious domains.}. Fig. \ref{motivation_exp} illustrates the relationship between the average ASR and the number of attacked nodes. While a single adversarial attack can successfully evade detection for a single node with an ASR of over 80\%, it remains unclear whether these attacks can still be effective when targeting multiple connected nodes. As the number of attacked nodes increases, the overall ASR decreases significantly, indicating that individual attacks can harm the evasion attempts of their connected nodes.\footnote{It is noted that IG-ADV is operated with the same constraints either when attacking single or multiple nodes. Thus, the degradation in ASR is due to overlooking the impact on connected nodes and is not due to constraints.}


\begin{figure}[!htb]
\centering
\resizebox{0.9555\columnwidth}{!}{
\includegraphics{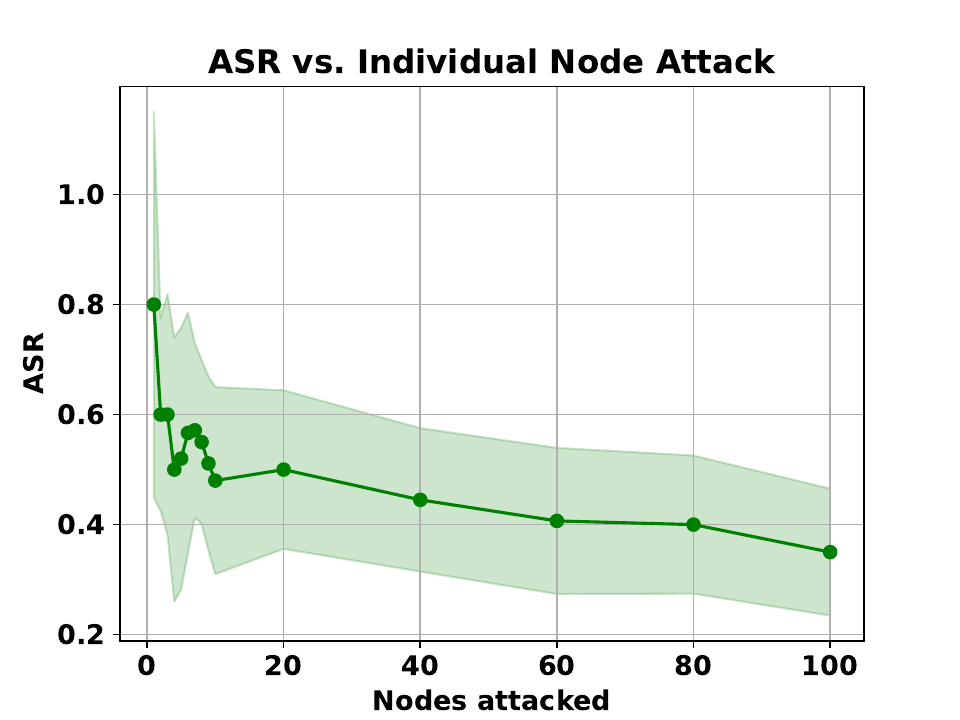}
}
\caption{The attack success rate (ASR) of attacked nodes with sequentially applying a targeted adversarial attack.}
\label{motivation_exp} 
\end{figure}

\par The preceding discussion asserts that the current adversarial attack methods are insufficient to fulfill the attackers' needs to carry out effective multi-node detection evasive attacks. Therefore, we devise a two-objective optimization problem for this multi-instance detection evasive attack. This problem aims to determine the best feature perturbation at a specific node $i$, denoted as $\Delta \boldsymbol{X'}_i$, which would help evade the detection of both node $i$ and its neighboring nodes $j \in \mathcal{N}(i)$, where $\mathcal{N}(i)$ is the set of direct neighbors to $i$. In this study, we acknowledge the possibility of an adversary strategically optimizing perturbations across all nodes to achieve their attack objectives. Nevertheless, our approach focuses on optimizing perturbations at individual nodes, specifically node $i$, and assessing their impact on its direct neighbors, denoted as $j$. We find this strategy adequate for crafting a coordinated attack across the adversary's nodes because targeting the direct neighbors of node $i$ allows the attack impact to propagate effectively throughout the graph.


\par Since the GNN-based MDD model is unknown to adversaries, we use a 2-hop linearized graph convolutional network (GCN) surrogate model \cite{wu2019simplifying}. This surrogate, known as the simplified GCN model, is commonly used in adversarial attacks \cite{zugner2018adversarial,li2021adversarial,chen2022understanding,zhu2022resisting} due to its simplicity and the transferability of attacks conducted on it to different GNN architectures. This surrogate model has the same task as the target model (node classification) and is trained on a different dataset that is labeled by querying the target model. The sole similarity between the actual and surrogate models lies in the graph task they carry out, which is node classification. The surrogate model is shown to work well through validation, and attacking it is expected to transfer to the unknown target MDD model. 


\par Let us use $\mathcal{G}=(\mathcal{V},\mathcal{E})$ to denote the DMG constructed by the MDD entity where $\mathcal{V}$ is the set of nodes and $\mathcal{E}$ is their edges. Equivalently, $\mathcal{G}$ can be written as $\mathcal{G}=(\boldsymbol{A}, \boldsymbol{X})$ where $\boldsymbol{A}$ is the adjacency matrix and $\boldsymbol{X}$ is the node attribute matrix. The adversary knows only its node set which we denote by $\mathcal{V'} \subset \mathcal{V}$ and some edges connecting them $\mathcal{E'} \subset \mathcal{E}$. Let us denote by $\mathcal{G}'=(\boldsymbol{A'}, \boldsymbol{X'})$ the adversary subgraph with adjacency matrix $\boldsymbol{A'}$ and node feature matrix $\boldsymbol{X'}$. We can write messages at the second layer of the surrogate model as $\boldsymbol{H'}^{(2)}=\boldsymbol{\hat{A'}}^2 \boldsymbol{X'} \boldsymbol{W}$, where $\boldsymbol{\hat{A'}}=\boldsymbol{D'}^{-\frac{1}{2}}(\boldsymbol{A'}+\boldsymbol{I}) \boldsymbol{D'}^{-\frac{1}{2}}$ is the normalized symmetric adjacency matrix, $\boldsymbol{D'}$ is a diagonal matrix of node degrees \cite{kipf2016semi}, and $\boldsymbol{W}$ denotes the coefficients of the surrogate model. For simplicity, let $\boldsymbol{B}$ denote $\boldsymbol{\hat{A'}}^2$. Further, let us assume a sigmoid loss function, and ignore it in the loss calculation as done in \cite{zugner2018adversarial}. Thus, the change in the loss function of the surrogate model at node $i$ due to a perturbation $\Delta \boldsymbol{X}_i$ can be written as follows.
\begin{equation} \label{eq8}
\Delta \mathcal{L}_i\left(\boldsymbol{A'}, \Delta \boldsymbol{X'}_i ; \boldsymbol{W}, i\right)=
\|\boldsymbol{B} \Delta \boldsymbol{X'}_i \boldsymbol{W}\|_2^2.
\end{equation}
\noindent Similarly, the change in loss at node $j \in \mathcal{N}(i)$ is:
\begin{equation} \label{eq9}
\Delta \mathcal{L}_j\left(\boldsymbol{A'}, \boldsymbol{X'}+\Delta \boldsymbol{X'}_i; \boldsymbol{W}, j\right)=\|
\boldsymbol{B} (\Delta \boldsymbol{X'}_i+\boldsymbol{H'}_j) \boldsymbol{W}\|_2^2,
\end{equation}
\noindent where $\mathcal{L}_i$ ($\mathcal{L}_j$) is the loss function value at node $i$ ($j$), and $\boldsymbol{H'}_j$ is the message at node $j$.


\par Considering a given adversary's node $i$, and its direct neighbors $j \in \mathcal{N}(i)$ in its neighborhood set $\mathcal{N}(i)$ known to the adversary (in its subgraph), the formulation of the proposed feature optimization problem is expressed in (\ref{eq1}). It is noted that $i$ may have other neighbors known to the MDD entity but not known to the adversary.
\begin{multline}
\Delta \boldsymbol{X'}^*_i =\underset{\Delta \boldsymbol{X'}_i}{\operatorname{argmax}}~ \alpha\Delta \mathcal{L}_i\left(\boldsymbol{A'}, \boldsymbol{X'}+\Delta \boldsymbol{X'}_i ; \boldsymbol{W}, i\right)
\\+\sum_{j \in \mathcal{N}(i)} \beta \frac{1}{d_j} \Delta \mathcal{L}_j\left(\boldsymbol{A'}, \boldsymbol{X'}+\Delta \boldsymbol{X'}_i ; \boldsymbol{W}, j\right),~\forall~i \in \mathcal{V}',\label{eq1}
\end{multline}
\noindent where $\alpha$ and $\beta$ are weighted average parameters to control the trade-off between maximizing the local loss at node $i$ and maximizing the loss at the neighbors $j \in \mathcal{N}(i)$.



\par A solution to the problem in (\ref{eq1}) is obtained by maximizing a weighted average of $F_1= \Delta \mathcal{L}_i\left(\boldsymbol{A'}, \boldsymbol{X'}+\Delta \boldsymbol{X'}_i ; \boldsymbol{W}, i\right)$ and $F_2= \sum_{j \in \mathcal{N}(i)} \frac{1}{d_j} \Delta \mathcal{L}_j\left(\boldsymbol{A'}, \boldsymbol{X'}+\Delta \boldsymbol{X'}_i; \boldsymbol{W}, j\right)$. With the use of the surrogate model, the magnitude of the change in loss $\Delta \mathcal{L}_i$ due to $\Delta \boldsymbol{X'}_i$ can be written as $ \|\mathcal{L}_i(\boldsymbol{A'}, \boldsymbol{X'}+\Delta \boldsymbol{X'}_i; \boldsymbol{W},i)- \mathcal{L}_i(\boldsymbol{A'}, \boldsymbol{X'}; \boldsymbol{W},i)\|=\|\boldsymbol{B} \boldsymbol{X'}_i \boldsymbol{W}\|_2^2$. The optimization of $F_1$ with respect to $\Delta \boldsymbol{X'}_i$ can be written as follows.
\begin{equation}\label{eq2}
\underset{\Delta \boldsymbol{X'}_i}{\operatorname{argmax}} ~F_1=\underset{\Delta \boldsymbol{X'}_i}{\operatorname{argmax}}\left\|\boldsymbol{B} \Delta \boldsymbol{X'}_i \boldsymbol{W}\right\|_2^2.
\end{equation} 
\noindent Similarly, we can view optimizing $F_2$ as follows.
\begin{equation}\label{eq3}
\underset{\Delta \boldsymbol{X'}_i}{\operatorname{argmax}} ~F_2=\underset{\Delta \boldsymbol{X'}_i}{\operatorname{argmax}}\sum_{j\in \mathcal{N}(i)} \left\|\boldsymbol{B}\left[\Delta \boldsymbol{X'}_i+\boldsymbol{H'}_j\right] \boldsymbol{W}\right\|_2^2.
\end{equation}
\noindent From (\ref{eq2}), and Corollary \ref{corollary} in Appendix \ref{appendixA}.1, an optimal $\Delta \boldsymbol{X'}_i$ for maximizing $F_1$ is the one that maximizes the sum of its inner products with the columns in the model coefficient matrix $\boldsymbol{W}$. Equivalently, it is a perturbation vector that maximizes $\boldsymbol{\Phi}_i=\boldsymbol{W}^T \Delta \boldsymbol{X'}_i$, where $\boldsymbol{\Phi}_i$ is equal to $\boldsymbol{W}^T$. Then, equivalently, an optimal $\Delta \boldsymbol{X'}_i$ is the one that maximizes the quantity $\|\boldsymbol{\Phi}_i \boldsymbol{X'}_i\|_2^2$. So the problem of perturbation optimization becomes a matrix-vector inner product maximization problem, which we can solve using Theorem \ref{theorem2}, detailed below.

\begin{thm}
\label{theorem2}
The following problem:
\begin{equation}\label{eq7}
\begin{gathered}
\underset{\Delta \boldsymbol{\boldsymbol{X'}_i}}{\operatorname{argmax}}~\|\boldsymbol{\Phi}_i \Delta\boldsymbol{ \boldsymbol{X'}_i}\|_2^2, \\
\text { s.t. } \|\Delta \boldsymbol{ \boldsymbol{X'}_i}\|_2^2=\epsilon,
\end{gathered}
\end{equation}
\noindent where $\epsilon$ is a threshold representing the budget of the perturbation, has a closed-form solution, which is the eigenvector corresponding to the largest eigenvalue of $\boldsymbol{\Phi}_i^T \boldsymbol{\Phi}_i$.
\end{thm}
\par The proof of Theorem \ref{theorem2} is in Appendix \ref{appendixA}.2. According to Theorem \ref{theorem2}, to solve (\ref{eq2}), an optimal $\Delta \boldsymbol{X'}_i$ with respect to $F_1$ is the principal eigenvector of the matrix $\boldsymbol{\Phi}_i \boldsymbol{\Phi}_i^T =\boldsymbol{W}^T \boldsymbol{W}$, denoted by $\boldsymbol{e}_i$. Also, as shown in our analysis of the need for a coordinated subgraph attack in Appendix \ref{appendixA}.1, to solve (\ref{eq3}) for a given $j \in \mathcal{N}(i)$, $\Delta \boldsymbol{X'}_i$ with respect to $F_2$, is the principal eigenvector of the matrix $
\boldsymbol{\Phi}_j \boldsymbol{\Phi}^T_j$, where $\boldsymbol{\Phi}_j=(\boldsymbol{W}-\boldsymbol{H'}_j)$. Let us denote this solution by $\boldsymbol{e}_j$. Since $\boldsymbol{e}_i$ optimizes $F_1$ and $\boldsymbol{e}_j$ optimizes $F_2$ $~\forall~j \in \mathcal{N}(i)$, we can approximately meet both objective functions by a perturbation which is the weighted average of these solutions as shown in (\ref{eq6}).
\begin{equation}\label{eq6}
\Delta \boldsymbol{X'}_i^*=\alpha\boldsymbol{e}_i+\sum_j \beta d_j \boldsymbol{e}_j.
\end{equation}
\noindent The perturbation obtained in (\ref{eq6}) maximizes the weighted average of the loss functions at the node itself and its neighbors where $\alpha$ and $\beta$ control their relative importance, respectively.


\begin{algorithm}[t!]
\caption{MintA-feature perturbation.}
\label{Algorithm1}
\begin{algorithmic}[1]{
\algsetup{linenosize=\small}
\renewcommand{\algorithmicrequire}{\textbf{Input:}} 
\renewcommand{\algorithmicensure}{\textbf{Output:}}
\REQUIRE An adversary subgraph $\mathcal{G'}:(\boldsymbol{A'} \in \mathbb{R}^{n \times n},\boldsymbol{X'} \in \mathbb{R}^{n \times k})$, a feature perturbation budget $k_f$, and a set of editable node features $p$.
\ENSURE A modified adversary subgraph $\mathcal{G'}^*:(\boldsymbol{A'},\boldsymbol{X'}^*)$ 
\STATE{Obtain labeled data by querying the target model.}
\STATE{Train a surrogate model.}
\STATE{Initialize total feature perturbation $\Delta \boldsymbol{X'}\gets\mathbf{0}$, $i=1$.}
\STATE{While $i\leq n ~ \text{AND} ~ \|\Delta \boldsymbol{X'}\|_2 \leq k_f$}
\STATE{Obtain an optimal perturbation $\Delta \boldsymbol{X'}_i^*$ according to (\ref{eq6})}
\STATE{Modify the editable features in $\boldsymbol{X'}_i$ ($p$) to best match $\boldsymbol{X'}_i^*=\boldsymbol{X'}_i+\Delta \boldsymbol{X'}_i^*$ }
\STATE{Update $\Delta \boldsymbol{X'}=\Delta \boldsymbol{X'}+\Delta\boldsymbol{X'}_i^*$.}
\STATE{Increment $i$.}
\RETURN {$\boldsymbol{X'}^*=\boldsymbol{X'}+\Delta \boldsymbol{X'}$}
}
\end{algorithmic}
\end{algorithm}


\par The adversary possesses nodes that are present in the DMG, and these nodes may be connected to non-adversary nodes. Such neighboring non-adversary nodes are susceptible to the perturbations made by their adversary neighbors. Therefore, it is worthwhile to examine the consequences of the adversary's perturbations on the non-adversary neighbors of their nodes. Proposition \ref{prop1} compares this impact to situations where attacks are optimized for individual nodes.


\begin{prop}
\label{prop1}
\par Consider an adversary node $i$, connected to adversary nodes $j \in \mathcal{N}(i)$ and a non-adversary node $l$. On node $l$, the effect of a perturbation optimized by maximizing the loss over the node $i$ and its direct neighbors $j \in \mathcal{N}(i)$ is smaller than the effect of optimizing the loss on only $i$.
\end{prop}
\noindent The proof of Proposition \ref{prop1} is in Appendix \ref{appendixA}.3.


\subsection{The proposed MintA algorithm}
\par The above perturbation optimization forms the foundation of our proposed MintA attack. The proposed MintA algorithm is described in Algorithm \ref{Algorithm1}. First, the adversary trains a surrogate model (Steps 1-2). Then, it constructs an estimate of its attributed adversary subgraph. After that, the proposed attack is performed collaboratively on adversary nodes by optimizing their perturbations according to (\ref{eq6}) and within a given feature perturbation budget $k_f$ (Steps 3-5). Next, the optimal feature perturbation is approximated by manipulating the editable node features to best fit the desired optimized values (Step 6). The final modified feature matrix of the adversary nodes is obtained accordingly. MintA is also applicable to edge perturbations as well. This can be done by selecting edge edits that can best match the average of the objective functions in (\ref{eq6}). Algorithm \ref{Algorithm2} presents the edge perturbation steps.

\begin{algorithm}[t!]
\caption{MintA-edge perturbation.}
\label{Algorithm2}
\begin{algorithmic}[1]{
\algsetup{linenosize=\small}
\renewcommand{\algorithmicrequire}{\textbf{Input:}} 
\renewcommand{\algorithmicensure}{\textbf{Output:}}
\REQUIRE An adversary subgraph $\mathcal{G'}:(\boldsymbol{A'} \in \mathbb{R}^{n \times n},\boldsymbol{X'} \in \mathbb{R}^{n \times k})$, an edge perturbation budget $k_e$.
\ENSURE A modified adversary subgraph $\mathcal{G'}^*:(\boldsymbol{A'}^*,\boldsymbol{X'})$ 
\STATE{Obtain labeled data by querying the target model.}
\STATE{Train a surrogate model.}
\STATE{Initialize the number of edge flips $i=0$.}
\STATE{While $i\leq k_e$}
\STATE{Select the node pair having the maximum value of the average of $F_1$ and $F_2$ appearing in (\ref{eq2}) and (\ref{eq3}).}
\STATE{Edge flip: edit the name or change the IP resolution to implement the edge edit.}
\STATE{Increment $i$.}
\RETURN {$\boldsymbol{A'}^*$}
}
\end{algorithmic}
\end{algorithm}

\begin{table}[htb]
\centering
\caption{A summary of research questions and answers.}
\begin{tabular}{|c|l|l|}
\hline
Q & Property investigated & Key Result \\ \hline
1 & Effectiveness (adversary nodes) & High ASR \& Low NFR \\ \hline
2 & Stealthiness (non-adversary nodes) & Low ASR \& Low NFR \\ \hline
3 & Robustness: outlier detection & High robustness \\ \hline
4 & Robustness: graph purification & High robustness \\ \hline
5 & Baseline comparisons & Higher ASR \& Lower NFR \\ \hline
6 & Costs & Low and scalable \\ \hline
\end{tabular}
\label{summarytable}
\end{table}

\section{Experiments}
\label{Section5}
\par We conduct experiments to evaluate MintA's performance. The source code and data are available on the link:
\href{https://github.com/mahmoudkanazzal/MintA}{https://github.com/mahmoudkanazzal/MintA}. Our evaluation aims to answer the questions summarized in Table~\ref{summarytable}.

\begin{table}[htbp]
 \centering
 \caption{Key statistics of the dataset used (Client and IP nodes have no labels).}
 \begin{tabular}{|l|l|l|l|}
 \hline
Node Type & Eidsiva & PDNS & Total \\ \hline
Clients & 12035 & 0 & 12035 \\\hline
Domains (total) & 12035 & 465,905 & 477940 \\\hline Domains (benign) & 5,000 & 4,963 & 9,963 \\\hline Domains (malicious) & 5,000 & 20,354 & 25,354 \\\hline
IPs & 8000 & 73,593 & 81593\\ \hline
 \end{tabular}
 \label{table2}
\end{table}

\subsection{The setup and dataset}

\par In our experiments, we consider Sun et al.'s algorithm \cite{sun2020deepdom} as a target GNN-based MDD model. The approach in \cite{sun2020deepdom} treats MDD as a semi-supervised node classification problem. The DMG in \cite{sun2020deepdom} is a heterogeneous graph created according to the network schema in Fig. \ref{network_schema}(a) with the 21 FANCI features \cite{schuppen2018fanci} as its domain node attributes. It is noted that \cite{sun2020deepdom} and all the other works in \cite{manadhata2014detecting,khalil2016discovering,sun2019hindom,li2021dydom,zhang2021attributed,li2022heterogeneous} do not share datasets or artifacts as their graphs have private enterprise data. Therefore, we obtain a heterogeneous DMG with a similar structure by merging a domain-IP graph \cite{kumarasinghe2022pdns} with the Eidsiva enterprise graph, obtained from Eidsiva Bredband (a broadband telecom operator in Norway), collected by \cite{rismyhr2020graph}, and available at \cite{rismyhr}. Table \ref{table2} lists key statistics of the merged dataset. It is noted that we use another dataset from \cite{marques2021dns} to query the target MDD model for obtaining data to train the surrogate model. By doing that, the surrogate model is trained independently from the training and testing DMGs of the target model. Experiments are conducted on a Lambda GPU Workstation with 128 GB of RAM, two GPUs each of 10 GB RAM, and an I9 CPU of 10 cores at a clock speed of 3.70 GHz.

\begin{figure}[htb]
\centering
 \resizebox{0.999\columnwidth}{!}{
\begin{tabular}{cc}
\includegraphics[width=12cm]{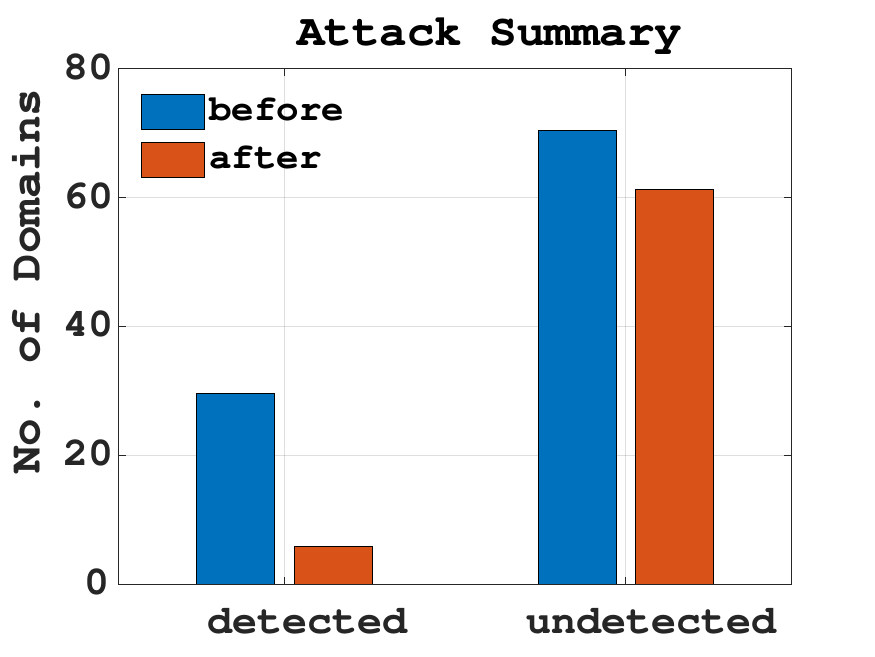}
&
\includegraphics[width=12cm]{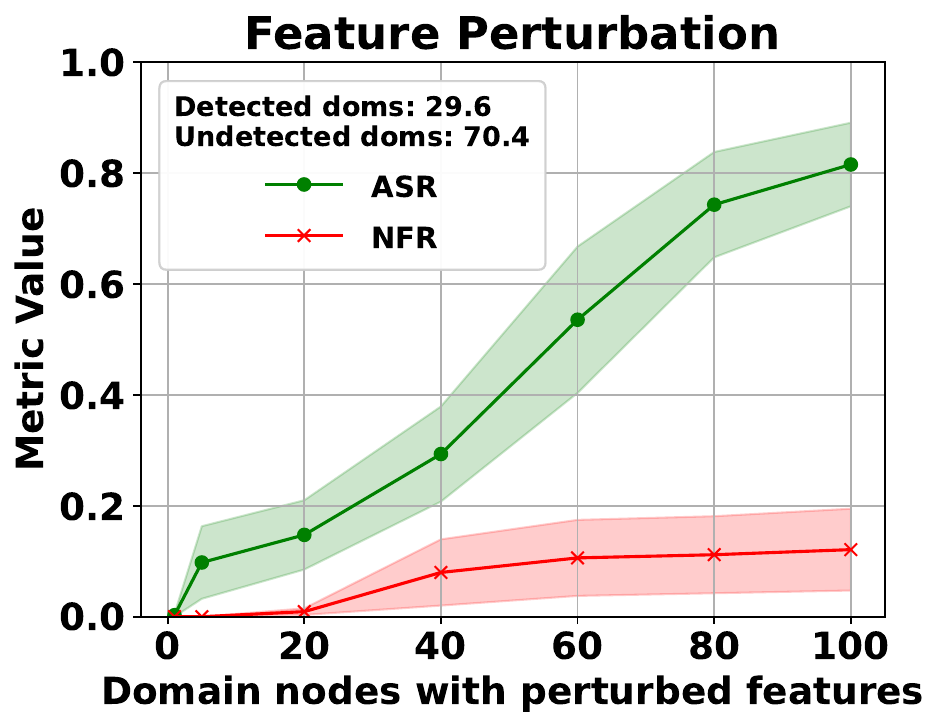}
 \\
 \Huge{(a)} & \Huge{(b)} \\
\multicolumn{2}{c}{
\includegraphics[width=12cm]{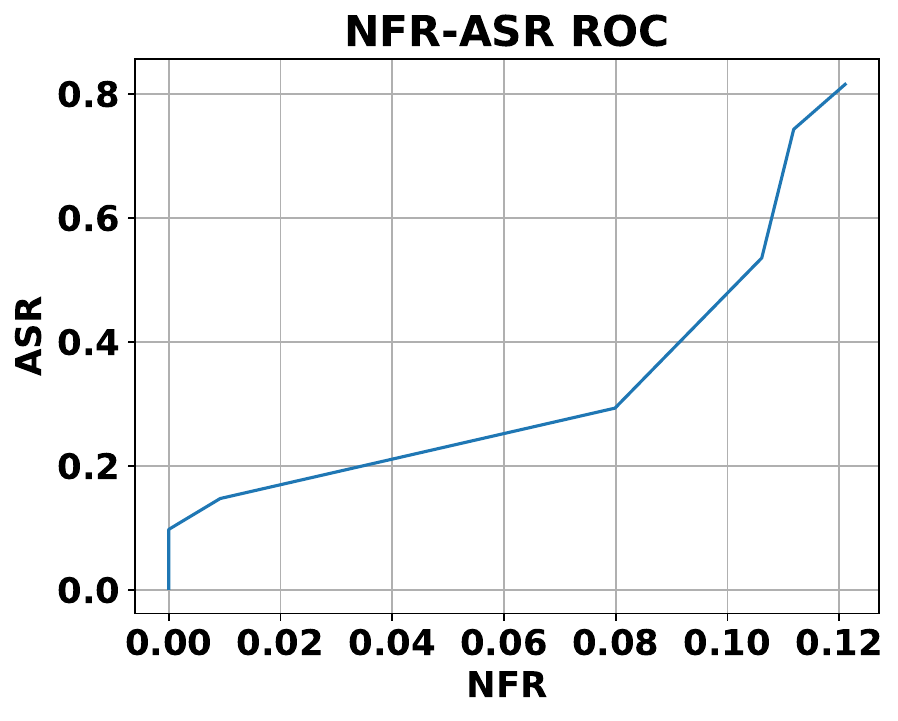}}
\\
\multicolumn{2}{c}{
\Huge{(c)}}
\end{tabular}}
\caption{With feature perturbation and \textit{created adversary subgraphs}; attack summary when 100 domains are attacked in (a), ASR, NFR, and ROC when less than 100 domains are attacked in (b) and (c), respectively.}
\label{feature_res_mydoms} 
\end{figure}

\par To conduct experiments with the proposed attack, it is essential to model the adversary's intervention in the DMG using what we refer to as the ``adversary's subgraph''. While subgraphs of domains owned by real adversaries would be ideal for this purpose, unfortunately, such data is not accessible. As a result, we resort to the following modeling approaches and incorporate them into our experiments. 
\begin{itemize}[leftmargin=*]
\item A \textit{created adversary modeling approach:} We create a set of domain nodes, which will serve as the adversary's nodes (Registered in Feb. 2023 and intended to persist for a few months.) The MDD entity then utilizes the DNS logs to construct a DMG that incorporates these nodes. We use a free service provided by \url{https://profreehost.com} to create and host ULRs associated with the registered domain names, which will be dedicated to emulating the adversary's domain nodes. This process entails the creation of URLs, their free hosting, and the linkage of domain names to these URLs. By doing so, the adversary can construct its subgraph based on the knowledge of these domain names and their respective host IP addresses. It is noteworthy that this approach closely reflects the constraints faced by actual adversaries. However, it has some limitations as the target MDD might not accurately classify the URLs associated with these domains as malicious. To address this limitation, we also employ the following additional complementary adversary modeling approach.
\item A \textit{sampled adversary modeling approach:} In a given inference DMG, we acquire adversary nodes by randomly sampling connected nodes solely from the set of domain nodes that are labeled as malicious. We incorporate these adversary nodes into the test set to enable the MDD model to infer their malicious nature. This approach effectively models the adversary as a collection of interconnected malicious nodes.
\end{itemize}

\subsection{Performance evaluation of MintA}
\par In the following experiments, we address Q1: Does MintA meet the adversary's goals? We evaluate the performance of MintA in terms of two metrics. First is the attack success rate (ASR) which is the percentage of undetected malicious adversary domains due to the attack. Second, is a negative flip rate (NFR) which is the percentage of undetected adversary domains before applying MintA that are classified as malicious after we conduct MintA. This represents the side effect of the attack. It is noted that for the adversary to maximize the effectiveness of the attack, all of its domains need to be attacked (i.e., 100 in our case). Nonetheless, in our experiments, we examine the ASR and NFR for the cases of attacking less than 100 domains. We also present the ASR-NFR trade-off in a ROC curve. We repeat each experiment for 30 trials and report the average values of these metrics with their 95\% confidence intervals. In each trial, we consider a different realization of the MDD model, the adversary nodes, and the attack.

\begin{figure}[htb]
\centering
 \resizebox{0.999\columnwidth}{!}{
\begin{tabular}{ccc}
\includegraphics[width=12cm]{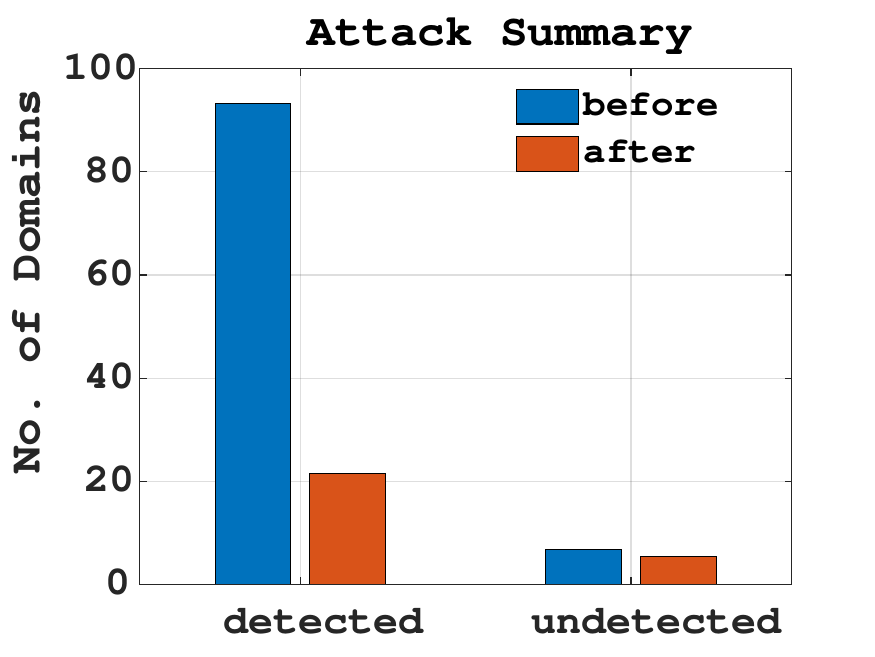}
&
 \includegraphics[width=12cm]{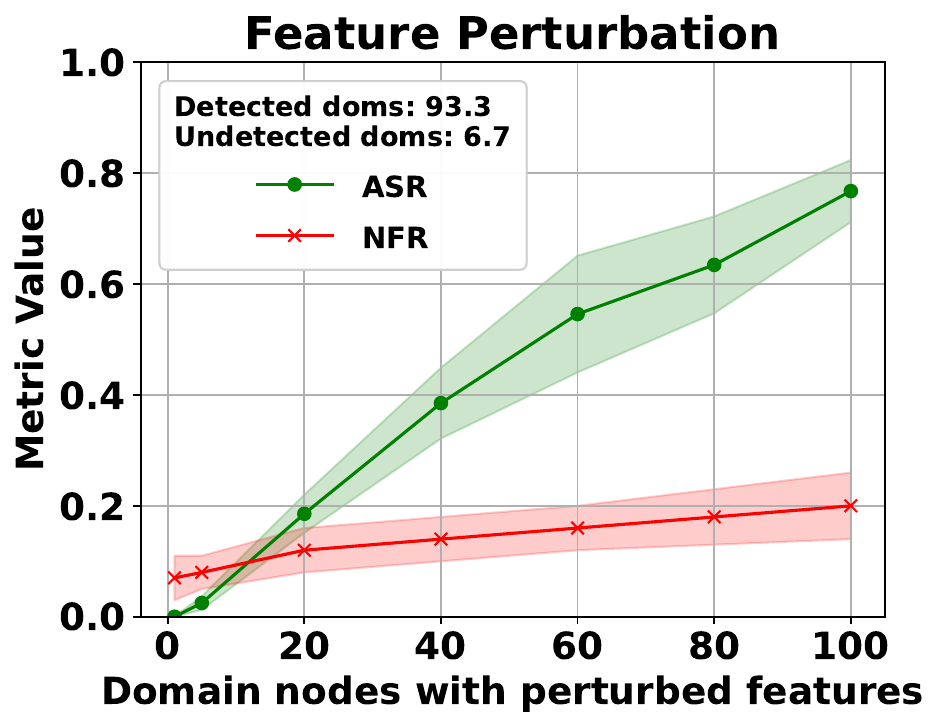}
 \\
\Huge{(a)} & \Huge{(b)}
 \\
\multicolumn{2}{c}{
\includegraphics[width=12cm]{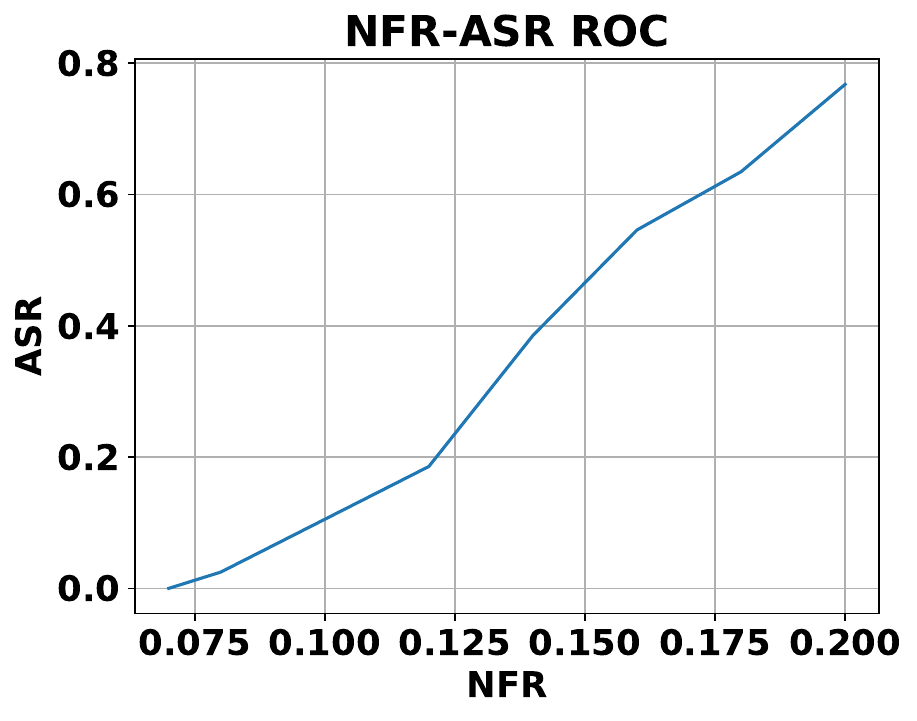}}
\\
\multicolumn{2}{c}{
 \Huge{(c)} }
\end{tabular}}
\caption{With feature perturbation and \textit{sampled adversary subgraphs}; attack summary when 100 domains are attacked (a), ASR, NFR, and ROC when less than 100 domains are attacked in (b) and (c), respectively.}
\label{feature_res_ras} 
\end{figure}

\subsubsection{Feature perturbation}
\par In this experiment, we record the ASR and NFR after the adversary activates the feature perturbation attack (Algorithm \ref{Algorithm1}). We set the feature perturbation budget $k_f$ to the sum of dynamic ranges of editable features (domain name length feature range is 5 and binary features' range is 1, so the budget is 5.38 times the number of nodes involved in the attack). First, we consider the \textit{created adversary subgraph}. We found that before activating the attack, around 29.6\% of the domains created on average are detected as malicious by the target MDD model, and the remaining 70.4\% on average are not detected\footnote{This rationale is justifiable since these newly created domains do not currently host real cyber-attacks. As a result, they lack exploitable associations, such as connections with compromised clients, that could potentially expose their existence.}. The attack summary is shown in Fig. \ref{feature_res_mydoms}(a). It shows that the attack evades the detection of 23.7 (0.80 $\times$ 29.6) adversary domains on average out of the 29.6\% detectable adversary domains. The attack results in the detection of around 9.15 (0.13 $\times$ 70.4) domains on average from the undetectable domains. This shows that the attack is successful because the number of evaded domains is significantly larger than the ones detected after launching MintA. For the cases of attacking less than 100 domains, the ASR and NFR are shown in Fig \ref{feature_res_mydoms}(b) and the ROC curve is in Fig. \ref{feature_res_mydoms}(c). 

\par Next, we repeat the above experiment with \textit{sampled adversary subgraphs}. Before launching the MintA attack, we found out that 93.3\% of the domains sampled, on average, are detected as malicious by the target MDD model, and the remaining 6.7\% on average are not detected. For this experiment, the attack summary shown in Fig. \ref{feature_res_ras}(a) shows that the attack evades the detection of 71.8 (0.77 $\times$ 93.3) adversary domains on average out of the 93.3\% detectable adversary domains. The attack results in the detection of about 1.3 (0.2 $\times$ 6.7) domains on average from the undetectable domains. Therefore, this attack is successful and is even more successful than in the case of the \textit{created adversary subgraph} since more domains are evaded and fewer are made detectable. For the cases of attacking less than 100 domains, the ASR and NFR are shown in Fig. \ref{feature_res_ras}(b) and the ROC curve is in Fig. \ref{feature_res_ras}(c).

\subsubsection{Edge perturbation}
\par In this experiment, we evaluate edge perturbation attacks (Algorithm \ref{Algorithm2}). We allocate the budget for edge perturbation, denoted as $k_e$, to be equivalent to the number of edges that will be swapped among the adversary domain nodes involved in the attack. This decision is based on the fact that the adversary possesses knowledge and control solely over the edges connected to its own nodes. We select the \textit{apex}, \textit{resolve}, and \textit{similar} edges to perturb. Perturbing these edges can be implemented as specified in Section \ref{Section4}.1. 

\par First, we consider the \textit{created adversary subgraph}. The results with \textit{apex} edge perturbation are shown in the top row of Fig. \ref{edge_perturb_res_mydoms}. For this attack, the attack summary plot is shown in Fig. \ref{edge_perturb_res_mydoms}(a). This plot shows that the attack evades the detection of 21.9 (0.74 $\times$ 29.6) adversary domains, on average, out of the 29.6\% detectable adversary domains. The attack results in the detection of around 8.5 (0.12 $\times$ 70.4) domains, on average, from the undetectable domains. This shows that the attack is successful because the number of evaded domains is significantly larger than the ones detected from the undetectable domains. Compared with feature perturbation, the attack has less ASR and NFR. For the cases of attacking less than 100 domains, the ASR and NFR are shown in Fig \ref{edge_perturb_res_mydoms}(b) and the ROC curve is in Fig. \ref{edge_perturb_res_mydoms}(c).

\begin{figure}[htb]
\centering
 \resizebox{0.999\columnwidth}{!}{
\begin{tabular}{ccc}
\includegraphics[width=12cm]{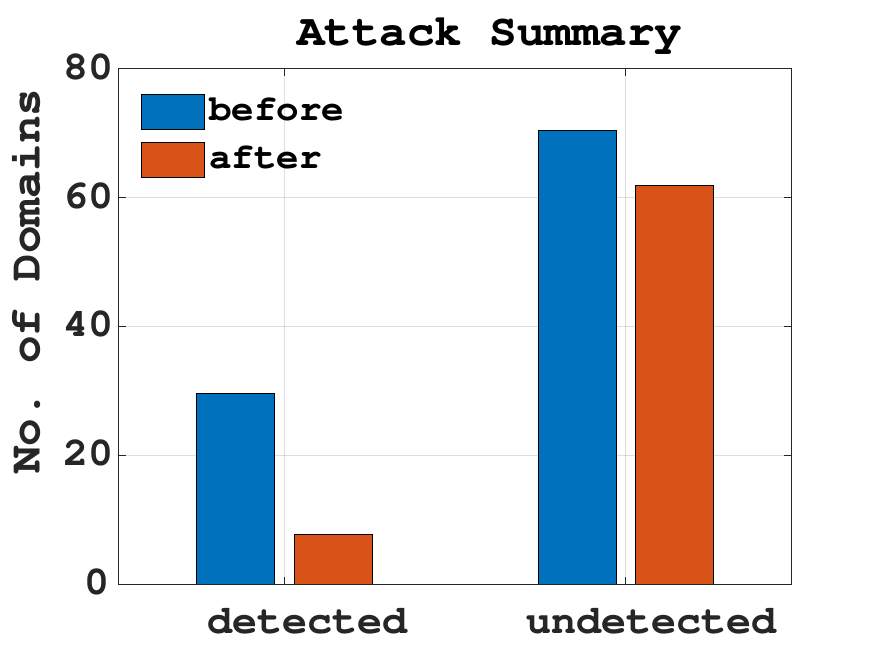}&
\includegraphics[width=12cm]{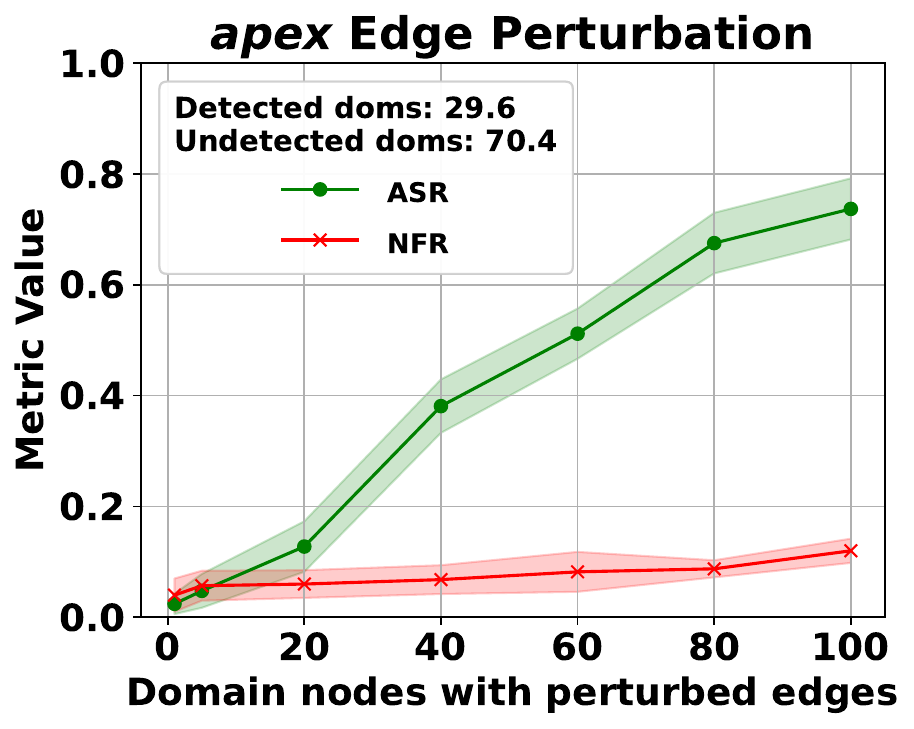}&
\includegraphics[width=12cm]{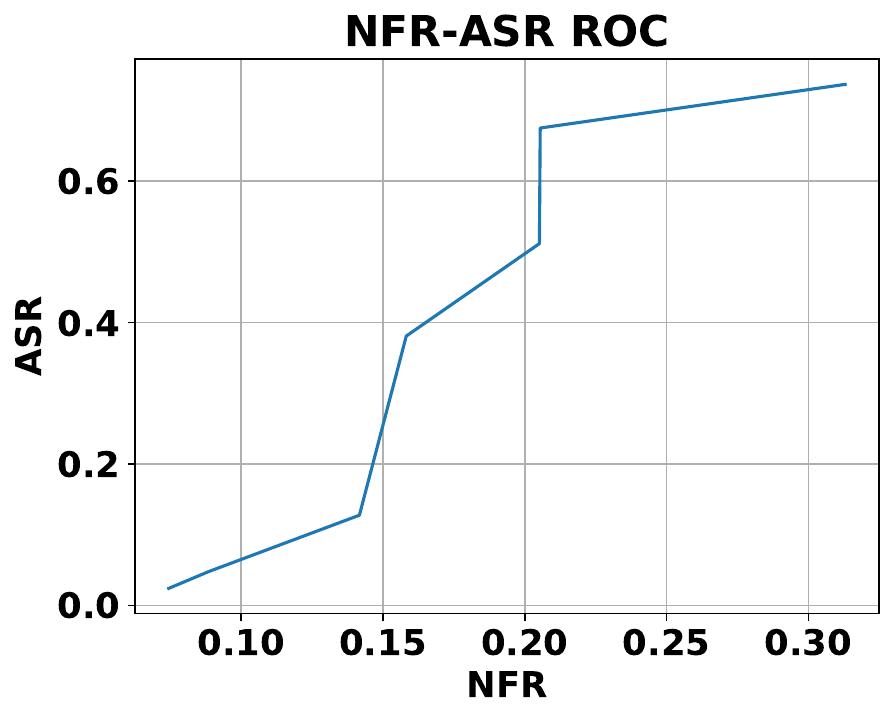}
\\
\Huge{(a)} & \Huge{(b)} & \Huge{(c)}
\\
\includegraphics[width=12cm]{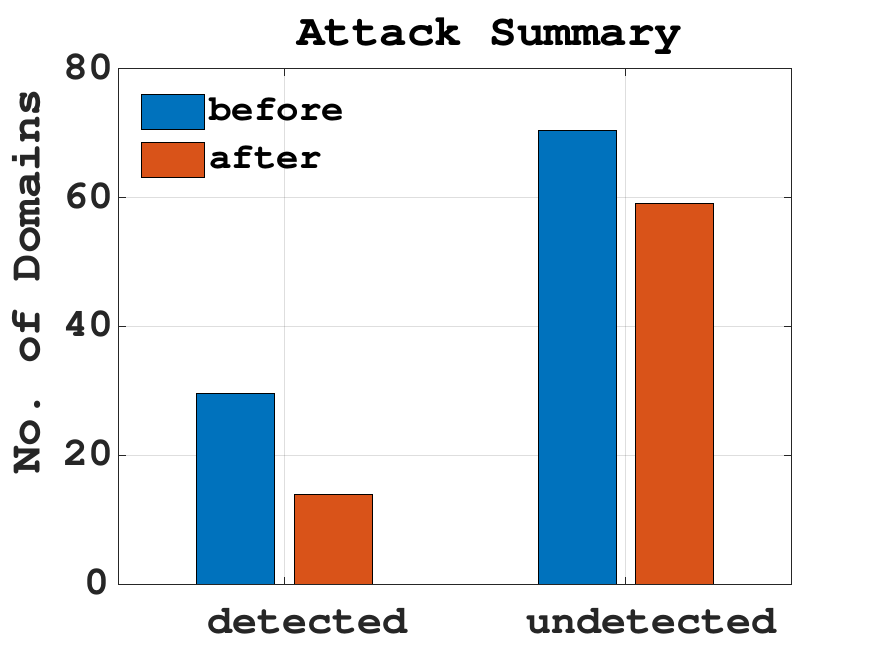}&
\includegraphics[width=12cm]{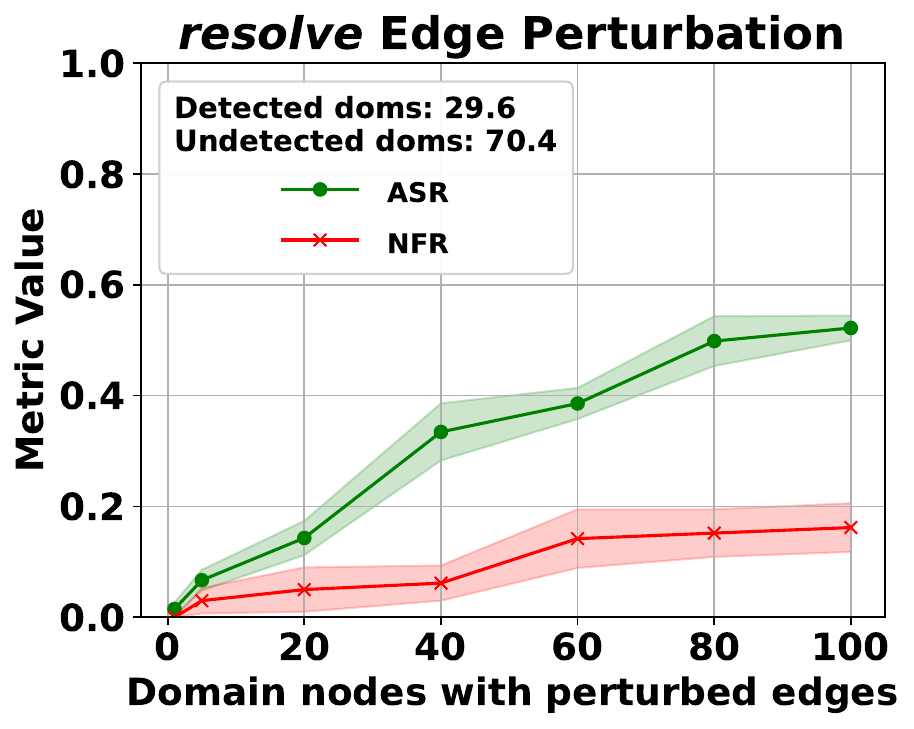}&
\includegraphics[width=12cm]{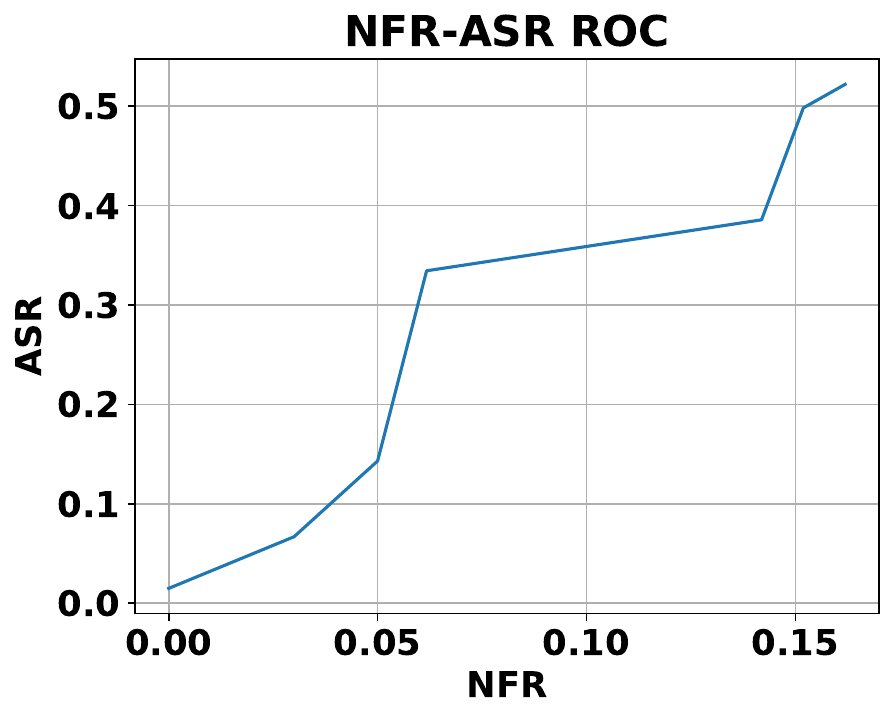}
\\
\Huge{(d)} &\Huge{(e)} & \Huge{(f)} 
\\
\includegraphics[width=12cm]{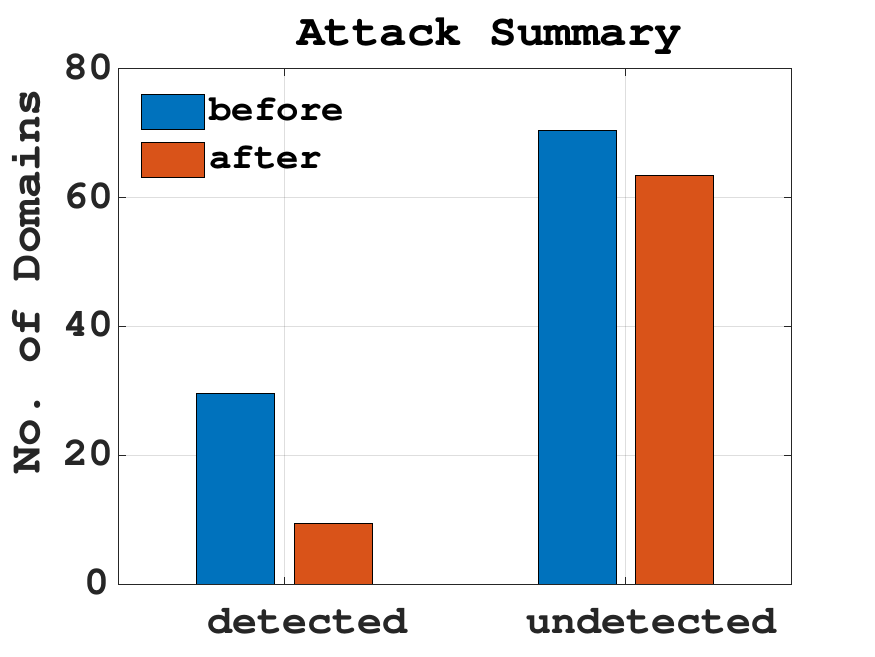}&
 \includegraphics[width=12cm]{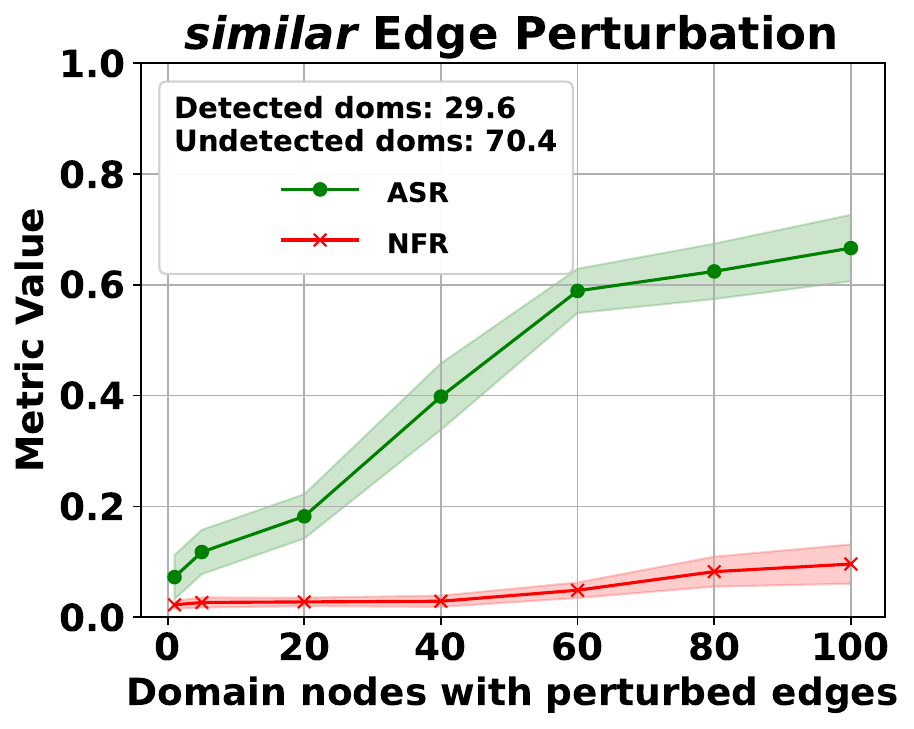}&
\includegraphics[width=12cm]{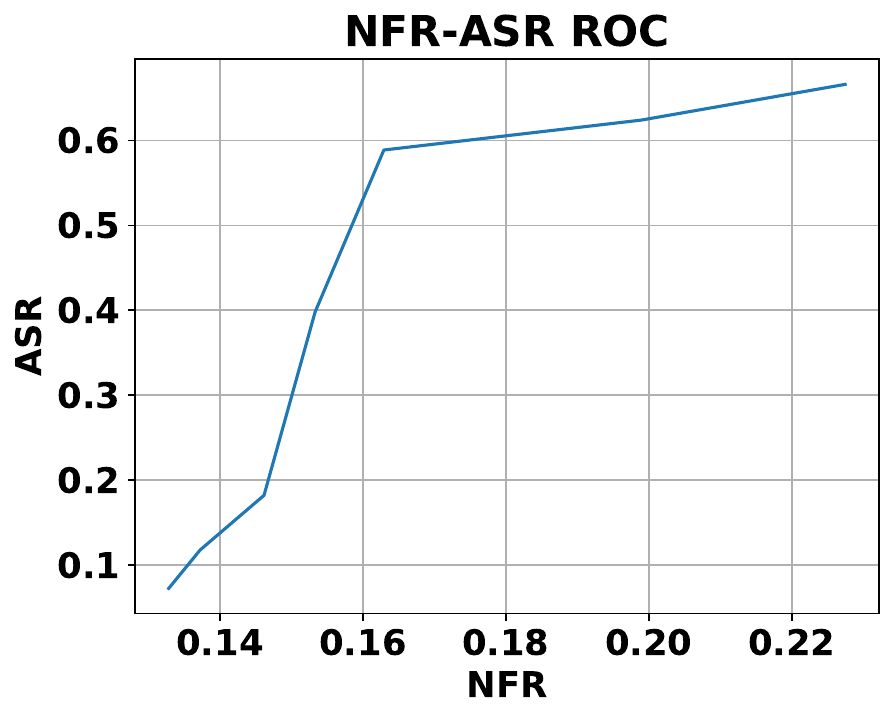}
\\
\Huge{(g)} &\Huge{(h)} &\Huge{(i)}
\end{tabular}}
\caption{With \textit{apex}, \textit{resolve}, and \textit{similar} edge perturbation and \textit{created adversary subgraphs}; attack summary when 100 domains are attacked, ASR, NFR, and ROC when less than 100 domains are attacked, in rows 1, 2, and 3, respectively.}
\label{edge_perturb_res_mydoms} 
\end{figure}

\par The \textit{resolve} edge perturbation results are shown in the second row of Fig \ref{edge_perturb_res_mydoms}. Looking at its attack's summary graph in Fig \ref{edge_perturb_res_mydoms}(d), it evades the detection of 15.6 (0.53 $\times$ 29.6) adversary domains, on average, out of the 29.6\% detectable adversary domains. The attack results in the detection of around 10.56 (0.15 $\times$ 70.4) domains, on average, from the undetectable domains. This attack is marginally successful. Nonetheless, it is less effective than feature and \textit{apex}-edge attacks. The reason is that the adversary has limited IP addresses, and it can only use these addresses. This restriction limits the effectiveness of the attack. For the cases of attacking less than 100 domains, the ASR and NFR are shown in Fig \ref{edge_perturb_res_mydoms}(e), and the ROC curve is in Fig. \ref{edge_perturb_res_mydoms}(f). Finally, the \textit{similar} edge perturbation results are shown in the last row of Fig \ref{edge_perturb_res_mydoms}. Considering the attack summary in Fig \ref{edge_perturb_res_mydoms}(g), this edge perturbation performs better than the \textit{resolve} edge (evades 20.1 domains on average (0.68 $\times$ 29.6), and causes the detection of 7.0 domains on average (0.1 $\times$ 70.4) and is less efficient than the \textit{apex} edge case. For the cases of attacking less than 100 domains, the ASR and NFR are shown in Fig \ref{edge_perturb_res_mydoms}(h) and the ROC curve is in Fig. \ref{edge_perturb_res_mydoms}(i).

\begin{figure}[htb]
\centering
 \resizebox{0.999\columnwidth}{!}{
\begin{tabular}{ccc}
\includegraphics[width=12cm]{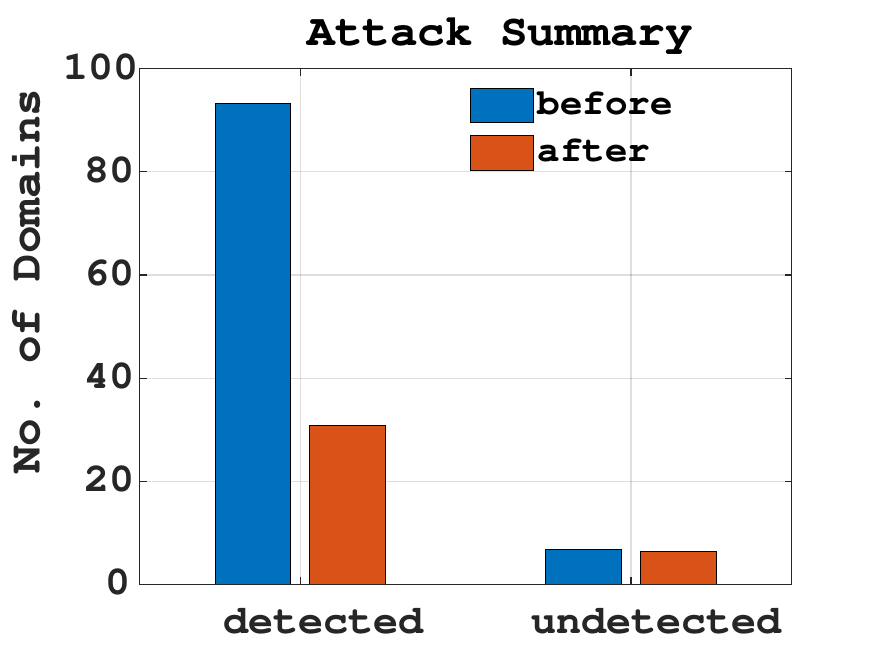}&
 \includegraphics[width=12cm]{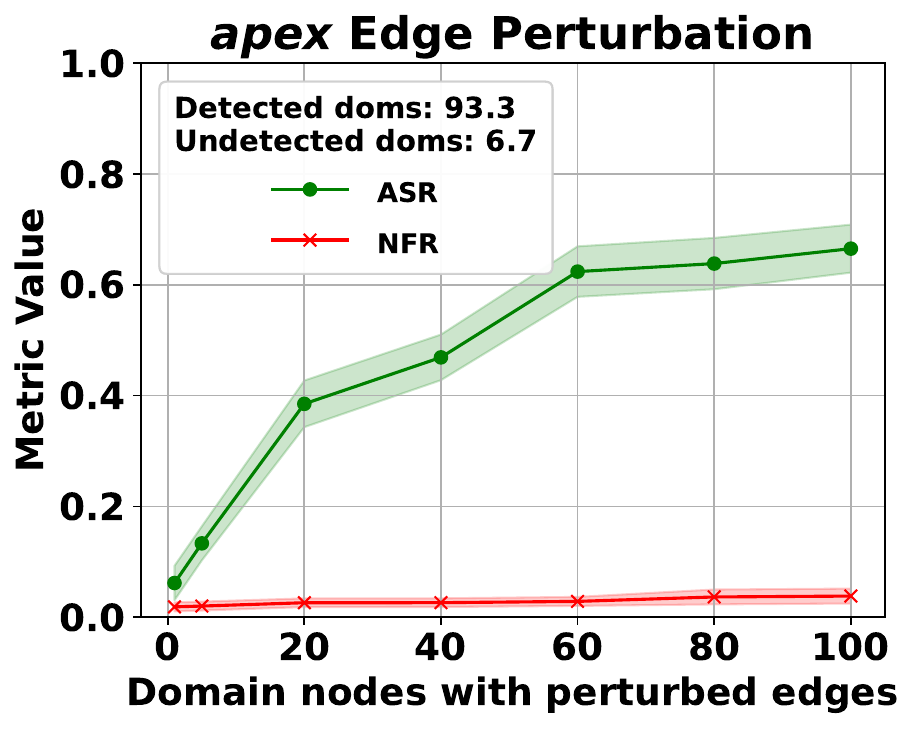}&
\includegraphics[width=12cm]{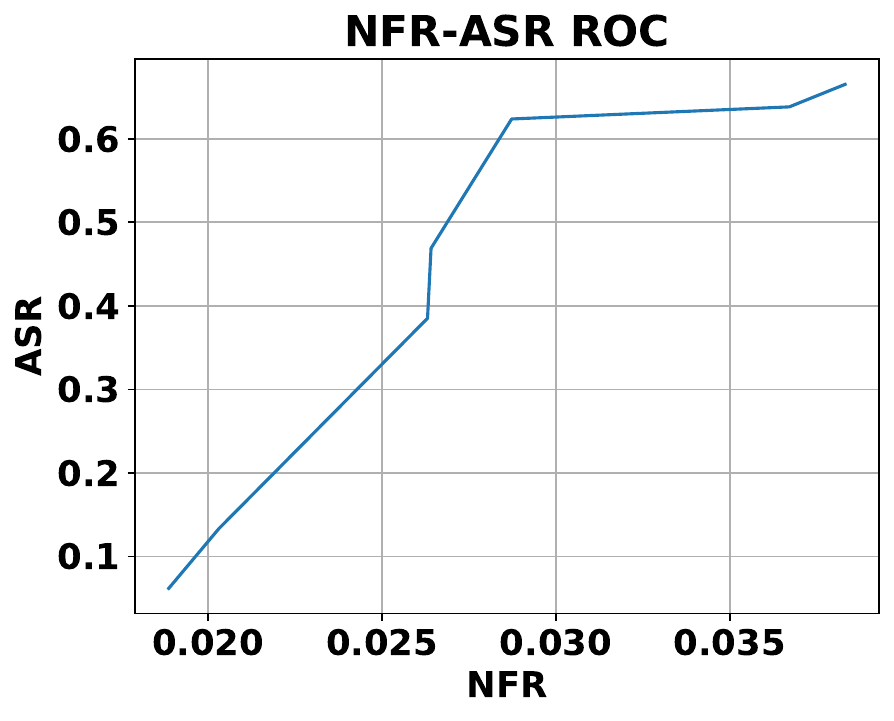}
\\
\Huge{(a)} & \Huge{(b)} & \Huge{(c)} 
\\
\includegraphics[width=12cm]{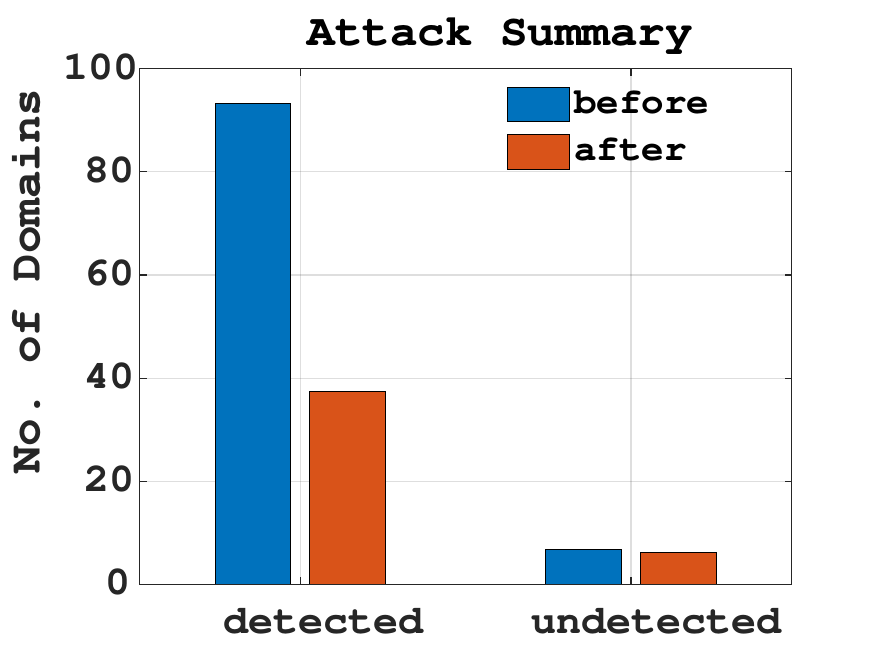}&
\includegraphics[width=12cm]{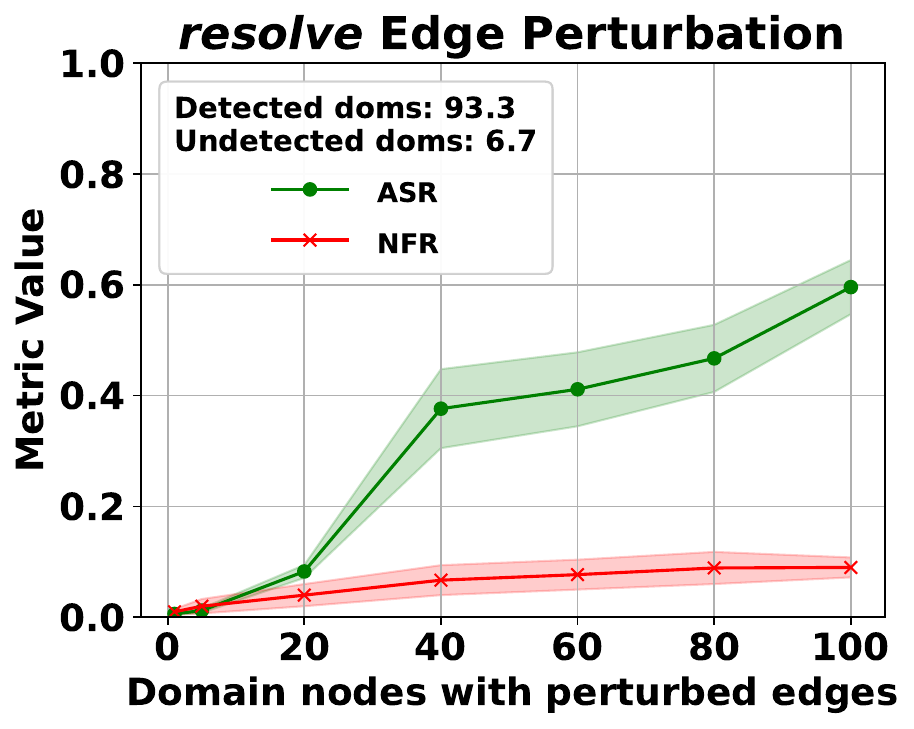}&
\includegraphics[width=12cm]{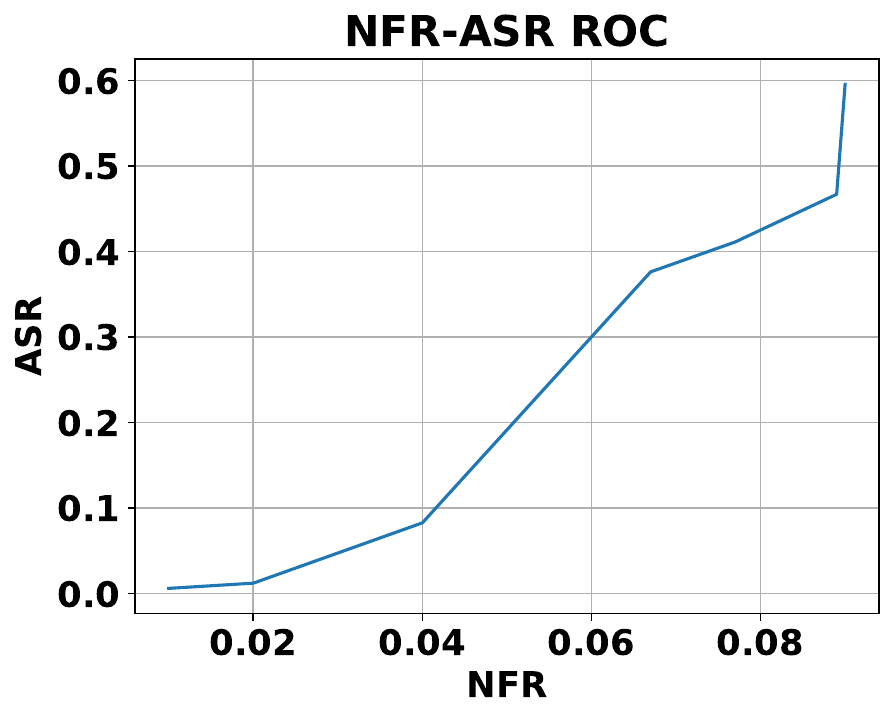}
\\
\Huge{(d)} & \Huge{(e)} & \Huge{(f)}
\\
\includegraphics[width=12cm]{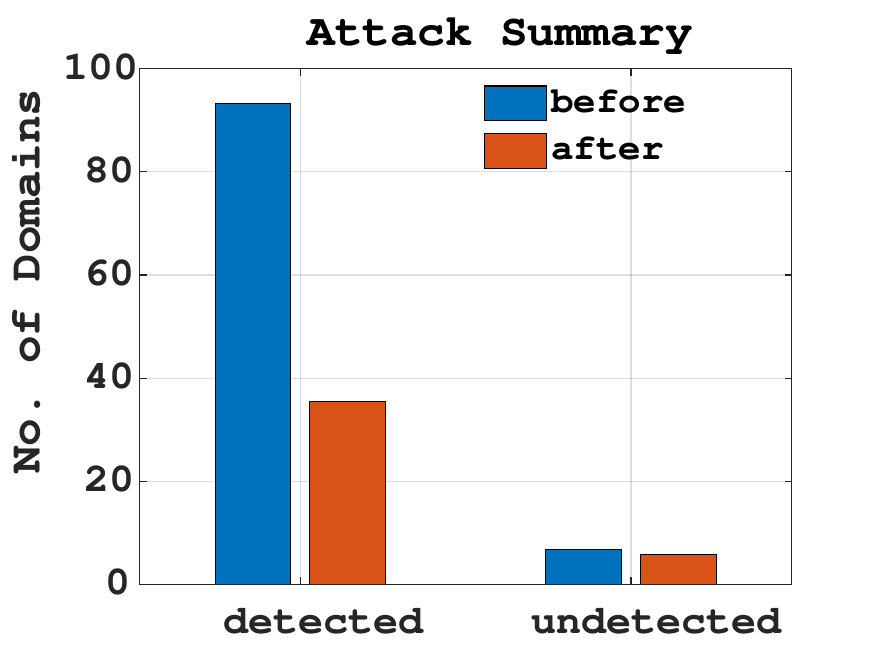}&
 \includegraphics[width=12cm]{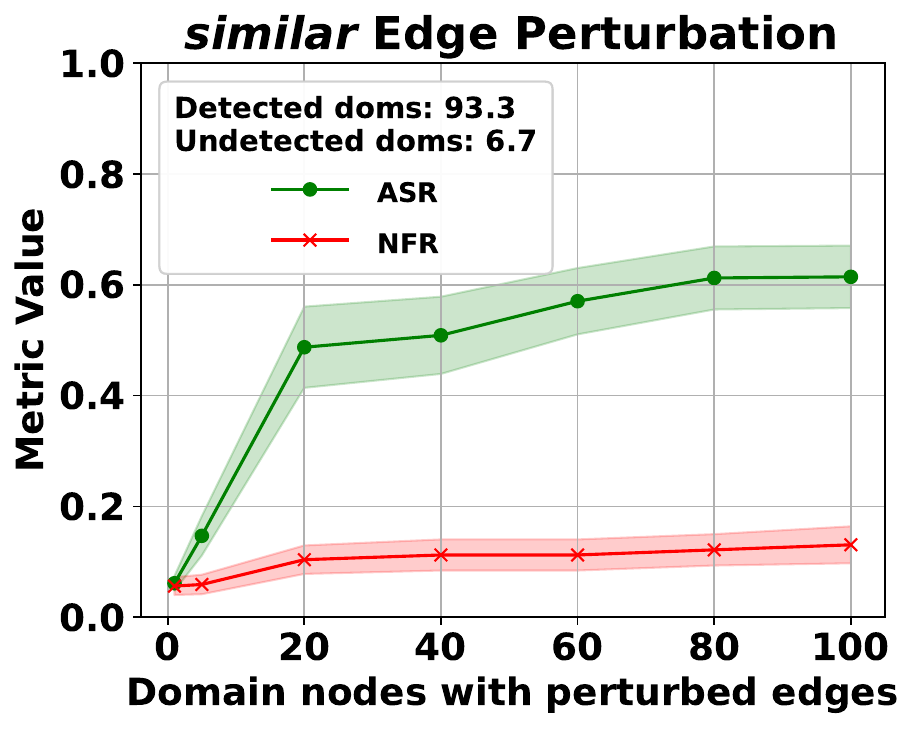}&
\includegraphics[width=12cm]{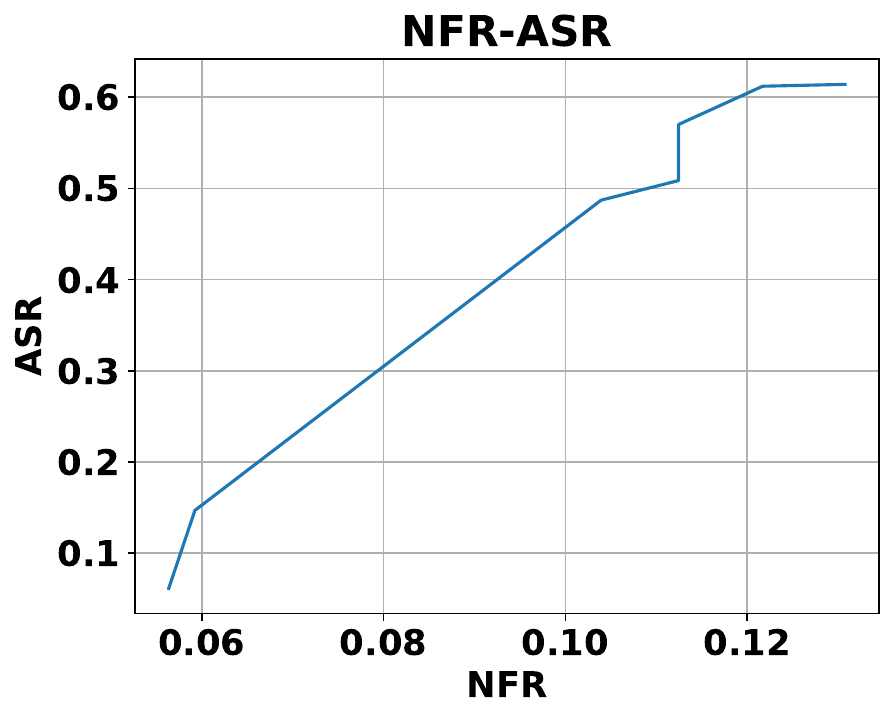}
\\
\Huge{(g)} & \Huge{(h)}& \Huge{(i)}
\end{tabular}}
\caption{With \textit{apex}, \textit{resolve}, and \textit{similar} edge perturbation and \textit{sampled adversary subgraphs}; attack summary when 100 domains are attacked, ASR, NFR, and ROC when less than 100 domains are attacked, in rows 1, 2, and 3, respectively.}
\label{edge_res} 
\end{figure}

\par Next, we repeat the previous experiment with \textit{sampled adversary subgraphs}. The results with \textit{apex} edge perturbation are in the top row of Fig. \ref{edge_res}. For this case, the attack summary is in Fig. \ref{edge_res}(a). It shows that, on average, the attack evades the detection of 62.5 (0.67 $\times$ 93.3) adversary domains out of the 93.3\% detectable adversary domains. The attack results in the detection of around 0.26 (0.04 $\times$ 6.7) domains, on average, from the undetectable domains. Thus, the attack is successful in terms of ASR and NFR. Compared to feature perturbation, this attack has a similar ASR but a much lower NFR. For the cases of attacking less than 100 domains, the ASR and NFR are shown in Fig \ref{edge_res}(b) and the ROC curve is in Fig. \ref{edge_res}(c).

\par The \textit{resolve} edge perturbation results are shown in the second row of Fig \ref{edge_res}. Looking at its attack summary graph in Fig \ref{edge_res}(d), this attack evades the detection of 55.9 (0.6 $\times$ 93.3) adversary domains, on average. The attack results in the detection of around 0.6 (0.09 $\times$ 6.7) domains, on average, from the undetectable domains. Similar to the case with created domains, the \textit{resolve} edge attack is relatively poorer than the \textit{apex} edge attack. Nonetheless, it is not that poor compared to the case of created domains. This is because the adversary has access to more IP addresses since its domains are obtained by \textit{sampled adversary subgraphs} which may \textit{resolve} to more diverse IP addresses. Finally, the \textit{similar} edge perturbation results are shown in the last row of Fig \ref{edge_res}. Considering the attack summary in Fig \ref{edge_res}(g), this attack evades 57.8 (0.62 $\times$ 29.6) domains, on average, while causing the detection of 0.87 (0.13 $\times$ 6.7) domains, on average. The performance here is slightly better than that of the \textit{resolve} edge perturbation and is less competitive than that of the \textit{apex} edge perturbation. For the cases of attacking less than 100 domains, the ASR and NFR are shown in Fig. \ref{edge_res}(h) and the ROC curve is in Fig. \ref{edge_res}(i).

\begin{figure}[htb]
\centering
 \resizebox{0.999\columnwidth}{!}{
\begin{tabular}{ccc}
\includegraphics[width=12cm]{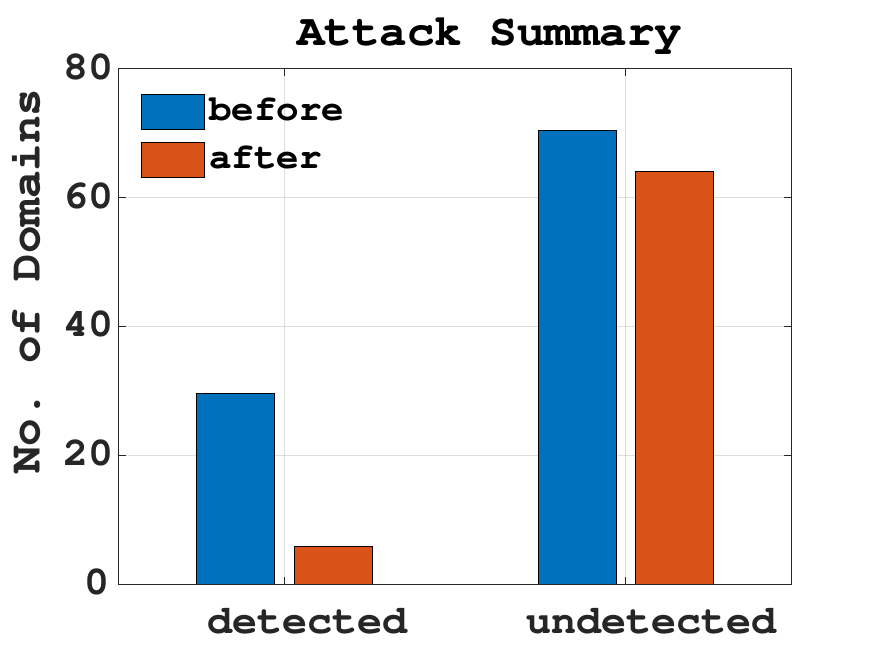}&
 \includegraphics[width=12cm]{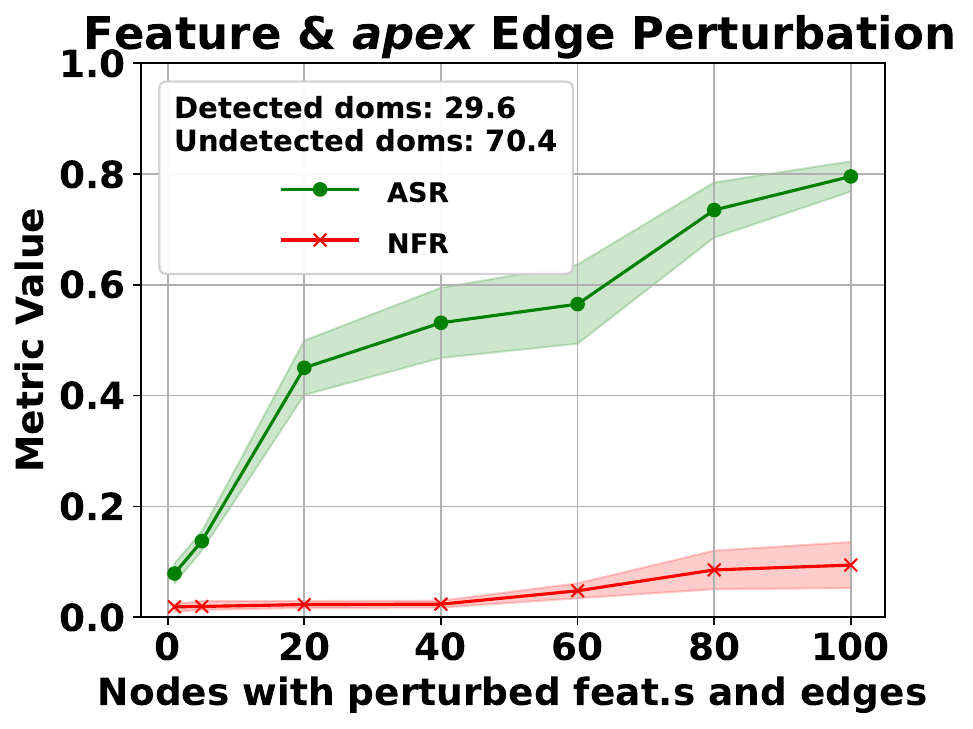}&
\includegraphics[width=12cm]{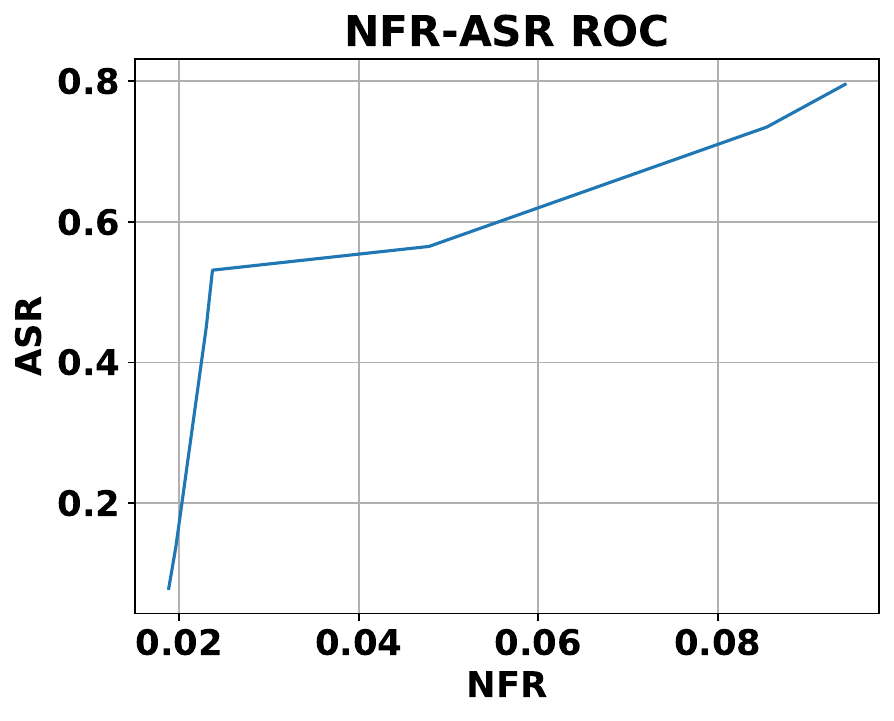}
\\
\Huge{(a)} & \Huge{(b)} & \Huge{(c)}
\\
\includegraphics[width=12cm]{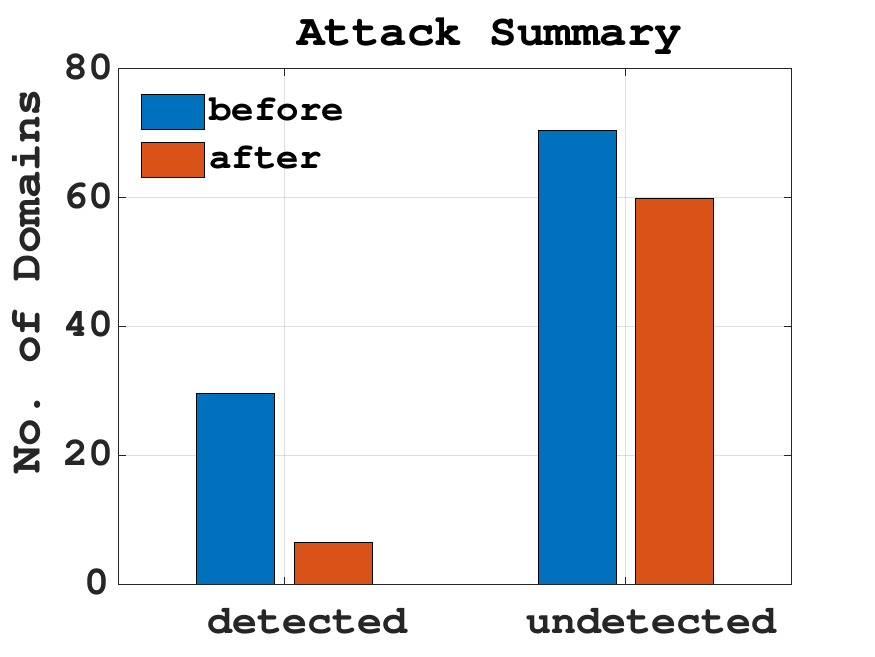}&
 \includegraphics[width=12cm]{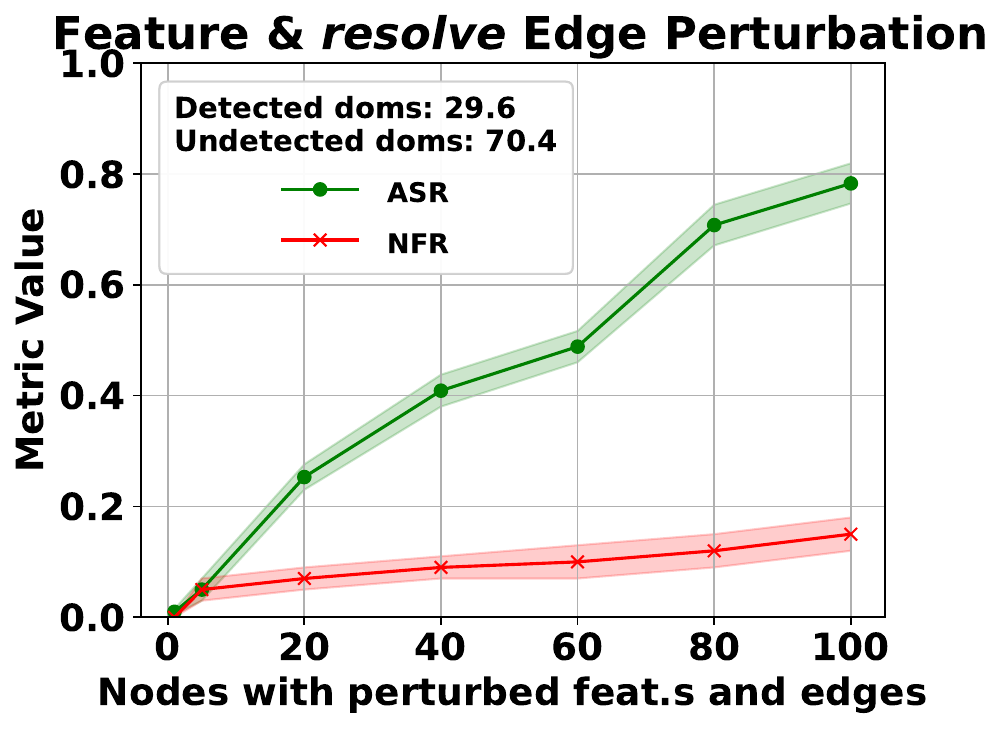}&
\includegraphics[width=12cm]{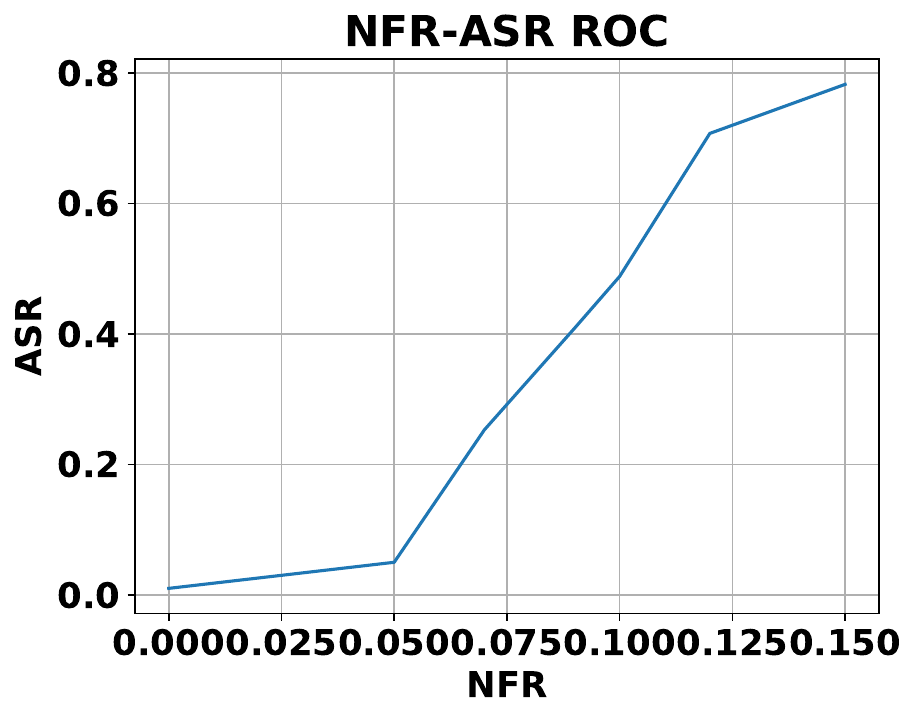}
\\
\Huge{(d)} & \Huge{(e)} & \Huge{(f)} 
\\
\includegraphics[width=12cm]{figs/attack_summay_comb3_mydoms.pdf}&
\includegraphics[width=12cm]{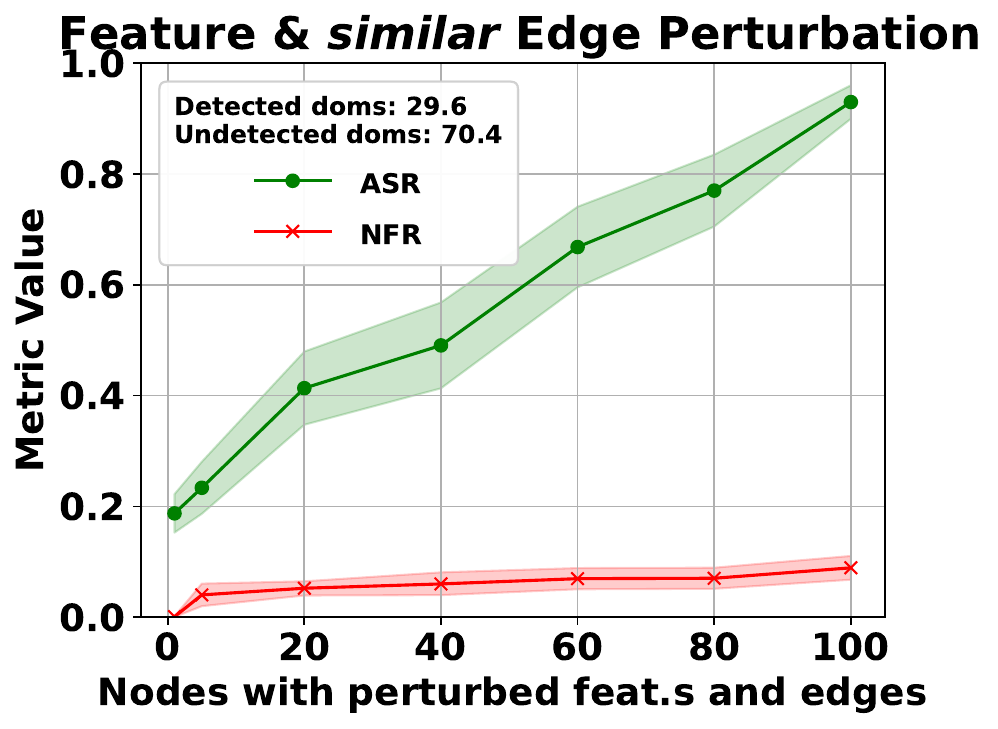}&
\includegraphics[width=12cm]{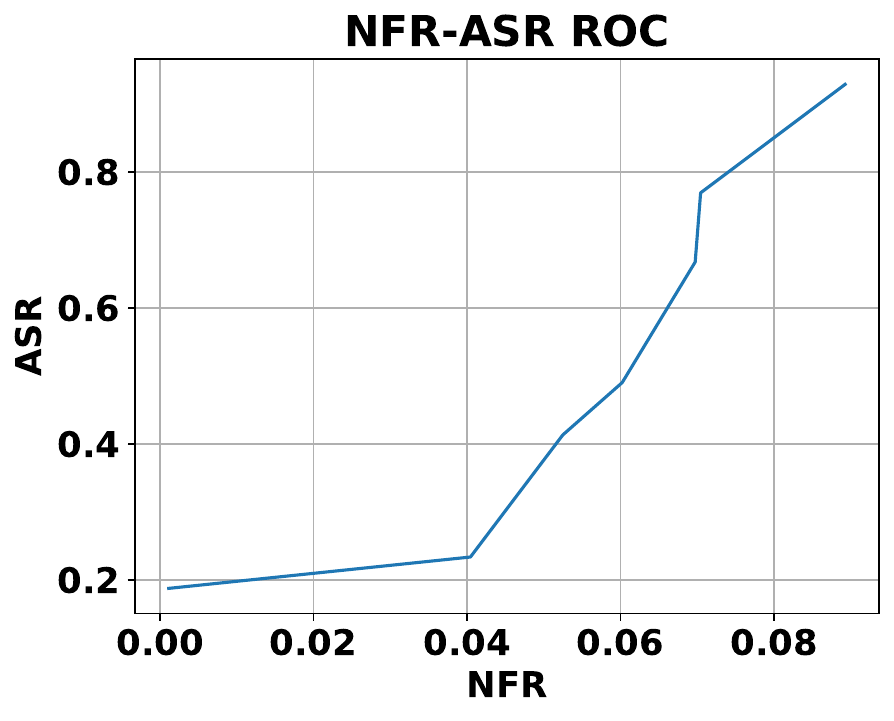}
\\
\Huge{(g)} & \Huge{(h)} & \Huge{(i)}
\end{tabular}}
\caption{With feature-\textit{apex}, feature-\textit{resolve}, and feature-\textit{similar} edge perturbation and \textit{created adversary subgraphs}; attack summary when 100 domains are attacked, ASR, NFR, and ROC when less than 100 domains are attacked, in rows 1, 2, and 3, respectively.}
\label{feature_edges_res_mydoms} 
\end{figure}

\subsubsection{Joint perturbation}
\par In this experiment, we combine feature and edge perturbations (Algorithms \ref{Algorithm1} and \ref{Algorithm2}). First, with created domains, the results with feature-\textit{apex}, feature-\textit{resolve}, and feature-\textit{similar} perturbations are shown in the three rows of Fig. \ref{feature_edges_res_mydoms}, respectively. According to the attack summary in Fig. \ref{feature_edges_res_mydoms}(a), the feature-\textit{apex} attack evades the detection of 23.67 (0.8 $\times$ 29.6) domains, on average, and causes the detection of 6.3 (0.09 $\times$ 70.4) domains, and is thus successful. Let us next consider the attack summary plot of the feature-\textit{resolve} edge in Fig. \ref{feature_edges_res_mydoms}. This attack evades 23.1 (0.78 $\times$ 29.6) domains, on average, and results in the detection of 10.6 (0.15 $\times$ 70.4) domains, on average. Thus, it is less successful compared to the previous attack. However, this attack is stronger than the case of feature and \textit{resolve} edge attacks alone. Finally, the attack summary of feature-\textit{similar} perturbation is in Fig. \ref{feature_edges_res_mydoms}(g). This attack evades the detection of 27.2 (0.92 $\times$ 29.6) domains, on average, and only causes the detection of 6.3 (0.09 $\times$ 70.4) undetectable domains, on average. This attack is the strongest among the combined attacks and is more successful than individually attacking the features or \textit{similar} edges.

\par Finally, we repeat the previous experiment with \textit{sampled adversary subgraphs}. The results for feature-\textit{apex}, feature-\textit{resolve}, and feature-\textit{similar} perturbation attacks are shown in the three rows of Fig. \ref{feat_edge_res}, respectively. The feature-\textit{apex} edge causes the evasion of 76.5 (0.82 $\times$ 93.3) domains on average while causing the detection of 0.53 (0.08 $\times$ 6.7) domains on average. The results of the other two scenarios are similar. In general, the performance with these perturbations is better than individually perturbing the features or edges. These results show the advantage of combining edge and feature perturbations. 

\par The above experiments show that MintA is successful as a framework for adversarial attacks in MDD. This is validated by its ability to evade the detection of high numbers of adversary domains while leading to the detection of fewer undetectable adversary domains. 

\begin{figure}[!tb]
\centering
 \resizebox{0.999\columnwidth}{!}{
\begin{tabular}{ccc}
\includegraphics[width=12cm]{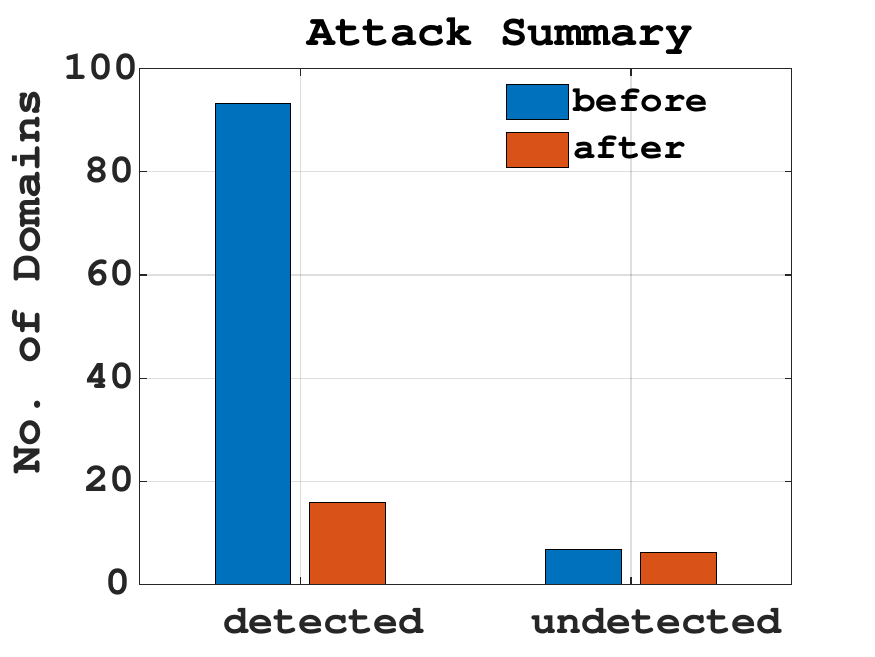}&
 \includegraphics[width=12cm]{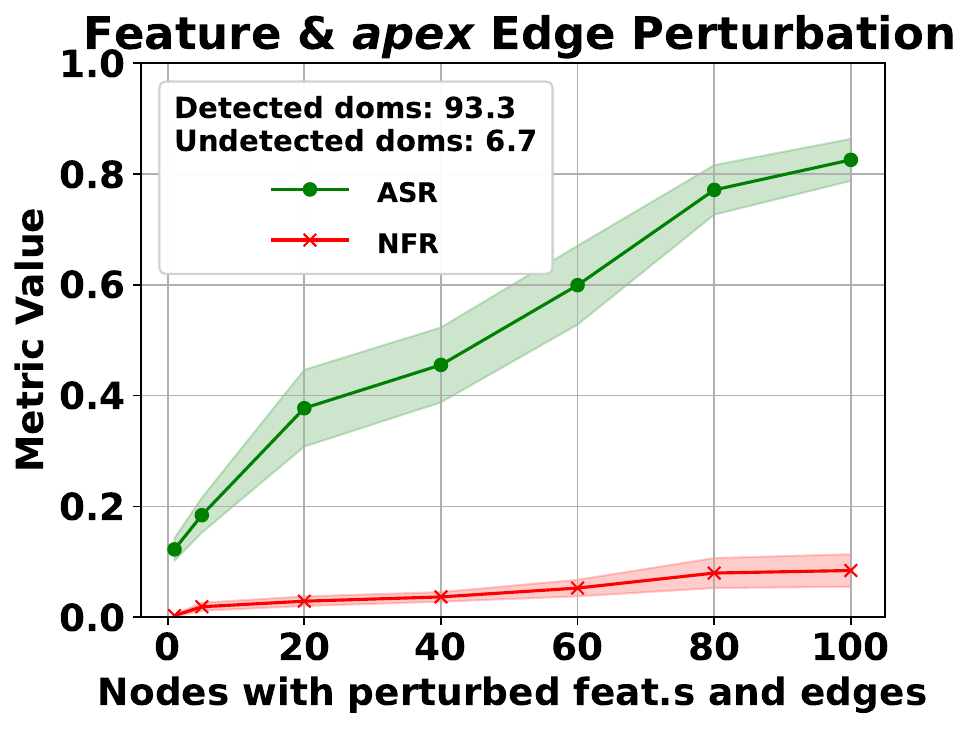}&
\includegraphics[width=12cm]{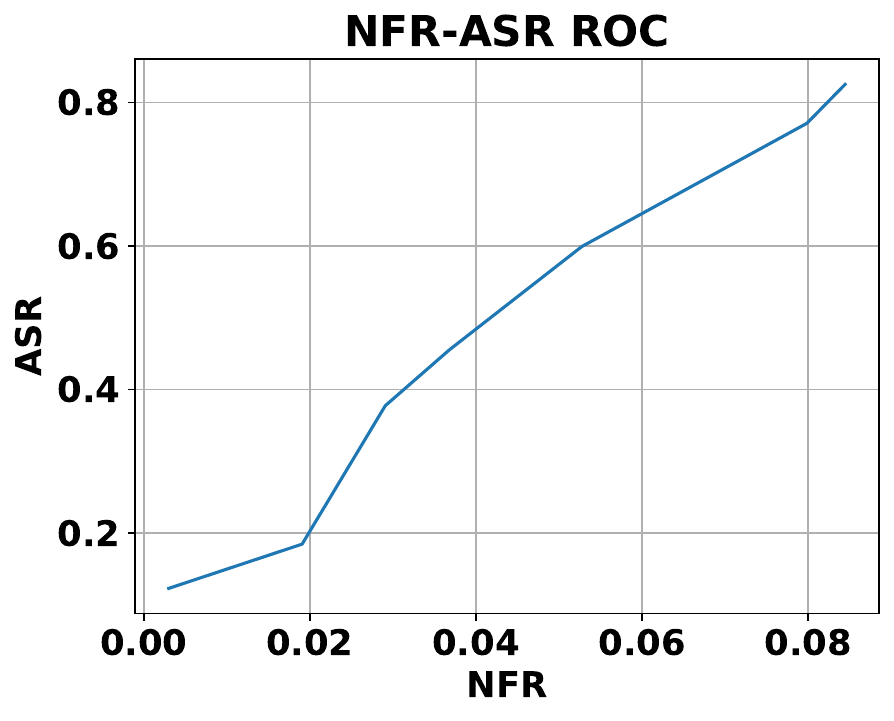}
\\
\Huge{(a)} & \Huge{(b)} & \Huge{(c)}
\\
\includegraphics[width=12cm]{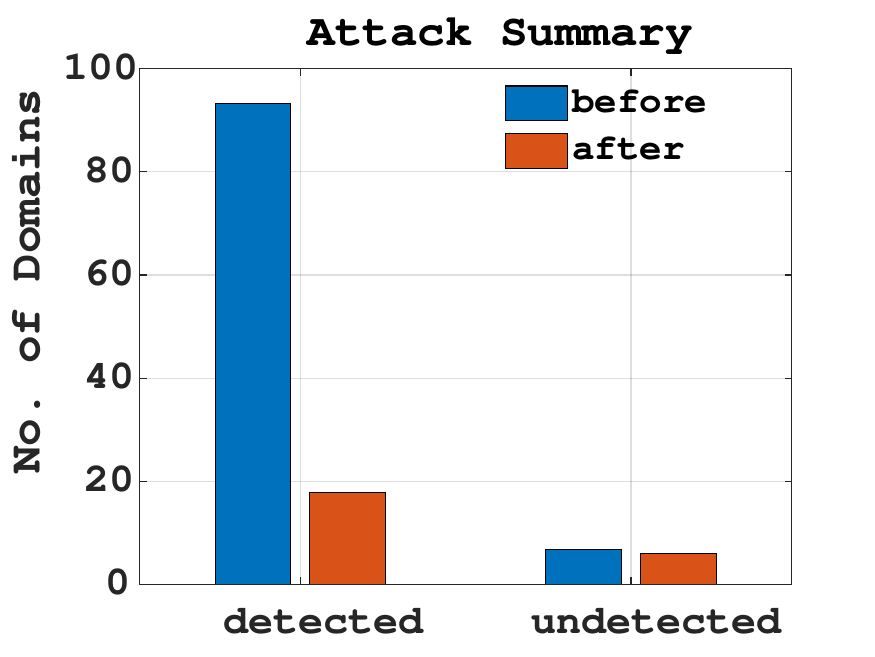}&
 \includegraphics[width=12cm]{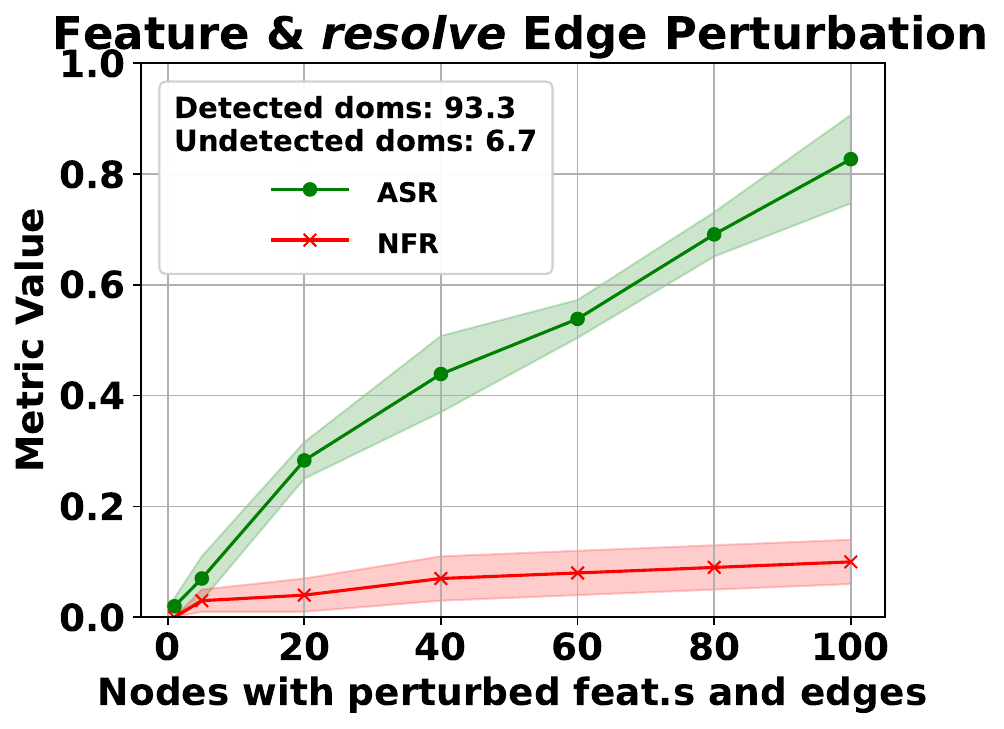}&
\includegraphics[width=12cm]{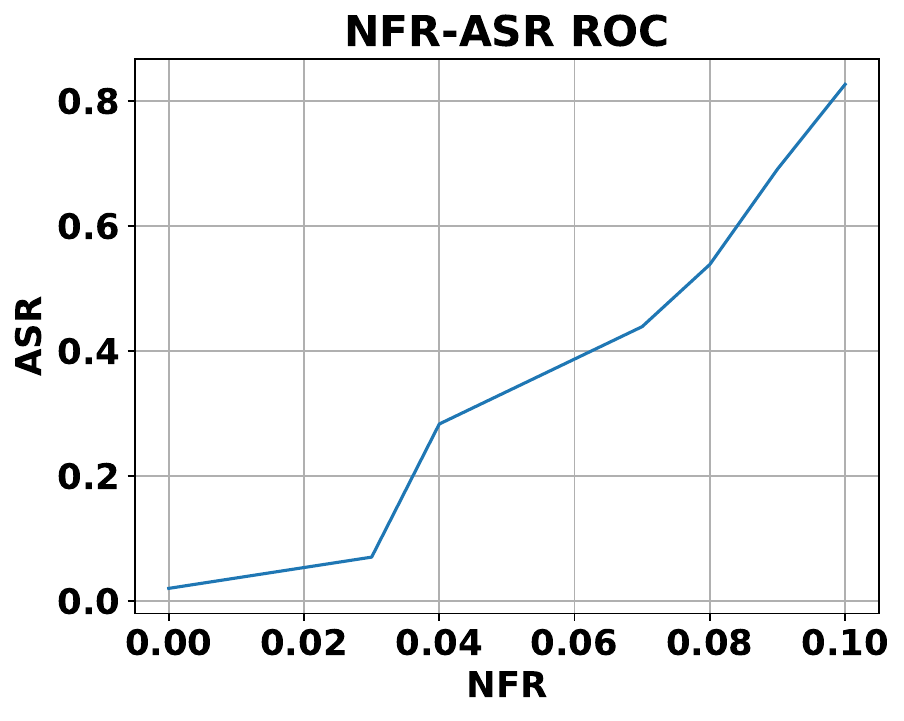}
\\
\Huge{(d)} & \Huge{(e)} & \Huge{(f)} 
\\
\includegraphics[width=12cm]{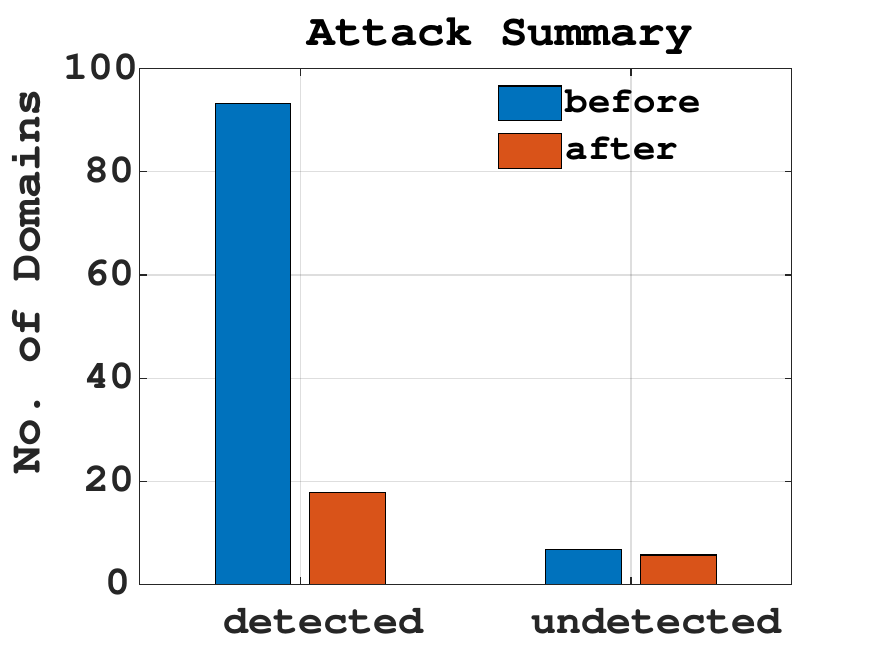}&
\includegraphics[width=12cm]{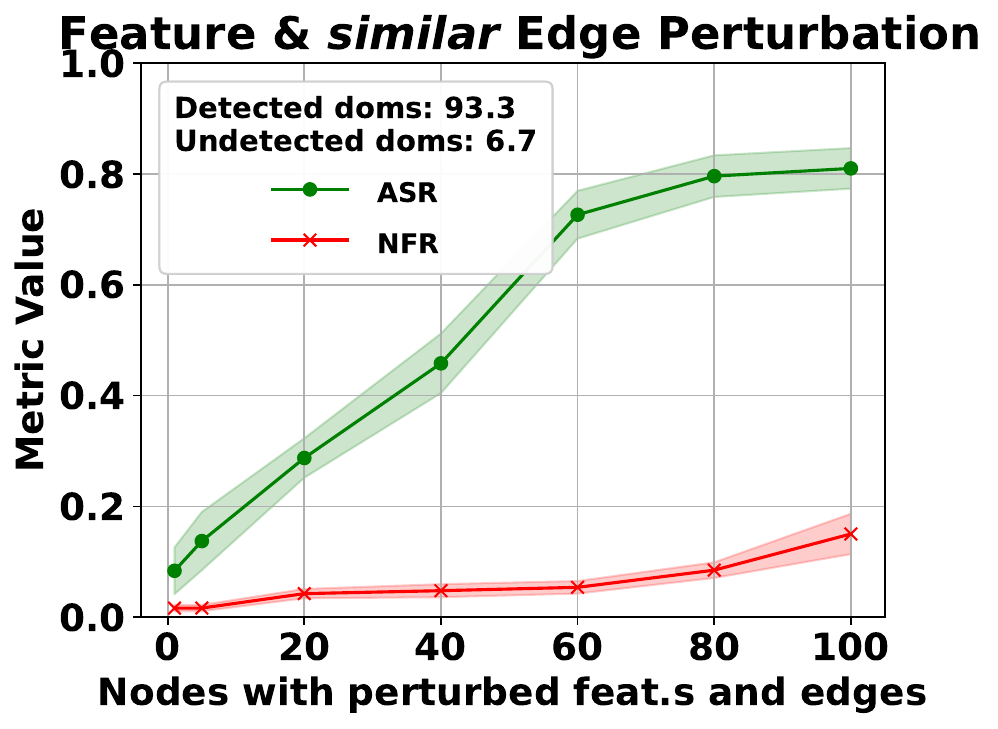}&
\includegraphics[width=12cm]{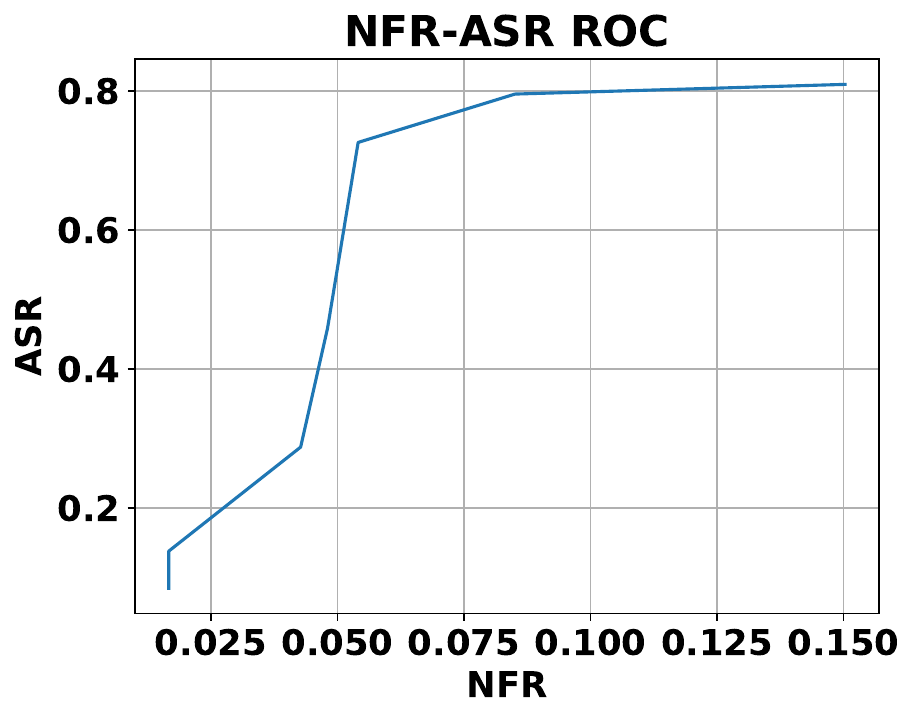}
\\
\Huge{(g)} & \Huge{(h)} &\Huge{(i)}
\end{tabular}}
\caption{With feature-\textit{apex}, feature-\textit{resolve}, and feature-\textit{similar} edge perturbation and \textit{sampled adversary subgraphs}; attack summary when 100 domains are attacked, ASR, NFR, and ROC when less than 100 domains are attacked, in rows 1, 2, and 3, respectively.}
\label{feat_edge_res} 
\end{figure}

\subsection{The impact of MintA on non-adversary nodes}
\par An adversary's subgraph resides in a larger DMG formed by the MDD entity. To address Q2: ``What is MintA's impact on non-adversary nodes?'', we present the following experiment. We consider feature perturbations for the adversary's nodes. For each number of adversary nodes with perturbed features, we monitor the impact on non-adversary nodes residing in the same DMG being direct neighbors to adversary nodes. 
Fig. \ref{feat_dos} presents the results of this experiment. As seen in the figure, both the ASR and NFR impacts of adversary node perturbation on non-adversary nodes increase when attacking more adversary nodes. Still, the attack summary plot shows that the attack results in evading the detection of 29 domains that were detectable before the attack (11\%$\times$266.2=29.38), while it causes around 7 benign domains (10\%$\times$66.7=6.67) to be detected as malicious. The results demonstrate that MintA's impact on non-adversary neighboring nodes is relatively moderate and inconspicuous. In other words, the effect is too subtle to raise suspicions of adversarial activity or compromise its stealthiness.

\begin{figure}[!t]
\centering
 \resizebox{0.999\columnwidth}{!}{
\begin{tabular}{cc}
\includegraphics[width=12cm]{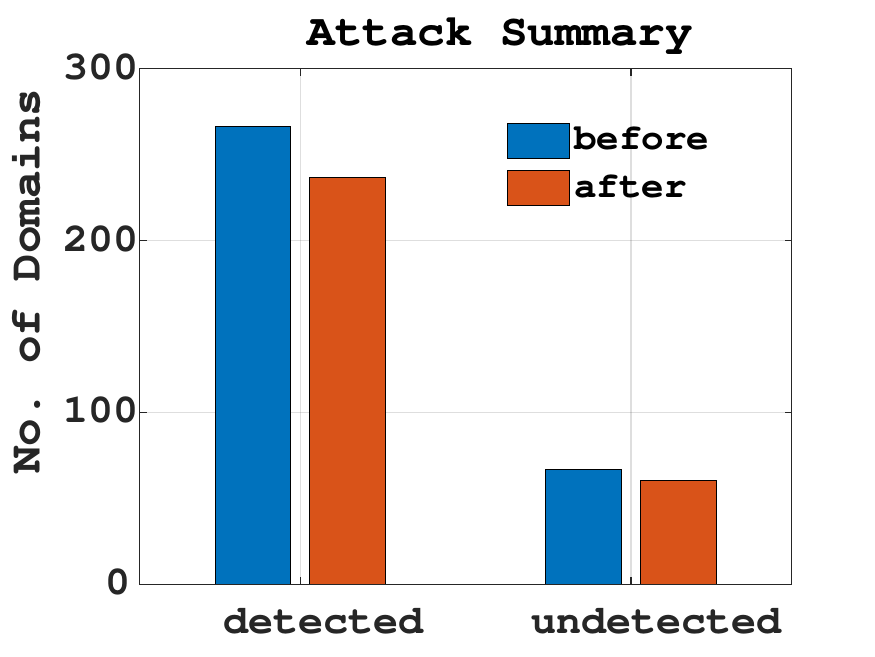}&
 \includegraphics[width=12cm]{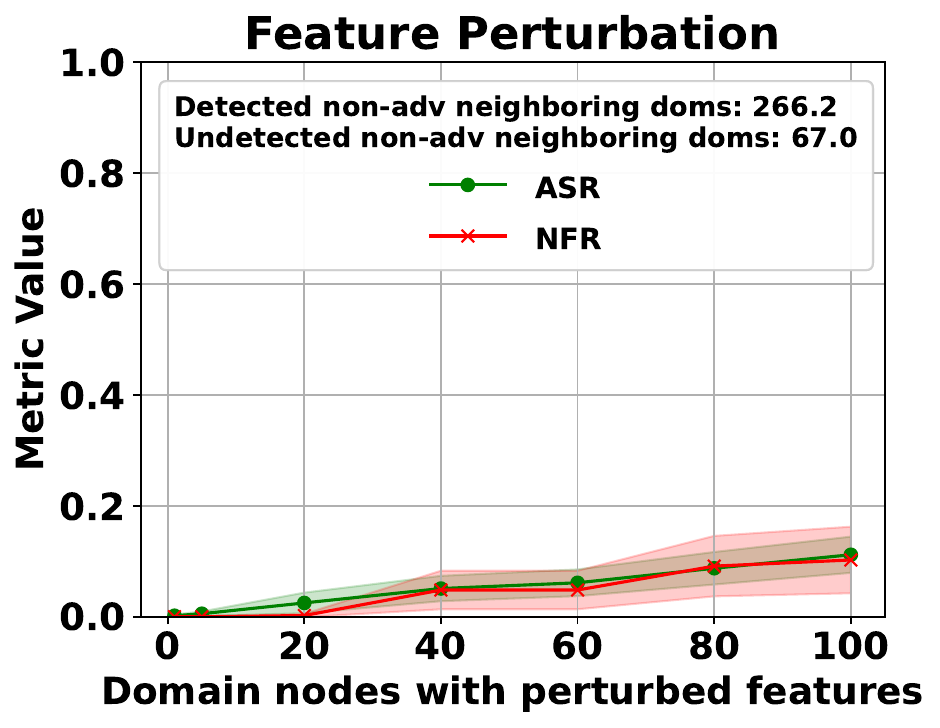}
 \\
\Huge{(a)} & \Huge{(b)}
 \\
\multicolumn{2}{c}{
\includegraphics[width=12cm]{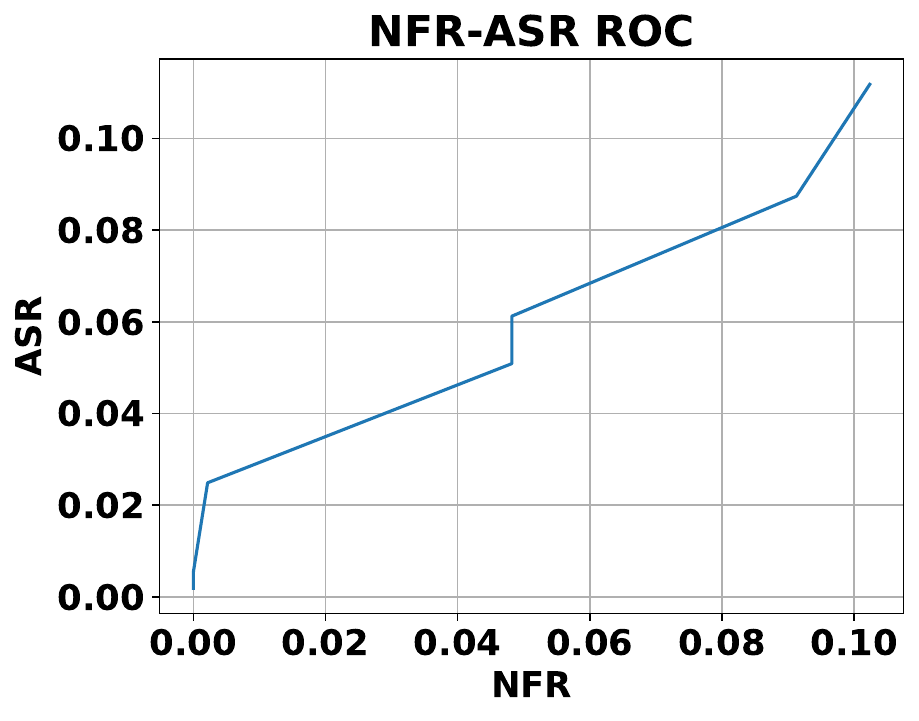}}
\\
\multicolumn{2}{c}{
 \Huge{(c)}}
\end{tabular}}
\caption{The effect of feature perturbations on non-adversary (direct) neighboring domains;
attack summary in (a), and 
ASR, NFR, and ROC in (b) and (c), respectively.}
\label{feat_dos} 
\end{figure}

\subsection{The performance with outlier detection}
\par Recent literature considers the use of simple outlier detection as a basic defense mechanism to detect adversarial examples \cite{paudice2018detection}. It is reasonable to assume that an MDD entity may apply outlier detection as a preliminary defense measure where domains identified as outliers are considered malicious, even before using the MDD model. Hence, it is important to quantify whether MintA causes the attacked nodes to appear as outliers. For this purpose, we address Q3: Can the domains perturbed by MintA be detected as outliers? through the following two experiments. In the first experiment, we sample a total of 4k nodes and randomly sample 100 of these to be the adversary nodes. Next, we perturb the features of all adversary nodes using Algorithm \ref{Algorithm1}. Then, we consider the features of nodes in the whole set as a sample set and apply outlier detection to this set. We use isolation forest \cite{liu2012isolation} as an outlier detection method. We adjust the parameters of outlier detection such that the number of outliers is around 100 nodes. To this end, we quantify how many of the detected outliers are data points (features of nodes) of the adversary nodes. In Fig. \ref{outlier_det}(a), we plot the detection histograms of the overall sample (left) and the adversary nodes (right). It is seen that only 2 of the 100 adversary nodes are detected from the overall 100 outliers in the sample set of data points. This is close to the overall percentage of outliers in the whole test set which is 2.5\% ($=$100/4k). Thus, an outlier detection method does not distinguish the nodes perturbed by MintA from other nodes.

\begin{figure}[htb]
\centering
\resizebox{0.999\columnwidth}{!}{
\begin{tabular}{cc}
 \includegraphics[width=12cm]{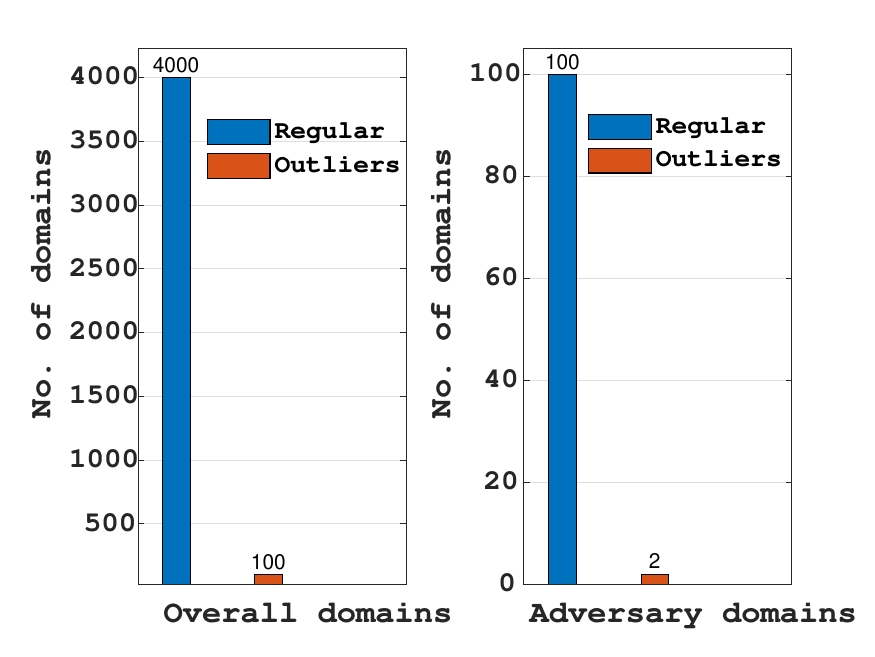}&
\includegraphics[width=12cm]{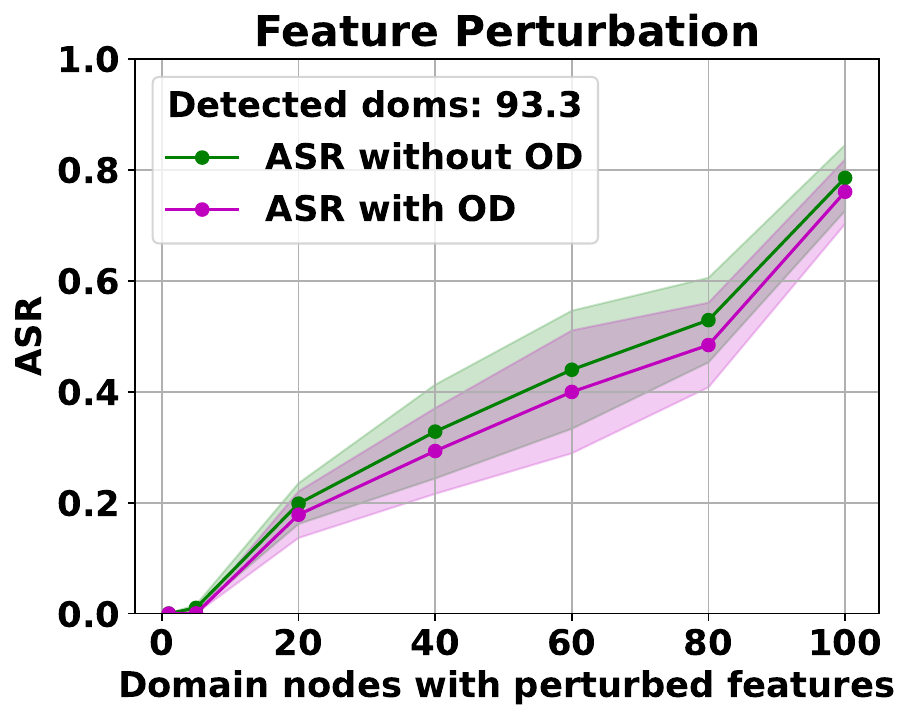}
\\
\Huge{(a)} & \Huge{(b)} 
\end{tabular}}
\caption{(a) Histogram of outliers detected when setting the number of outliers to be 100 from a total set of 4k samples. (b) The ASR of MintA with and without outlier detection.}
\label{outlier_det} 
\end{figure}

\par In the second experiment, We compare the performance of MintA against plain MDD to the case where an MDD entity employs outlier detection and identifies outlier nodes as malicious automatically before the use of the MDD model. Therefore, for this latter scenario, we subtract the ratio of outliers detected in the adversary nodes from the calculated ASR. We also adjust the outlier detection algorithm to produce about 100 outliers. Results are shown in Fig. \ref{outlier_det}(b). In this figure, it is evident that the MintA attack is only marginally affected by outlier detection.

\subsection{Performance with graph purification defense}

\par To examine the performance of MintA when the MDD entity applies graph purification defense, we raise Q4: Can the domains perturbed by MintA bypass
graph purification-based defense? by considering the case of applying the GCN-Jaccard method \cite{wu2019adversarial}. This defense method is based on dropping Jaccard-dissimilar edges in a given graph suspecting their maliciousness. Hence, any domain node in the DMG dropped by this method is considered malicious. Fig. \ref{defense_perf} shows the impact of this graph purification on feature and \textit{apex} edge perturbations in parts (a) and (b), respectively. We only show \textit{apex} perturbation since the \textit{resolve} and \textit{similar} edge perturbations perform similarly. The attack generally bypasses this defense in both scenarios. However, there is more degradation in the edge perturbation case compared to feature perturbation. This is reasonable since graph purification naturally affects the topology of the graph. Thus, it is more harmful to edge perturbations than to feature perturbations. These results validate that MintA is resilient to graph purification-based defenses. 

\begin{figure}[!htb]
\centering
\resizebox{0.999\columnwidth}{!}{
\begin{tabular}{cc}
\includegraphics[width=6cm]{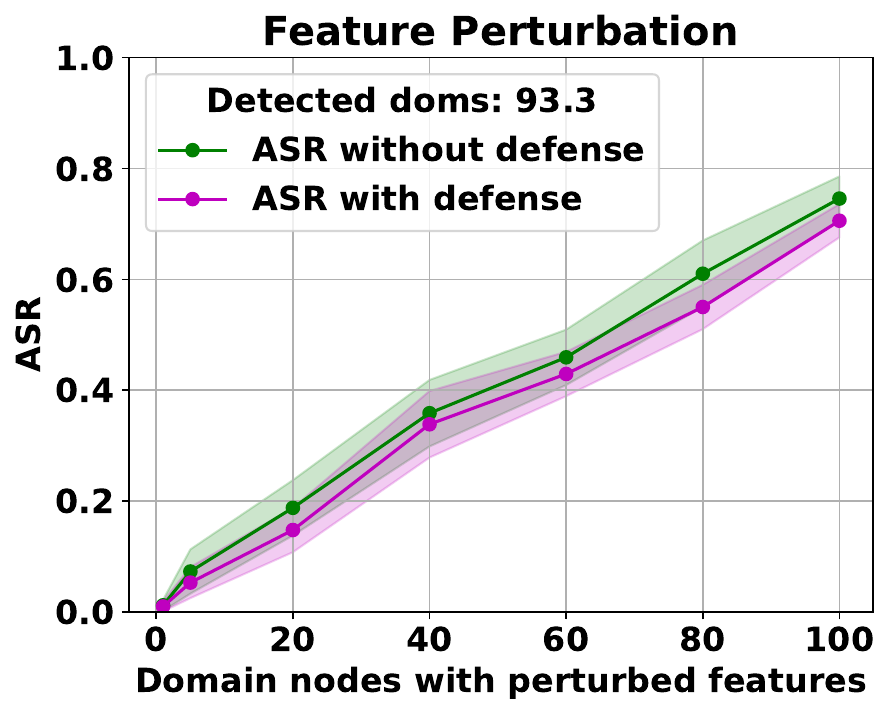}&
\includegraphics[width=6cm]{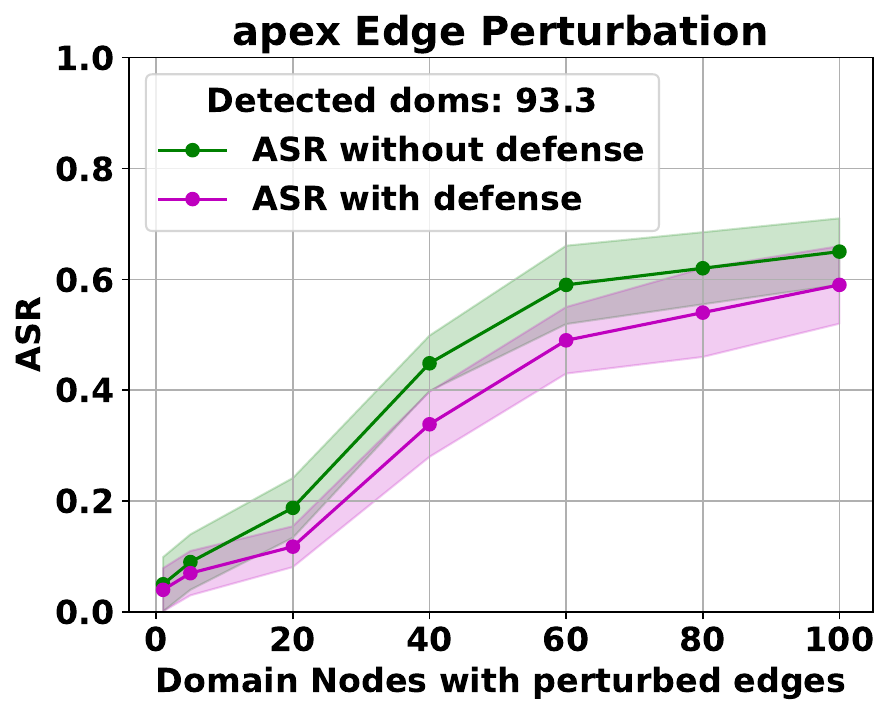}
\\
\Large{(a)}&
\Large{(b)}
\end{tabular}}
\caption{The impact of graph purification defense by the MDD entity on MintA's ASR with feature perturbation in (a) and \textit{apex} edge perturbation in (b).}
\label{defense_perf} 
\end{figure}

\subsection{Comparison with targeted adversarial attacks}
\par Here, we address Q5: What is the performance of MintA compared to a targeted adversarial attack? MintA is the only adversarial attack capable of handling multiple connected nodes, which is precisely its intended purpose. Therefore, we consider the closest baseline for comparison with MintA to be the repeated application of single-instance attacks. We have already presented an experiment on the ineffectiveness of this approach when considering the IG-ADV attack by Wu et al. \cite{wu2019adversarial} in Section \ref{Section4}.2. For completeness, we compare MintA to two examples of single-instance attack methods. It is noted that the ASR used in the experiment in Section \ref{Section4}.2 is the ratio of domains evaded from the attacked domains, whereas, in this section, we use the same ASR definition used in the other experiments, i.e., the ASR is the ratio of evaded domains to the total number of adversary's nodes (100 in our case), for consistency.

\begin{figure}[t]
\centering
\resizebox{0.999\columnwidth}{!}{
\begin{tabular}{cc}
\includegraphics[width=6cm]{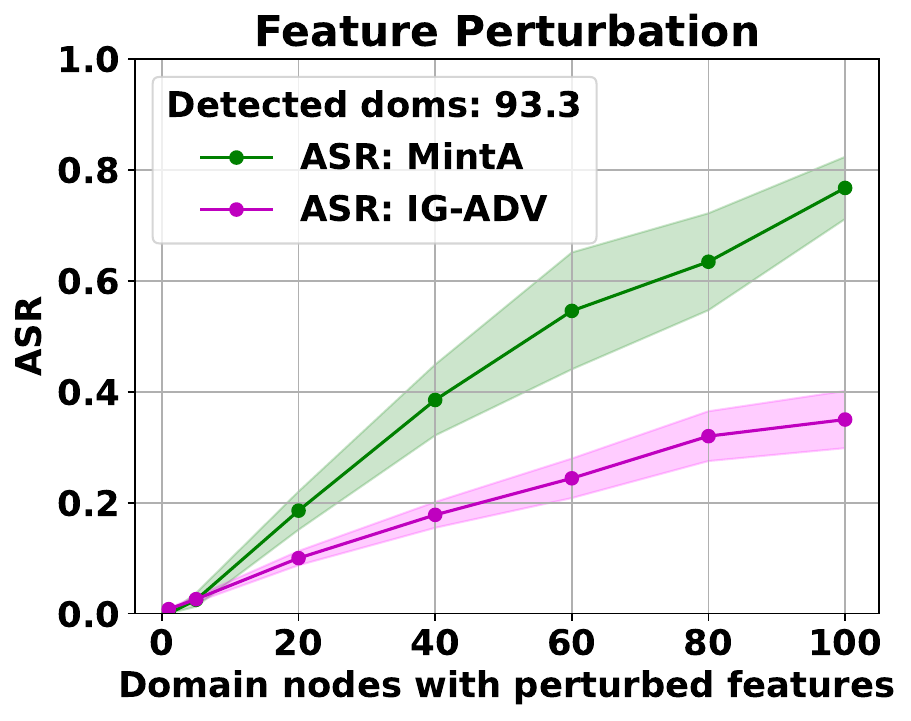}&
\includegraphics[width=6cm]{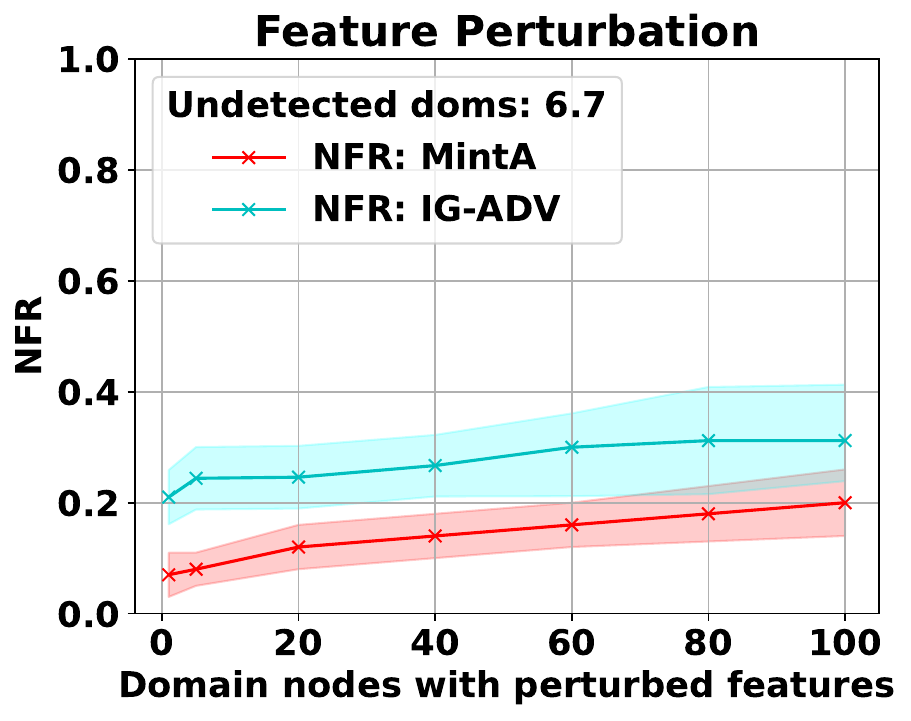}
\\
\Large{(a)}&
\Large{(b)}
\end{tabular}}
\caption{ASR and NFR comparison of MintA with sequentially applying the IG-ADV attack \cite{wu2019adversarial}.}
\label{sota_comp} 
\end{figure}

\par First, we present an experiment on feature perturbation. Fig. \ref{sota_comp} compares the ASR and NFR of MintA and IG-ADV attack \cite{wu2019adversarial} where IG-ADV is applied individually on each node. IG-ADV node attack starts to have increases in ASR with increasing the number of nodes attacked. However, as more nodes are attacked, the increase in ASR diminishes. This is due to the aggregate effect of attacking multiple nodes without coordination as discussed in Section \ref{Section4}.2. As for NFR, the NFR with IG-ADV starts with a relatively high value compared to MintA. Next, both exhibit increases in NFR as more nodes are attacked. Still, the NFR values of MintA are less than those of IG-ADV.

\par In another experiment, we compare MintA's edge perturbation to that of the projected gradient descent topology attack (PGD-ADV) by Xu et al. \cite{xu2019topology}, as a baseline attack method. Fig. \ref{sota_comp2} shows the results. This figure shows that the MintA attack is successful on the adversary nodes, whereas PGD-ADV is less effective. Also, the NFR of PGD-ADV is higher than that of MintA and the difference increases with an increased number of attacked nodes. These results show the advantage of MintA over targeted attacks focusing on a specific node, as MintA alleviates the negative effect of node perturbations on each other's evasiveness, as exhibited with these attacks. 

\begin{figure}[thb]
\centering
\resizebox{0.999\columnwidth}{!}{
\begin{tabular}{cc}
\includegraphics[width=6cm]{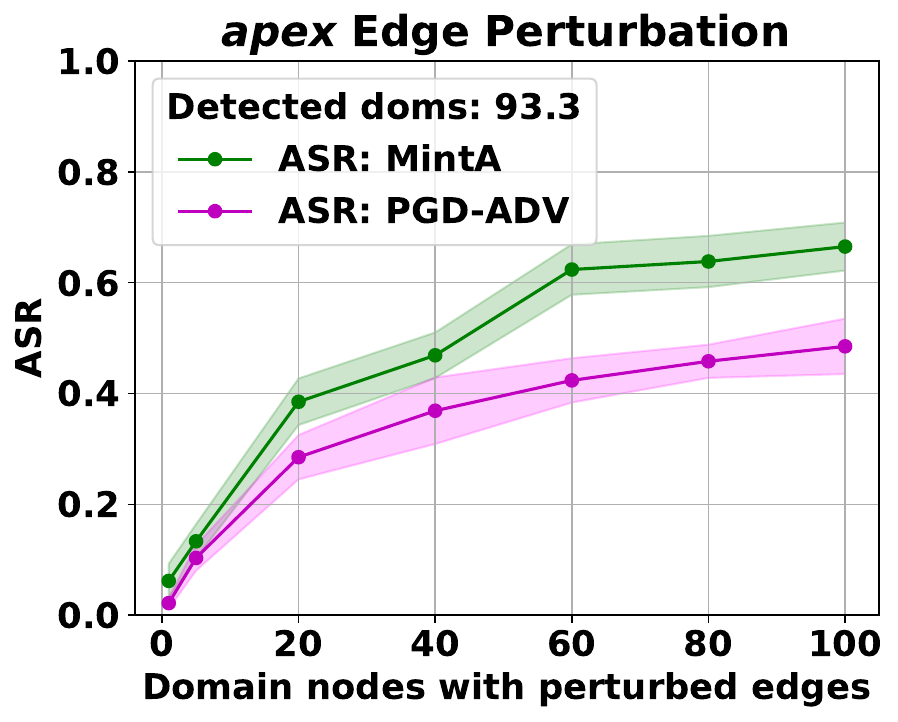}&
\includegraphics[width=6cm]{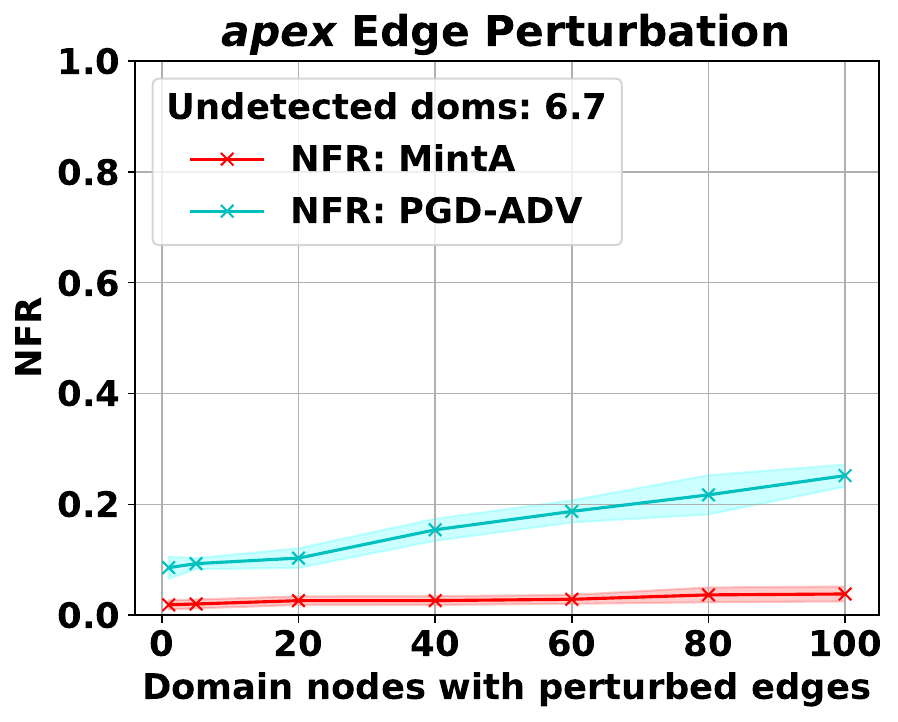}
\\
\Large{(a)}&
\Large{(b)}
\end{tabular}}
\caption{ASR and NFR comparison of MintA and sequentially applying PGD-ADV attack \cite{xu2019topology}.}
\label{sota_comp2} 
\end{figure}

\subsection{Empirical cost analysis}
\par MintA entails primarily the following cost-effective expenses:
\textbf{Preparation cost}: This involves querying the target model to label training data for the surrogate model, and in our case, we utilized 600 queries. Querying is typically facilitated through free API or website services. Surrogate training can be conducted using a standard or cloud computing platform.
\textbf{Perturbation optimization cost}: This cost depends on the number of adversary nodes and can be handled by any computer or cloud server.
\textbf{Perturbations implementation cost}: Implementing perturbations requires human effort in modifying domain name and/or IP resolution, which involves contacting domain registrants and hosting service providers. There may be potential fees depending on the subscription plans. To this end, to address Q6, which pertains to the time and memory costs of MintA, we measure the execution time (in seconds) and memory usage (in megabytes) for the processes of surrogate training and MintA optimization (Algorithm \ref{Algorithm1}). We quantify memory usage using the \textit{psutil} Python package and conduct 30 trials, varying the number of adversary nodes. The empirical costs are summarized in Table \ref{empiricalcosts}. As indicated, the time and memory costs for MintA's preparation and optimization are relatively low and can be afforded by any regular or cloud computing platform. The optimization time cost scales linearly with the number of adversary nodes, while the other costs remain relatively stable. These results demonstrate the scalability of MintA.

\begin{table}[htbp]
\centering
\caption{Average execution time (in second) for preparation ($t_{p}$) and optimization ($t_o$), and the corresponding memory usages (in Mb) ($m_p$), and ($m_o$), respectively.}
\resizebox{0.85\columnwidth}{!}{
\begin{tabular}{|c|c|c|c|c|}
\hline
No. of Adv. nodes & $t_p$ & $t_o$ & $m_p$ & $m_o$ \\ \hline
50 & 2.3 & 4.2 & 62.7 & 58.5\\ \hline
100 & 2.3 & 22.1 & 61.0 & 64.7\\ \hline
200 & 2.2 & 28.4 & 61.1 & 64.7\\ \hline
500 & 2.2 & 94.7 & 61.0 & 64.9\\ \hline
1000 & 2.2 & 284.8 & 61.1 & 64.7\\ \hline
\end{tabular}
}
\label{empiricalcosts}
\end{table}

\par \subsection{Additional experiments}
\par In addition, we explore the effects of the following factors on the attack: the percentage of adversary nodes in a DMG, the complexity of the surrogate model utilized, and white-box model access. Detailed information and results of these experiments can be found in Appendix \ref{appendixB}.

\section{Related Work}
\label{Section6}
\par \textbf{Adversarial attacks on GNNs.} The vulnerability of GNN models to adversarial attacks is well-studied \cite{xu2020adversarial}. This includes their operation in node classification \cite{ma2019attacking}, graph classification \cite{zhang2021projective}, and link prediction \cite{bhardwaj2021poisoning} in various applications. However, the vulnerability of GNNs in MDD is not studied. Only attacks on node classification can be relevant to MDD as a node-label inference operation. Based on the attack scope, attacks are either targeted, i.e., focusing on a given target node, or untargeted (availability attacks) aiming to degrade the performance of the target model. Examples of untargeted attacks include using the gradient \cite{xu2019topology} and reinforcement learning \cite{ma2019attacking} to add/remove edges. More recently, Ma et al. \cite{ma2020towards,ma2022adversarial} consider the problem of selecting the best nodes for crafting an attack in a given graph. These works address the node selection process, not optimizing the perturbation itself. On the other hand, examples of targeted attacks include reinforcement learning-based perturbation \cite{dai2018adversarial}, absolute gradient FGA \cite{chen2018fast}, integrated gradients \cite{wu2019adversarial}, and greedy perturbation selection for maximizing the loss of a surrogate model as in Nettack \cite{zugner2018adversarial}. More recently, the concept of anchor nodes has arisen. Those are important nodes in a graph such that flipping their connectivity to a given target node is sufficient to flip the model’s outcome for this node. Anchor nodes are either found by searching for them in a given graph \cite{zang2020graph}, or by creating and adding them to the graph \cite{dai2022targeted,zang2023guap}. The novelty of MintA lies in handling simultaneous evasion of multiple connected nodes in a coordinated manner. Besides, using MintA to attack other GNN operations that use node features and graph edges can be a future extension. 

\par Other works consider indirect adversarial attacks, where an adversary manipulates the neighbors of a given target node to \textit{remotely} influence it \cite{zugner2018adversarial,chen2021graphattacker,takahashi2019indirect}. Nettack introduces the concept of influencer attacks \cite{zugner2018adversarial}, while \cite{chen2021graphattacker} employs a generative approach to manipulate multiple nodes connected to the target node. Also, \cite{takahashi2019indirect} proposes an influencer attack that operates at a 2-hop distance from the target node, extending Nettack's approach. However, these works do not consider simultaneous attacks on a set of multiple connected nodes and overlook the significance of multi-instance attacks \cite{takahashi2019indirect}.

\section{Discussion}
\label{Section7}
\par \textbf{Impact of the study.} This study establishes a significant security threat against one of the most prominent MDD approaches, which has garnered attention from both academia \cite{sun2019hindom,li2021dydom,li2022heterogeneous,sun2020deepdom,zhang2021attributed} and industry \cite{nabeel2022brand,khalil2020method,nabeel2021method,tirumala2022scoring}. Moreover, it is worth noting that commercialized MDD tools like DomainTools with Maltego \cite{domaintools-maltego} explicitly mention adopting a graph-based approach but do not disclose their exact graph inference technique for proprietary reasons. In our experiments, we specifically focus on the approach presented in \cite{sun2020deepdom} as a representative case. However, it is essential to emphasize that MintA is applicable to any other GNN-based MDD method utilizing DNS logs. We demonstrate that the costs associated with the preparation, optimization, and implementation of MintA are reasonable and manageable under mild conditions. 

\par \textbf{Limitations of the study.} The primary limitations of this study arise from several factors. Firstly, the scarcity of MDD datasets is a significant constraint due to concerns regarding potential security vulnerabilities found in enterprise data available through DNS logs. Additionally, the usage of patented GNN-based approaches within the industry restricted our ability to explore a broader range of comparison options.
Another limitation stems from the absence of applicable adaptive graph purification defenses for hetGNNs, which could have enhanced the robustness of the MDD approach. Furthermore, the absence of actual adversary subgraphs necessitated the adoption of sampled and created adversary modeling approaches. Although these approaches were designed to mimic adversary behaviors, utilizing actual domains associated with a real adversary would have provided a more accurate representation of adversary subgraphs during the experiments.

\par \textbf{Possible countermeasures and future work.} Future research can explore different avenues to enhance the attack and defense methods in the context of state-of-the-art MDD approaches that leverage heterogeneous graphs and hetGNNs. One potential direction is to harness the heterogeneity of the DMGs to further improve the effectiveness of the attack.
Conversely, another important extension is to develop efficient defense mechanisms to counter MintA's vulnerability that has been demonstrated. A promising direction for defense is to integrate and complement the knowledge gained from DNS logs to bolster the robustness of MDD, given that malicious actors can access these logs. It is worth noting that, despite MintA's stealthiness against outlier detection, exploring approaches that leverage additional information from DNS logs could enhance the defense against such attacks.

\section{Conclusions}
\label{Section8}
\par In this research, we investigate the vulnerability of current Graph Neural Network (GNN)-based Malicious Domain Detection (MDD) approaches to evasive adversarial attacks during inference. We demonstrate how adversaries can practically exploit the reliance of GNN-based MDD on DNS logs to carefully modify their domain nodes, effectively altering the resulting graphs constructed by MDD entities.
To optimize these modifications for effective evasion, we propose a novel multi-instance adversarial attack called MintA. This attack is particularly suited for practical adversaries who own multiple interconnected domains, considering cost and stealthiness factors. Unlike existing adversarial attacks, MintA aims to simultaneously evade detection for as many adversary domain nodes as possible.

Our experiments reveal that MintA achieves remarkable success rates, exceeding 80\%, when targeting a state-of-the-art GNN-based MDD algorithm with real-world data. Also, we demonstrate that MintA can bypass defense methods based on outlier detection and graph purification, underscoring its effectiveness in evading detection mechanisms.

\section*{Acknowledgment}
This work is partially supported by grants NSF IIS-2041096 and NSF CNS-1935928.

\bibliography{references_security.bib}
\bibliographystyle{IEEEtran}

\appendices
\section{Analysis and Proofs}
\label{appendixA}
\subsection{Analysis of the need for a coordinated subgraph attack}
\par Let us recall the formulations of the two objective functions $F_1$ and $F_2$ in (\ref{eq2}) and (\ref{eq3}), respectively. We can optimize $\Delta \boldsymbol{X'}_i$ to maximize $F_1$ and $F_2$ using Corollary \ref{corollary}.
\begin{cor}
\label{corollary}
The quantity
$\mathrm{F}_1=\|\boldsymbol{B} \Delta \boldsymbol{X'}_i \boldsymbol{W}\|_2^2$
is maximized by maximizing the sum of inner products $\langle\Delta \boldsymbol{X'}_i, \boldsymbol{W}_k \rangle, ~\forall~k \in \{1, \ldots, K\} $ where $\boldsymbol{W} \in \mathbb{R}^{n\times K}$.
\end{cor}
\begin{proof}
\par Le us make use of the fact that $\Delta \boldsymbol{X'}_i$ only changes the feature matrix at the $i$-th row (corresponding to node $i$). Let us now quantify the effect of this change with the help of matrix form representation of the formulations of $F_1$ and $F_2$, as follows. 
\begin{multline}\label{eq14}
\mathrm{F}_1=\|\boldsymbol{B} \Delta \boldsymbol{X'}_i \boldsymbol{W}\|_2^2=
\\
\left\|\left[
\begin{array}{cccc}
B_{11} & \mkern-10mu B_{12} & \mkern-10mu \ldots & \mkern-10mu B_{1n} \\
B_{21} & \mkern-10mu B_{22} & \mkern-10mu \ldots & \mkern-10mu B_{2n} \\
\vdots & \mkern-10mu \ddots & \mkern-10mu \ldots & \mkern-10mu \vdots \\
B_{n1} & \mkern-10mu \ldots & \mkern-10mu \ldots & \mkern-10mu B_{nn}
\end{array}
\right]\left[
\begin{array}{c}
\mspace{-12mu} 0 \mspace{-12mu} \\
\mspace{-12mu} 0 \mspace{-12mu} \\
\mspace{-12mu} \Delta \boldsymbol{X'}_i \mspace{-12mu} \\
\mspace{-12mu} 0 \mspace{-12mu}
\end{array}
\right]\left[
\begin{array}{@{} ccc @{}}
\boldsymbol{W}_1 & \mspace{-12mu} \ldots & \mspace{-12mu} \boldsymbol{W}_n
\end{array}
\right] \right\|_2^2,
\end{multline}
\noindent where $B_{i,j}$ is the element in $\boldsymbol{B}$ at the $i$-th row and $j$-th column, and $\boldsymbol{W}_i$ is the $i$-th column in the model parameter matrix $\boldsymbol{W}$. With a few algebraic simplification steps, we can write.
\begin{align}\label{eq15}
F_1=\left\|\left[\begin{array}{llll}
B_{1 i} \langle \Delta \boldsymbol{X'}_i , \boldsymbol{W}_1 \rangle \ldots B_{1 i} \langle\Delta \boldsymbol{X'}_i, \boldsymbol{W}_n \rangle \\
B_{2 i} \langle\Delta \boldsymbol{X'}_i, \boldsymbol{W}_1 \rangle \ldots B_{2 i} \langle\Delta \boldsymbol{X'}_i, \boldsymbol{W}_n \rangle \\
~~~~~~~~~\vdots~~~~~~~~~~\ddots ~~~~~~~~~\vdots \\
B_{n i} \langle\Delta \boldsymbol{X'}_i, \boldsymbol{W}_1 \rangle \ldots B_{n i} \langle\Delta \boldsymbol{X'}_i, \boldsymbol{W}_n \rangle
\end{array}\right]\right\|_2^2.
\end{align}

\par It is clear now that maximizing $F_1$ requires maximizing the sum of inner products $ \langle\Delta \boldsymbol{X'}_i, \boldsymbol{W}_k \rangle, ~\forall~k$.
\end{proof}

\par Now, the effect of $\Delta \boldsymbol{X'}_i$ on the message at node $j \in \mathcal{N}(i)$ is
\begin{equation}\label{eq16}
\Delta \boldsymbol{X'}_j=\left[B_{j i} \langle\Delta \boldsymbol{X'}_i, \boldsymbol{W}_1 \rangle~\ldots \quad B_{j i} \langle\Delta \boldsymbol{X'}_i, \boldsymbol{W}_n \rangle\right].
\end{equation}
This will be added to the aggregate message at the node. $j\left(\boldsymbol{H'}_j\right)$. So, the loss at node $j$ will be changed as follows.
\begin{equation}\label{eq17}
\Delta \mathcal{L}_j\left(\boldsymbol{A'}, \boldsymbol{X'}+\Delta \boldsymbol{X'}_i ; \boldsymbol{W}, j\right)=\|\boldsymbol{B}\left[\Delta \boldsymbol{X'}_i+\boldsymbol{H'}_j\right] \boldsymbol{W}\|_2^2. 
\end{equation}

\par Following Corollary \ref{corollary} and (\ref{eq17}) for maximizing $\mathcal{L}_j ~\forall~j~\in~\mathcal{N}(i)$, or equivalently, maximizing $F_2$, this requires maximizing the inner products $ \langle\Delta \boldsymbol{X'}_i, \boldsymbol{W}_k-\boldsymbol{H'}_j \rangle,~\forall~k \in \{1, \ldots, K\},~j~\in~\mathcal{N}(i)$. It is clear that optimizing $\Delta \boldsymbol{X'}_i$ only considering $F_1$ will not meet maximizing $F_2$. This proves that adversarial perturbations applied at node $i$ are not optimized to serve the evasion of node $j$.

\subsection{The proof of Theorem \ref{theorem2}}
\begin{proof}
\par Introduce a Lagrange multiplier $\lambda$ to the formulation in (\ref{eq7}), for simplicity, let us use $\boldsymbol{X'}$ to denote $\Delta \boldsymbol{X'}_i$.
\begin{equation*}
F=\|\boldsymbol{\Phi} \boldsymbol{X'}\|^2-\lambda \sum_j \boldsymbol{X'}_j^2-\epsilon=0. 
\end{equation*}
\noindent Differentiating both sides:\\
$
\begin{aligned}
&\mspace{-2mu} \frac{\partial F}{\partial \boldsymbol{X'}_k} = \mspace{-4mu} \frac{\partial}{\partial \boldsymbol{X'}_{\mathrm{k}}} \sum_j\left(\sum_i \boldsymbol{\Phi}_{i j} \boldsymbol{X'}_i\right)^2\mspace{-4mu} -\lambda \frac{\partial}{\partial \boldsymbol{X'}_{\mathrm{k}}} \sum_j \boldsymbol{X'}_j^2 = 0,
\\
&\frac{\partial}{\partial \boldsymbol{X'}_{\mathrm{k}}} \sum_j\left(\sum_i \boldsymbol{\Phi}_{i j} \boldsymbol{X'}_i\right)^2=\lambda \frac{\partial}{\partial \boldsymbol{X'}_k} \sum_j \boldsymbol{X'}_j^2, 
\\
& \sum_j \boldsymbol{\Phi}_{k j} \sum_i \boldsymbol{\Phi}_{i j} \boldsymbol{X'}_i=\lambda \boldsymbol{X'}_k. 
\end{aligned}
$
\par A solution is the eigenvector of $\boldsymbol{\Phi}^T \boldsymbol{\Phi}$, since 
\begin{equation*}
\left\|\boldsymbol{\Phi} \boldsymbol{X'}^*\right\|^2=\lambda^2\left\|\boldsymbol{X'}^{* 2}\right\|. 
\end{equation*}
\noindent Then, we select the eigenvector having the largest eigenvalue of $\boldsymbol{\Phi}^T \boldsymbol{\Phi}$ since this will maximize the objective function.
\end{proof}

\subsection{The proof of Proposition \ref{prop1}}

\begin{proof}
\par consider an adversary node $i$, connected to $j$ adversary nodes and a non-adversary node $l$. The adversary knows $i$ and $j$ and is not aware of $l$. Let us compare the contribution of the perturbation applied at node $i$ to the message collected at node $l$ in cases where the optimization is done by maximizing the loss on node $i$, and the proposed case where the perturbation is done by maximizing a weighted average of the loss at node $i$ and its direct adversary neighbors $j \in \mathcal{N}(i)$, in a coordinated fashion.
\par For the case of perturbation optimization by maximizing the loss at node $i$, from (\ref{eq2}), and Corollary \ref{corollary} in Appendix \ref{appendixA}.1, an optimal $\Delta \boldsymbol{X'}_i$ must maximize its inner products with the columns $\boldsymbol{W}_1, \ldots, \boldsymbol{W}_n$ in the model parameter matrix $\boldsymbol{W}$. Next, we showed in Theorem \ref{theorem2} that this solution is the principal eigenvector of the matrix $\boldsymbol{W}^T\boldsymbol{W}$. This solution is denoted by $\boldsymbol{e}$ in equation (\ref{eq6}). Let us use the notation $\boldsymbol{e}(\boldsymbol{W})$ to denote the solution. The magnitude of the effect of the perturbation at $i$ received at node $l$ can be written as follows.
\begin{equation}
 \|m_{single}\|=\frac{1}{d_l} \|\boldsymbol{e}(\boldsymbol{W})\|,
\end{equation}
\noindent where $d_l$ is the degree of node $l$.
\par In the proposed MintA attack case, the magnitude of message contribution of the perturbation is.
\begin{equation}
 \|m_{proposed}\|=
 \frac{1}{d_l} \|\underset{\Delta \boldsymbol{X'}_i}{\alpha\operatorname{argmax}} ~ F_1\|
 +
 \frac{1}{d_l} \|\underset{\Delta \boldsymbol{X'}_i}{\beta \operatorname{argmax}}~F_2\|.
\end{equation}
\par This can be written as: 
\begin{equation}
 \|m_{proposed}\|=
 \|(\alpha\boldsymbol{e}(\boldsymbol{W})
 +
 \beta\sum_i\frac{1}{d_j}\boldsymbol{e}(\boldsymbol{W}-\boldsymbol{H'}_j)\|.
\end{equation}
\noindent Applying the triangle inequality,
\begin{equation}
 \|m_{proposed}\|\leq
 \alpha\|(\boldsymbol{e}(\boldsymbol{W})\|
 +
 \beta\|\sum_i\frac{1}{d_j}\boldsymbol{e}(\boldsymbol{W}-\boldsymbol{H'}_j)\|.
\end{equation}

\par Now, let us evaluate the ratio of $m_{proposed}$ to $m_{single}$.
\begin{small}
\begin{multline}\label{eq_last}
 \frac{ \|m_{proposed}\| }{ \|m_{single}\| }\leq
 \frac{\alpha\|\boldsymbol{e}(\boldsymbol{W})\|+\beta\|\sum_j\frac{1}{d_j}\boldsymbol{e}(\boldsymbol{W}-\boldsymbol{H'}_j)\|}{\|\boldsymbol{e}(\boldsymbol{W})\|}
 \\
 =
 \alpha+
 \beta \frac{\|\sum_j\frac{1}{d_j}\boldsymbol{e}(\boldsymbol{W}-\boldsymbol{H'}_j)\|}{\|\boldsymbol{e}(\boldsymbol{W})\|}. 
\end{multline}
\end{small}
\noindent The numerator in the quotient in (\ref{eq_last}) is the magnitude of the summation of multiple vectors of arbitrary directions. For these vectors to add up in magnitude, their directions need to be aligned. This means that the eigenvectors of the matrices 
$(\boldsymbol{W}-\boldsymbol{H'}_j))^{T}(\boldsymbol{W}-\boldsymbol{H'}_j))~$ need to align $\forall~j$. This requires these matrices to be scalar multiples of each other. This condition can be met only if the messages $(\boldsymbol{H'}_j)$ are equal. This requirement can not be met since different nodes can have different messages. Now, let us assume that the perturbation $\Delta \boldsymbol{X'}_i$ has the same perturbation magnitude in both cases (single node attack and the proposed MintA), and the solutions $\boldsymbol{e}$ are unit-norm eigenvectors to be scaled by the perturbation budget, and let us assume that nodes $i$ and $j$ have the same degree. Then, the maximum value of the quotient in (\ref{eq_last}) is $\alpha +\beta$ which is 1 since $\alpha$ and $\beta$ are weighted average coefficients. Thus, we can write
\begin{equation}
\|m_{proposed}\| / \|m_{single}\| <
 \alpha +\beta= 1. 
\end{equation}
Therefore, $\|m_{proposed}\|$ is strictly less than $\|m_{single}\|$.
\end{proof}

\section{Supplementary experiments and results}
\label{appendixB}

\par \textbf{Impact of the adversary subgraph's share from the DMG.} The experiments in Section \ref{Section5} consider an adversary owning 100 domains and a DMG of 4000 nodes. To further evaluate MintA, we experiment with varying numbers of adversary nodes (from 50 to 1000). For each case, the ASR and NFR are recorded when all adversary nodes are attacked, averaged over 30 trials, and presented in Fig. \ref{ASR_NFR_adv_scale}.

\begin{figure}[!bht]
\centering
\resizebox{0.46\columnwidth}{!}{
\includegraphics{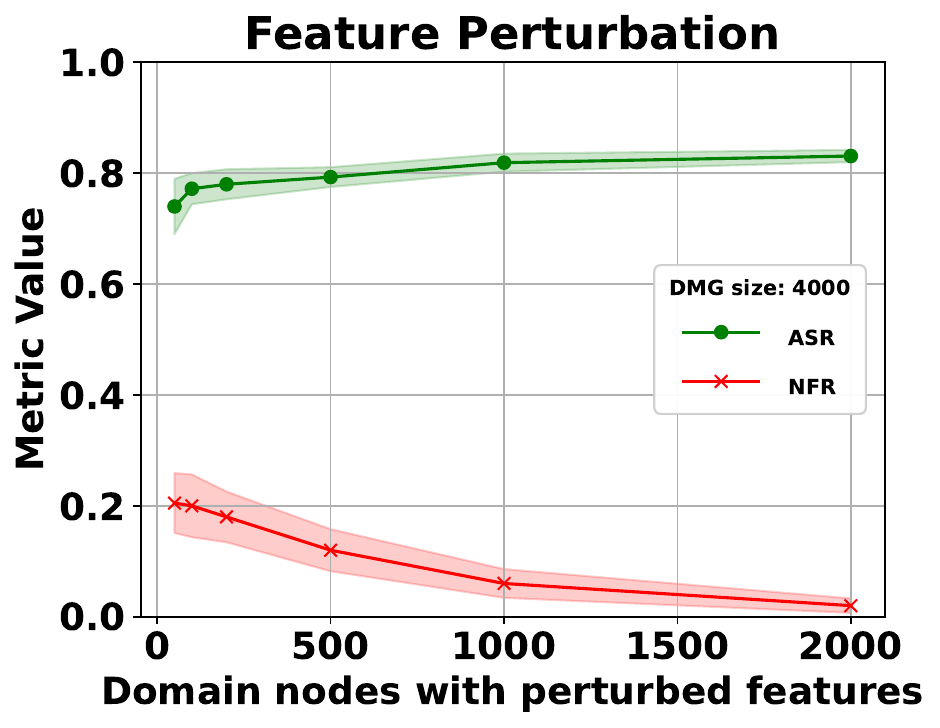}
}
\caption{Average ASR and NFR versus the adversary's share on a DMG when all adversary nodes are involved in the attack.}
\label{ASR_NFR_adv_scale} 
\end{figure}

\par From Fig. \ref{ASR_NFR_adv_scale}, in general, MintA is more successful with increasing its node share in the DMG, in terms of both ASR and NFR. However, the ASR increase and NFR decrease are not commensurate with the increase in the adversary's share in the DMG. We attribute this behavior to the impact of messages fed to the adversary's nodes from non-adversary nodes (unknown to the adversary). This impact dramatically decreases as more nodes in the DMG belong to the adversary. This means that the effect of ignoring non-adversary messages decreases. 
This will continue till the limiting case where the adversary possesses the entire DMG. However, it is more practically sound to assume an adversary with a small portion of the DMG as it is not possible to intervene with the DMG construction at the MDD entity.

\begin{figure}[!tb]
\centering
\resizebox{0.9555\columnwidth}{!}{
\begin{tabular}{cc}
\includegraphics[width=6cm]{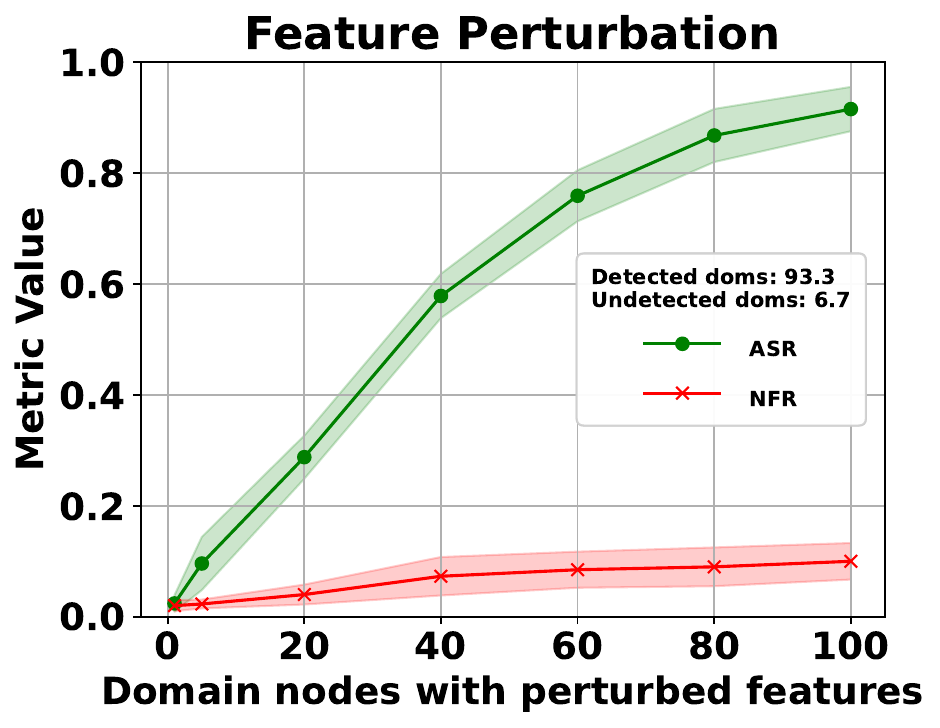}&
\includegraphics[width=6cm]{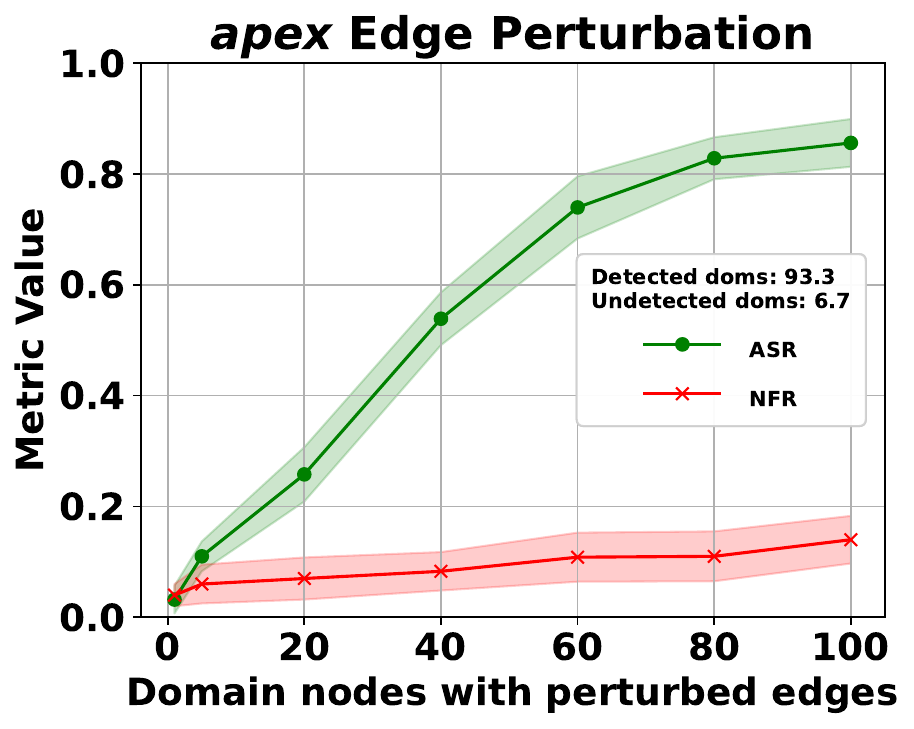}
\\
\Large{(a)}&
\Large{(b)}
\end{tabular}}
\caption{Feature and \textit{apex} edge performance with white-box model access in (a) and (b), receptively.}
\label{ASR_NFR_ras_WB} 
\end{figure}

\par \textbf{White-box model access.} In the following experiment, we evaluate MintA in a white-box setting. We examine MintA with feature and \textit{apex} edge perturbation attacks assuming access to the target MDD model parameters. Fig. \ref{ASR_NFR_ras_WB}(a) shows the performance with feature perturbation while Fig. \ref{ASR_NFR_ras_WB}(b) shows the performance with apex-edge perturbation. According to this figure, it is seen that white-box access to the model improves both attack metrics (ASR and NFR). However, the attack is still not perfect since the validity of the linearized model assumed in the formulations (and perturbation optimizations) is violated.

\par \textbf{Different surrogates.} To quantify the impact of the complexity and architecture of a surrogate model on MintA's performance, we repeat the feature perturbation attack (with a \textit{sampled adversary subgraph}) with two surrogates; a homogeneous graph attention network (GAT) \cite{velickovic2017graph} and a heterogeneous GAT (hetGAT) that uses edge-aware message passing with GAT and assumes the same edge types used by the target model. Fig. \ref{diff_surrogates} shows that both surrogates are inferior to the case of using a linearized GCN. This is because the validity of the linearized surrogate model becomes less accurate with using more sophisticated surrogates. However, the performance with the hetGAT surrogate is slightly better than that of the simple GAT one. 

\begin{figure}[htb]
\centering
\resizebox{0.9555\columnwidth}{!}{
\begin{tabular}{cc}
\includegraphics[width=6cm]{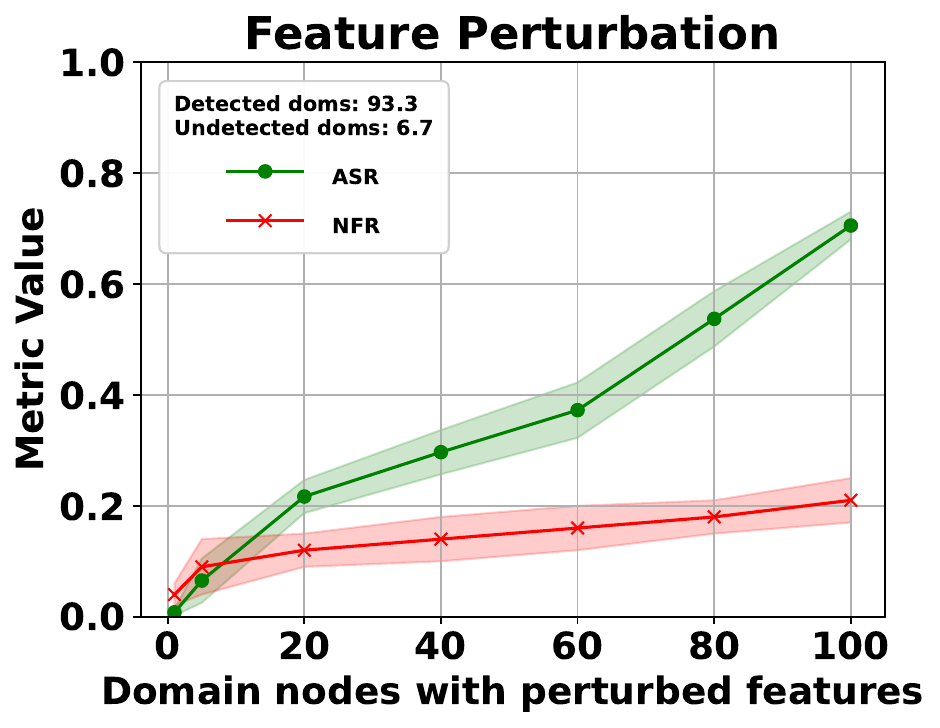}&
\includegraphics[width=6cm]{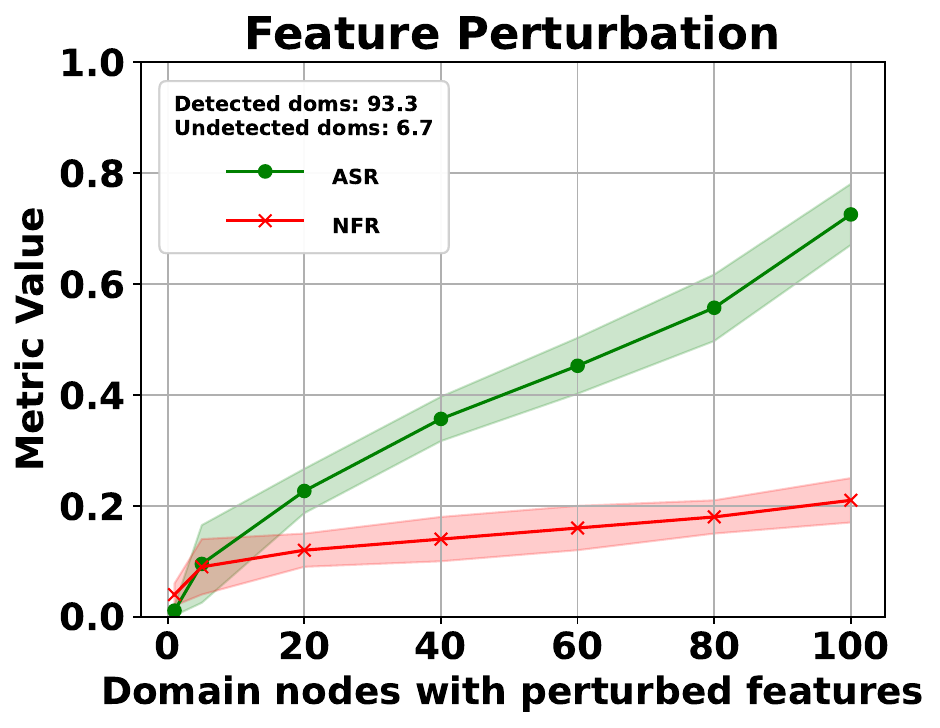}
\\
\Large{(a)}&
\Large{(b)}
\end{tabular}}
\caption{Average ASR and NFR with a GAT surrogate and a hetGAT surrogate in (a) and (b), respectively.}
\label{diff_surrogates} 
\end{figure}

\newpage 

\section{Meta-Review}

\subsection{Summary}
\par This paper proposes an adversarial evasion attack on graph-neural-network-based malicious domain detectors (MDD). The authors optimize node/edge perturbations on a surrogate MDD based on the subgraph of adversary nodes, and evaluate the effects of perturbations on the target MDD.

\subsection{Scientific Contributions}
\begin{itemize}
\item Identifies an Impactful Vulnerability
\item Provides a Valuable Step Forward in an Established Field
\item Independent Confirmation of Important Results with Limited Prior Research
\end{itemize}

\subsection{Reasons for Acceptance}
\begin{enumerate}
\item The paper demonstrates an impactful vulnerability in GNN-based MDDs. The authors propose a multi-instance adversarial attack that requires black-box access to the target MDD, and consider both outlier detection and graph purification as defense mechanisms.
\item The paper provides a valuable step forward in an established field, demonstrating that the effort to evade the detection at a certain adversary node can contradict the evasion of other adversary nodes. The authors formulate the evasion as a two-objective optimization problem that jointly maximizes the model's loss at each adversary node and its neighbors.
\item The paper designs and performs experiments under both synthetic and real-world settings to evaluate the efficacy of the proposed attack, demonstrating its feasibility.
\end{enumerate}

\subsection{Noteworthy Concerns} 
\begin{enumerate} 
\item The methodology for the use of sampled and created adversary models lacks clarity, which is exacerbated by imprecise language and complicated evaluation setup. For example, the target malicious domain detection (MDD) model demonstrates inadequate performance in "created adversary" settings, where the detection rate for malicious domains is ~30\%, suggesting that the adversary has already achieved significant success with over 70\% of their domains being misclassified as benign, without even conducting any form of attack. Furthermore, the ''created adversary'' experiments appear to conflate the ground truth with own labeling. Reviewers believe that a discussion on threats to the validity of the experimental results could be beneficial, which can clarify the methodological choices of the evaluation.
\item Threat model diverges from prior work (i.e., [30]) without motivation. Specifically, the assumptions to preserve the graph structure and feature statistics, which ensure that perturbations are unnoticeable, are violated. It is thus unclear whether IG adversarial attack can serve as an appropriate baseline given the different assumptions.
\item The paper does not discuss the relationship between the quality of the surrogate and attack success rate, e.g., how many queries the adversary should use to train their MDD.
\end{enumerate}

\section{Response to the Meta-Review} 

\par Below is a summary of our responses to the remaining concerns in the meta-review.

\begin{itemize}
\item Concern 1: During our experiments, we aim to model the adversary's intervention in the DMGs using what we refer to as the ``adversary's subgraph''. However, since actual adversary domain nodes are unavailable, we resort to two different modeling approaches. In the first approach which we call the \textit{sampled adversary modeling}, we exclusively sample connected nodes from the ground-truth malicious domains in the test set. This allows us to simulate the behavior of an adversary without actual adversary domain data. In the second approach, which we call the \textit{created adversary modeling}, we manually register domains to simulate the creation of adversary domains. We carefully adhere to practical constraints while creating these domains to provide a realistic representation of an adversary's behavior. Further details on these modeling approaches, including the motivations behind their usage and their limitations, can be found in the last two paragraphs of Subsection \ref{Section5}.1 (\textit{The setup and dataset}). Additionally, in Section 7 (\textit{Discussion}), we address the limitations of our study, including the absence of datasets containing domains of actual adversaries. These limitations are explained in detail in the second paragraph of that section.

\item Concern 2: MintA maintains its stealthiness by manipulating only a small fraction, specifically 5 out of 45 elements, in the feature vector while preserving the overall statistics of the node features. In the experiment showcased in Subsection \ref{Section5}.4 (\textit{The performance with outlier detection}), MintA demonstrates its ability to evade outlier detection, as it does not manipulate statistical features, which are crucial for isolation forest-based detection methods. In comparison to existing single-instance attack methods like IG-ADV \cite{wu2019adversarial} and PGD-ADV \cite{xu2019topology}, which excel at attacking individual nodes but struggle against multiple connected nodes, MintA is specifically designed for multi-instance attacks. It outperforms these baselines when applied to connected nodes, as illustrated in Fig. 16 and Fig. 17. Further details explaining the ineffectiveness of IG-ADV with multiple-connected nodes are provided in Footnote 5. The rationale behind selecting IG-ADV and PGD-ADV as the closest baselines for comparison is discussed in the first paragraph of Subsection 5.6 (\textit{Comparison with targeted adversarial attacks}).

\item Concern 3: To facilitate ease of use and enable the analysis of input perturbations, we employ a simplified linearized 2-hop GCN as a surrogate model. During the surrogate model training, we query the target model for 600 domain nodes. It is essential to note that the training data for the surrogate is entirely distinct from the evaluation datasets to avoid any data leakage. The success of the attack using this surrogate model is an indicator of potential transferability to the target model. Although we evaluated more complex surrogate models, they were found to be inferior to the linearized GCN in terms of performance. To provide clarity on the surrogate model's quality, we include a clarification in the second paragraph of Section 3 (\textit{Threat Model}). Additionally, experiments on the two additional surrogates are presented in Fig. 20 and the last paragraph of Appendix \ref{appendixB}.
\end{itemize}

\end{document}